\documentclass[11pt,reqno,twoside]{amsart}

\usepackage[margin=1in]{geometry}

\usepackage{amsmath,amssymb,amscd,amsthm}
\usepackage{mathrsfs}
\usepackage{upgreek}
\usepackage{mathtools}
\usepackage[dvips]{graphics}
\usepackage{enumerate}
\usepackage{mdframed}
\usepackage{cite}
\usepackage{comment}
\usepackage[colorlinks,pdfstartview=FitB]{hyperref}
\hypersetup{allcolors=blue}
\usepackage{subcaption}
\usepackage{graphicx}
\usepackage{xcolor}
\usepackage{float}

\usepackage{tikz}
\usetikzlibrary{calc,
	matrix,
	fadings,
	fit,
	math,
	arrows.meta}
\usepackage{pgfplots}
\pgfplotsset{compat=1.18}

\usepackage{bbm}
\usepackage{tikz}
\usetikzlibrary{positioning}
\usetikzlibrary{shapes.geometric, positioning, backgrounds, fit}

\definecolor{mylightblue}{RGB}{200, 230, 255}
\definecolor{mylightorange}{RGB}{255, 230, 200}
\definecolor{mylightgreen}{RGB}{210, 255, 210}

\definecolor{myblue}{RGB}{0, 100, 200}
\definecolor{myred}{RGB}{200, 50, 50}
\definecolor{mydarkorange}{RGB}{180, 90, 0}
\definecolor{mydarkgreen}{RGB}{0, 120, 0}

\def\idty{\mathbbm{1}} 

\newcommand{\Z}{{\mathbb Z}}

\newcommand{\C}{{\mathbb C}}

\newcommand{\N}{{\mathbb N}}

\renewcommand{\Re}{{\mathrm{Re} \,}}

\newtheorem{theorem}{Theorem}[section]
\newtheorem{lemma}[theorem]{Lemma}

\newtheorem{remark}[theorem]{Remark}

\theoremstyle{definition}

\allowdisplaybreaks \numberwithin{equation}{section}

\DeclareMathOperator{\diag}{diag}

\makeatletter
\def\subsection{\@startsection{subsection}{2}%
	\z@{.5\linespacing\@plus.7\linespacing}{.5\linespacing}%
	{\normalfont\scshape\centering}}
\makeatother

\hypersetup{
	pdftitle = {Bottleneck Effect and Harmonic-Type Velocity Bounds for Periodic Quantum Walks},
	pdfauthor = {Houssam Abdul-Rahman, Thomas A. Jackson, Darren C. Ong},
	pdfsubject = { },
	pdfkeywords = {}
}

\makeatletter
\renewcommand{\section}{\@startsection
  {section}{1}{\z@}%
  {-2.5ex \@plus -1ex \@minus -.2ex}%
  {1ex \@plus .2ex}%
  {\normalfont\Large\bfseries}}
  
\renewcommand{\l@section}{%
  \@tocline{1}{3pt plus 1pt}{0pt}{2.6em}{\normalsize\bfseries}%
}

\renewcommand{\subsection}{\@startsection
  {subsection}{2}{\z@}%
  {-2ex \@plus -1ex \@minus -.2ex}%
  {0.8ex \@plus .2ex}%
  {\normalfont\large\bfseries}}
  
  \renewcommand{\subsubsection}{\@startsection
  {subsubsection}{3}{\z@}%
  {-1.5ex \@plus -1ex \@minus -.2ex}%
  {0.6ex \@plus .2ex}%
  {\normalfont\bfseries}}
\makeatother

\begin{document}

\title[Bottleneck Effects and Harmonic-Type Velocity Bounds for Periodic Quantum Walks]{Bottleneck Effects and Harmonic-Type Velocity Bounds for Periodic Quantum Walks}
		
\author[H.\ Abdul-Rahman]{Houssam Abdul-Rahman}
\address{[H.\ Abdul-Rahman] Department of Mathematical Sciences, United Arab Emirates University, AL Ain, UAE}
\email{\href{mailto:houssam.a@uaeu.ac.ae}{houssam.a@uaeu.ac.ae}}

\author[T. Jackson]{Thomas A. Jackson}
\address{[T. Jackson] Department of Mathematics and Center for Quantum Mathematics and Physics, University of California,  Davis, CA  95616-8633, USA} 
\email{\href{mailto:thomasjacksonmath@gmail.com}{thomasjacksonmath@gmail.com}}

\author[D. Ong]{Darren C. Ong}
\address{[D. Ong] Department of Mathematics, Xiamen University Malaysia, Sepang, Malaysia 
and School of Mathematical Sciences, Xiamen University, 361005 Xiamen, Fujian, China}
\email{\href{mailto:darrenong@xmu.edu.my}{darrenong@xmu.edu.my}}

\begin{abstract}

We prove explicit upper bounds on the propagation velocity of one-dimensional quantum walks with periodic coins of arbitrary period. We treat two complementary settings. First, in a perturbative regime where one transmission parameter is small, we show that the corresponding almost reflecting coin acts as a bottleneck for transport: the velocity is bounded linearly in this parameter, with an explicit leading order estimate. Second, for arbitrary nonzero transmission parameters, we prove a general a priori bound in terms of their harmonic mean, together with a refined version that detects the spatial variation of neighboring coins. Moreover, we prove a general lower bound on the velocity. These bounds apply directly to the corresponding CMV setting.
\end{abstract}

\maketitle

\markboth{H. Abdul-Rahman, T. Jackson, and D. Ong}{Bottleneck Effects and Harmonic-Type Velocity Bounds for Periodic Quantum Walks}

\allowdisplaybreaks

\tableofcontents
\section{Introduction}

In periodic quantum systems, transport is typically ballistic: wave packets propagate linearly in time \cite{damanikWhatBallisticTransport2024}. This raises a basic quantitative question: how large can the propagation velocity be in terms of the local parameters of the system? In this paper we address this question for one-dimensional quantum walks with periodic coins, proving explicit upper bounds on the maximal group velocity in terms of the local transmission parameters.

Discrete time quantum walks form a natural class of unitary dynamics on discrete structures. Originally introduced as quantum analogues of classical random walks, they incorporate quantum coherence and interference, and in periodic or homogeneous settings this leads to ballistic propagation rather than the diffusive spreading of classical random walks \cite{Kempe2003,VenegasAndraca2012}. Beyond this basic role, quantum walks have become important tools in quantum algorithms \cite{Portugal2013,Montanaro2016}, quantum simulation \cite{Childs2009,AspuruGuzik2012}, and the study of topological phases of driven systems \cite{kitagawaExploringTopologicalPhases2010,TopClass,AsboBB,Asbo}.

Quantum walks have also attracted considerable attention as models for the mathematical study of spectral and dynamical properties of unitary operators.  They provide a flexible family of unitary operators for studying localization \cite{HJ14,HJS2009,AhlbrechtVolkherScholzWerner2011,JM10,HamzaJoyeStolz2006,hamzaLyapunovExponentsUnitary2007,locQuasiPer,joye2026dynamical, Konno2010Localization, Konno2009QWRE}, transport \cite{ARJS26,ARCSW25,ARS2023-CMP,CedzichFillmanVelazquez2026,Suzuki2016AsymptoticVelocity}, recurrence \cite{recurrence,bourgainQuantumRecurrenceSubspace2014}, and spectral type on discrete spaces
\cite{BHJ2003,KumarSabri2026,CedzichFillmanVelazquez2026, MaedaSuzukiWada2022}. In one dimension, an especially important feature is the connection between coined quantum walks and CMV matrices, or more generally extended CMV matrices. This links quantum walk dynamics with the spectral theory of orthogonal polynomials on the unit circle and with the unitary analogue of Jacobi matrix methods \cite{CGMV2012QIP,canteroFivediagonalMatricesZeros2003,Simon2005OPUC2,canteroMatrixvaluedSzegoPolynomials2010,CFLOZ,Simon2007CMV}.

The spectral and dynamical behavior of a quantum walk depends strongly on the structure of the underlying coin sequence. Periodic quantum walks are associated with absolutely continuous spectrum and ballistic transport \cite{ARS2023-CMP,ARCSW25,damanikQuantumDynamicsPeriodic2015,ewalks}. By contrast, disordered quantum walks often exhibit pure point spectrum and dynamical localization \cite{HamzaJoyeStolz2006,hamzaLyapunovExponentsUnitary2007,HJS2009,HJ14,AhlbrechtVolkherScholzWerner2011,JM10,AschJoye2019}. Quasi-periodic quantum walks occupy an intermediate regime, where the dynamics may display localization, anomalous transport, or other nonballistic behavior \cite{Wojcik,Nguyen2019,locQuasiPer,WangDamanik2019,CFLOZ,CFO1}.

In this work we study one-dimensional discrete time quantum walks with periodic coins, which we refer to as \emph{periodic quantum walks}. For a local coin, we call the modulus of its $(1,1)$-entry the \emph{transmission parameter}. This terminology is motivated by the homogeneous shift-coin walk, where the propagation velocity is precisely given by this quantity, see, e.g., \cite{ARS2023-CMP,ARCSW25}.

While periodicity generally leads to ballistic transport, the propagation velocity may vary substantially with the local transmission parameters.  Recent work on periodic quantum walks has considered periodicity arising from external fields, e.g., electric quantum walks, where rational fields lead to an (exponential) suppression of the velocity \cite{ewalks,ARS2023-CMP,ARCSW25}. Related work in the physics literature has also examined quantum walks with time-periodic coins of period two \cite{KumarEtal2018}. In contrast, the present work treats position-dependent periodic coins of arbitrary period and gives explicit bounds on the maximal group velocity in terms of the full transmission profile.

The transmission parameters determine the basic transport regime. If one of them vanishes, then the corresponding coin is a \emph{perfect reflector}. In this case the walk decomposes into finite-size blocks and the velocity is zero, see, e.g., \cite{ARJS26}. In contrast, when all transmission parameters are non-zero, such periodic walks are known to have positive velocity, that is, they exhibit ballistic transport. The aim of this paper is to quantify this velocity in terms of the transmission parameters of the periodic coin sequence, including the transition regime between these two cases. One observation is that, in the non-homogeneous periodic setting, the velocity depends not only on the individual coins but also on their arrangement within a single period.

The main difficulty is that the velocity of a $p$-periodic quantum walk is encoded in the derivatives of the dispersion relations of a $2p\times 2p$ Floquet matrix. These dispersion relations are generally not analytically accessible when $p>2$. We therefore avoid relying on explicit eigenvalue formulas. Instead, we use perturbation theory, Hellmann--Feynman identities, and structural bounds on Floquet eigenvectors. This provides a way to control the derivatives of the eigenvalues directly from structural properties of the Floquet matrix.

Our first main result, Theorem \ref{thm:perturbative}, quantifies a \emph{bottleneck effect} for periodic quantum walks: a single very weak transmission parameter can determine the velocity scale of the entire periodic system. The terminology is meant in the broad transport sense of a local obstruction that limits the overall propagation, a terminology that also appears in the physics literature on quantum and nanoscale transport, see, e.g.,  \cite{Bottleneck2010,Bottleneck2021}. More precisely, we consider the regime in which the smallest transmission parameter $\varrho$ tends to zero. The corresponding coin then becomes an almost perfect reflector, separating the period into nearly decoupled regions. We show that the velocity is of the same order as $\varrho$ and obtain an explicit leading order estimate. In the special case where all remaining coins are perfect transmitters, i.e., coins with transmission parameters 1, this leading-order picture becomes exact. As shown in Theorem \ref{thm:1-defect}, a non-perturbative argument gives velocity precisely equal to $\varrho$.

The second main result, Theorem \ref{thm:harmonic} gives a general a priori upper bound for the velocity of periodic quantum walks with nonzero transmission parameters. This estimate does not single out a weakest coin. Instead, it bounds the velocity in terms of quantities averaged over the full period. The first such quantity is the harmonic mean of the transmission parameters, while the refined bound also detects variations between neighboring coins \footnote{ It is particularly interesting that the harmonic mean appears as a velocity bound, as this mean is well known to arise in effective averaging of local speeds and in series resistance laws.}. Thus the bounds are explicit, do not require solving the dispersion relations, and remain sensitive to the spatial structure of the periodic medium. These results apply directly to the generalized extended CMV matrices with periodic Verblunsky coefficients.

Some general upper bounds of the velocity are already known. Besides the trivial global velocity bound of $1$, \cite{ARJS26} shows that velocity of a general inhomogeneous quantum walk is bounded above by the limit superior of its transmission parameters. For a $p$-periodic walk, this bound is the maximum of the transmission parameters in one period.
To the best of our knowledge, these are the first explicit velocity bounds for periodic quantum walks with arbitrary periods in terms of the full transmission profile, obtained without solving the Floquet dispersion relations.

Finally, to complete the picture, we include a general lower bound for the velocity in Appendix \ref{sec:Lower-bound}, see Theorem \ref{thm:v-lower}. Such a result cannot be expected to have the same scale as the harmonic-type quantities appearing in the upper bounds, since phase effects may strongly suppress transport while leaving the transmission moduli unchanged. This is already visible in the periodic-field examples of \cite{ARCSW25,ARS2023-CMP}. These works consider a homogeneous shift-coin quantum walk with constant coin transmission parameter $|a|$, subject to a specific class of external fields of period $p$. In this particular setting, they prove that the resulting walk has velocity exactly $|a|^p$ when $p$ is odd. After absorbing the field phases into the coins, these examples become $p$-periodic coin models. Thus an exponential dependence on the period is the natural scale for lower bounds depending only on the transmission parameters. Theorem \ref{thm:v-lower} provides such a universal phase-independent lower bound, and Remark \ref{rem:Lower-bound} explains how it captures the scale of the periodic-field examples.

The paper is organized as follows. In Section \ref{sec:setup-results} we introduce the model and state the main results. 
Section \ref{sec:block-velocity-Fourier} develops the block velocity and Floquet transform framework, and includes a short numerical illustration of the bounds for period $5$. In Section \ref{sec:perturbation} we prove the perturbative bottleneck result. Section \ref{sec:harmonic} establishes the general a priori bound and its refined harmonic version. Finally, Section \ref{sec:simple-spec} proves the simplicity results for the Floquet spectrum used in the perturbative and dispersion relation arguments. 

The appendices contain complementary results which complete the picture. Appendix \ref{sec:warm-up} presents the explicit velocity calculation in the two-periodic case. The computation is somewhat technical, so it is deferred to the appendix in order to preserve the flow of the main text. Appendix \ref{sec:one-nontrivial-coin} treats the periodic case in which all coins are perfectly transmitting except for one nontrivial coin, while Appendix \ref{sec:Lower-bound} gives a general lower bound for the velocity.

\section*{Acknowledgments}
The authors are grateful to Mostafa Sabri  and G\"unter Stolz for fruitful discussions.\\
D.O. gratefully acknowledges the hospitality of the Department of Mathematical Sciences at the UAE University during a short visit, where part of this work was carried out. \\
D.O. was partially supported by Xiamen University Malaysia Research Fund with fund number XMUMRF/2023C11/IMAT/0024.


\section{The model and main results}\label{sec:setup-results}
\subsection{The shift-coin setup}
We consider discrete time quantum walks on the Hilbert space
\begin{equation}
  \mathcal H = \ell^2(\Z)\otimes \C^2 .
\end{equation}
Let $\{|j\rangle : j\in\Z\}$ and
$|+\rangle =\begin{bmatrix}1\\0\end{bmatrix},\  |-\rangle =\begin{bmatrix}0\\1\end{bmatrix}$ be the standard orthonormal bases of $\ell^2(\Z)$ and $\C^2$, respectively. The corresponding orthonormal basis of $\mathcal H$ is denoted by $\delta_j^\pm:=|j\rangle \otimes |\pm\rangle$.

We study the shift-coin quantum walk with periodic coins on $\mathcal H$. It is defined as the product of the following unitary operators.
\begin{enumerate}[(i)]
\item The coin operator is periodic of period $p\in\N$. It is determined by a $p$-periodic string as follows. Let
\begin{equation}\label{def:c-p}
c:=((a_1,b_1),(a_2,b_2),\ldots,(a_p,b_p)),
\qquad (a_k,b_k)\in\mathbb S^3,\quad k=1,\ldots,p.
\end{equation}
where $\mathbb S^3:=\{(z_1,z_2)\in\C^2:\ |z_1|^2+|z_2|^2=1\}$.
In this work we assume that $c$ is a primitive $p$-periodic string, i.e., $c$ is not obtained by periodically repeating a shorter block. 

Then the coin operator $C(c)$ on $\mathcal H$ acts locally on one period by  \footnote{The matrices in \eqref{def:C} parametrize all of $\mathrm{SU}(2)$. Indeed, every $A\in \mathrm{U}(2)$ can be written as
\[
A=e^{i\phi}
\begin{bmatrix}
a & b\\
-\overline{b} & \overline{a}
\end{bmatrix},
\qquad \phi\in\mathbb R,\quad (a,b)\in\mathbb S^3.
\]
Since the velocity defined in \eqref{def:v} is unaffected by the trivial phase $e^{i\phi}$, all our results remain valid for periodic $\mathrm{U}(2)$-valued coins as well.}
\begin{equation}\label{def:C}
  C_k = \begin{bmatrix}
    a_k & b_k\\
   -\overline{b_k} & \overline{a_k}
  \end{bmatrix}, \quad k=1,\ldots,p.
\end{equation}
The operator is then extended to all $k\in \Z$ by periodicity, $C_{k+p}=C_k$. 

\item The shift operator $S$ on $\mathcal H$ is defined by
\begin{equation}\label{def:S}
  S = T \otimes P_+ + T^{-1} \otimes P_-,
\end{equation}
where $T^{\pm1}$ denotes the translation operator on $\ell^2(\Z)$,
$T^{\pm 1}|j\rangle = |j\pm 1\rangle$,  $j\in\Z$,
and $P_\pm=|\pm\rangle\langle\pm|$ are the orthogonal projections onto $|\pm\rangle$.
\end{enumerate}
Given the $p$-periodic string $c$ as in \eqref{def:c-p}, the $p$-periodic shift-coin walk is given by
\begin{equation}\label{def:U}
  U_p(c) := S C(c).
\end{equation}
We note that the case $p=1$ corresponds to the well understood homogeneous shift-coin walk. In the following we assume that $p>1$.

For a given periodic string $c$ as in \eqref{def:c-p}.
We are interested in finding the maximal velocity of the $p$-periodic walk $U_p(c)$ in \eqref{def:U}. The maximal velocity of $U_p(c)$ is defined as, see \cite{ARCSW25,ARJS26,ARFFW, ARS2023-CMP,ADFS2024},
\begin{equation}\label{def:v}
v(U_p(c))=\sup_{\substack{\psi\in D(Q)\\ \|\psi\|=1}}\limsup_{t\to\infty}\frac{1}{t}\left\|Q (U_p(c))^t \psi\right\|,
\end{equation}
where $Q$ is the position operator on $\mathcal H$,
\begin{equation}\label{def:Q}
Q=\sum_{j\in\Z}j\ |j\rangle\langle j|\otimes \idty_2.
\end{equation}

As is well known, and as recalled in Section \ref{sec:block-velocity-Fourier}, the dynamics of periodic quantum walks is governed by the associated Floquet matrices and dispersion relations, see, e.g., \cite{AhlbrechtEtal2011,ARS2023-CMP,ARCSW25}.
After applying the Fourier transform, the walk is represented by a family of $2p\times 2p$ unitary matrices (Floquet matrices). The velocity is then controlled by the derivatives of the corresponding eigenvalues (eigenfunctions). Since these eigenvalues are in general not available in closed form, our analysis relies on perturbative information (as in Theorem \ref{thm:perturbative}) and structural  information (as in Theorem \ref{thm:harmonic})  about the Floquet matrix rather than on explicit formulas.

When $p=1$, we obtain the homogeneous shift-coin walk with one coin determined by the pair $(a,b)\in\mathbb{S}^3$. It is well known that its velocity is equal to  $|a|$, see, e.g., \cite{ARCSW25}.

The shift-coin walk with 2-periodic coins admits an explicit analysis, see Appendix \ref{sec:warm-up} below. We show that in this case the velocity is determined by the minimum of the two transmission parameters.
\begin{lemma}\label{lem:warm-up}
Consider the $2$-periodic shift-coin walk $U_2(c)=SC(c)$ determined by
\begin{equation}
c=((a_1,b_1),(a_2,b_2)),
\quad (a_j,b_j)\in\mathbb S^3 \quad \text{ for }j=1,2.
\end{equation}
Then the velocity is 
\begin{equation}
v(U_2(c)) =
\begin{cases}
0 & \text{if } |a_1a_2|=0\\
1 & \text{if } |a_1a_2|=1\\
\min\{|a_1|,|a_2|\} & \text{if } 0<|a_1a_2|<1 \text{ and } b_1\overline{b_2}\in\mathbb R\\
<\min\{|a_1|,|a_2|\} & \text{if } 0<|a_1a_2|<1 \text{ and } b_1\overline{b_2}\notin\mathbb R.
\end{cases}
\end{equation}
\end{lemma}
The proof is direct but technical, relying on an explicit computation of the dispersion relation and a careful analysis of the resulting one-variable maximization problem. It is included in Appendix \ref{sec:warm-up}.

\begin{remark}
The real case
$0<|a_1a_2|<1$ and $b_1\overline{b_2}\in\mathbb R$,
contains the homogeneous walk in a period-two electric field. Indeed, as explained in Remark \ref{rem:Lower-bound}, the electric field can be absorbed into the coins. For $p=2$, this gives
$(a_2,b_2)=-(a_1,b_1)$.
Thus $|a_1|=|a_2|$ and $b_1\overline{b_2}=-|b_1|^2\in\mathbb R$, so the previous lemma gives
$v(U_2(c))=|a_1|$.
This agrees with the electric-field formula $v=|a_1|^{p/2}=|a_1|$ for $p=2$ (even period).

On the other hand, the non-real case,  $0<|a_1a_2|<1$ and $b_1\overline{b_2}\notin\mathbb R$, shows that other periodic phase choices can lead to smaller velocities. Thus, for fixed transmission moduli, some choices of the coin phases produce a velocity smaller than the one produced by the period-two electric field.
\end{remark}

\subsection{The perturbative bottleneck regime}
Lemma \ref{lem:warm-up} shows that, for 2-periodic shift-coin walks, the velocity is determined by the smallest of the two transmission parameters. This naturally raises the question whether, for a general $p$-periodic coin, the velocity is also controlled by the  weakest transmission parameter. In this sense, the weakest coin acts as a \emph{bottleneck} for transport. We use perturbative arguments to confirm that this bottleneck effect is valid for $p>2$ in the regime where the smallest transmission parameter goes to zero.  In particular, Theorem \ref{thm:perturbative} below shows that when all remaining coins are perfect transmitters, the exact scaling of the velocity is given by the minimal transmission parameter. In this special case, Theorem \ref{thm:1-defect} proves the stronger non-perturbative statement that the velocity is exactly equal to this minimum. Somewhat surprisingly, if the other coins are not perfect transmitters, then the velocity can be larger. Numerics in Section \ref{sec:numerics} are confirming that for general periodic quantum walks, the minimum transmission parameter is not an upper bound for the velocity, see Remark \ref{rem:not-min}.

Consider a  $p$-periodic $c$ as in \eqref{def:c-p} with one weak transmission parameter, denoted by $\varrho$, i.e.,
\begin{equation}\label{def:varrho}
0<\varrho:=\min_{1\leq k\leq p}|a_k| \ll 1,
\end{equation}
and assume that $\varrho$ is non-degenerate in the sense that it is attained at a unique site
\begin{equation}\label{unique-site}
|\{j:\ |a_j|=\varrho\}|=1,
\qquad
m:=\arg\min_{1\leq k\leq p}|a_k|.
\end{equation}
We study $v(U_p(c(\varrho)))$ in the regime when $\varrho\to 0$. Here $c=c(\varrho)$ where
\begin{equation}
a_m(\varrho)=\alpha_m\varrho,\qquad b_m(\varrho)=\beta_m\sqrt{1-\varrho^2}, \qquad |\alpha_m|=|\beta_m|=1,
\end{equation}
while all other coins are fixed.

Define $L_m$ and $R_m$ to be the number of consecutive perfect transmitters, i.e. sites with $|a_j|=1$, to the \emph{left} and  \emph{right} of the weak site $m$:
\begin{align}\label{def:RL}
R_{m}&:=\max\{r\geq 0:\ |a_{m+1}|=|a_{m+2}|=\cdots=|a_{m+r}|=1\},\notag\\
L_{m}&:=\max\{\ell\geq 0:\ |a_{m-1}|=|a_{m-2}|=\cdots=|a_{m-\ell}|=1\},
\end{align}
where the indices are understood cyclically.

\begin{theorem}\label{thm:perturbative} 
For a given positive integer $p$, consider the $p$-periodic shift-coin walk $U_p(c(\varrho))=SC(c(\varrho))$, where $\varrho$ is the weakest transmission parameter, see \eqref{def:varrho}. Then
\begin{equation}\label{eq:main:bound}
\lim_{\varrho\to 0} \frac{v\left(U_p\left(c(\varrho)\right)\right)}{\varrho}
\begin{cases}
=1,
& \text{if } |a_j|=1 \text{ for all } j\neq m,\\[0.2cm]
\displaystyle \leq \frac{p}{2\sqrt{(L_m+1)(R_m+1)}},
& \text{otherwise}.
\end{cases}
\end{equation}
Here $R_m$ and $L_m$ are the numbers of perfect transmitters adjacent to the weak site $m$, as defined in \eqref{def:RL}.
\end{theorem}
The proof of Theorem \ref{thm:perturbative} is perturbative and is given mainly in Section \ref{sec:perturbation}. We regard $U_p(c(\varrho))$ as a perturbation of the decoupled walk $U_p(c(0))$. The first step is to understand the Floquet spectrum at $\varrho=0$. Under the non-degeneracy
assumption that the weakest coin is attained at a unique site, the corresponding Floquet matrix has simple spectrum, see
Theorem \ref{thm:simple-spec}(\ref{thm:spectrum:2}). This explains the uniqueness assumption in \eqref{unique-site}. When the minimum is attained at several sites, the same argument applies provided the associated decoupled Floquet matrix has simple spectrum. This case is addressed in Theorem \ref{thm:spectrum:3}, see also Remark \ref{rem:multiple-reflectors} that provides the bound \eqref{rem:MR-bound} in this several sites setting.

Once simplicity is available, the eigenvalues and eigenvectors can be perturbed from $\varrho=0$, and the velocity estimate reduces to controlling the first-order behavior of the Floquet eigenvectors and the resulting derivatives of the dispersion relations.

Here are some remarks.
\begin{enumerate}[(A)]\itemsep0.2cm
\item Theorem \ref{thm:perturbative} quantifies the bottleneck effect created by the
smallest transmission parameter. As $\varrho\to0$, the coin at the site $m$
approaches a perfect reflector, and the velocity goes to zero at least linearly in
$\varrho$. The leading coefficient depends on the local arrangement of perfect transmitter ($|a_j|=1$) coins
around the weak site.

The role of the adjacent perfect transmitters is particularly transparent. If the weak
coin $m$ is surrounded by $L_m$ consecutive perfect transmitters on the left and
$R_m$ consecutive perfect transmitters on the right, then the leading
coefficient is bounded as in \eqref{eq:main:bound}.
Thus, for fixed period $p$, the upper bound improves as the weak coin is placed
between longer chains of perfect transmitters. 

\item In the extremal case where all coins except the weakest one are perfect transmitters, the perturbative proof  shows that
\begin{equation}
v(U_p(c(\varrho)))=\varrho+\mathcal{O}(\varrho^2).
\end{equation}
In fact, this illustrates that the bottleneck effect is exact beyond the perturbative regime: if the $p$-periodic string $c$ is given as
\begin{equation}
c=\left((a,b),(1,0),\ldots,(1,0)\right),
\end{equation}
then $v(U_p(c))=|a|$, 
see Theorem~\ref{thm:1-defect}.

\item The distinction between the two cases in \eqref{eq:main:bound} may seem counterintuitive at first. One might expect the presence of perfect transmitters to increase the velocity, since transport through such sites is locally maximal. However, the case in which all coins except the weakest one are perfect transmitters has a special structure. In that situation, the motion between two consecutive weakest sites is essentially deterministic, and the perfect transmitting regions enhance transport symmetrically to the left and to the right of the bottleneck. Since the weakest coin is repeated periodically, the whole periodic chain behaves like a homogeneous walk on the line whose transmission parameter is precisely $\varrho$. 
When some of the coins are not perfect transmitters, this homogeneous picture breaks down. The additional scattering within one period may then produce a larger leading-order prefactor in the small-$\varrho$ scaling of the velocity.

\item This estimate also exhibits a quantitative trade-off between the period $p$ and the size of the minimal transmission parameter. The bound scales linearly in $p$,
whereas the bottleneck-effect regime is governed by the small parameter $\varrho=\min_k |a_k|$. Consequently, to obtain a comparable bottleneck effect for larger periods, one must take the minimal transmission parameter $\varrho$ sufficiently smaller.

\item Here is an explicit example:
Let $p=10$, and consider the periodic string of transmission parameters
\begin{equation}
(|a_1|,\ldots,|a_{10}|)
=
(\varrho,1,1,1,0.6,0.8,0.7,1,1,1),
\qquad 0<\varrho\ll0.6.
\end{equation}
The unique weakest coefficient is $|a_1|=\varrho$. To the right of $a_1$, there are
three consecutive perfect transmitters: $
|a_2|=|a_3|=|a_4|=1$, i.e., $R_1=3$.
To the left of $a_1$, using the cyclic convention, we have
$|a_{10}|=|a_9|=|a_8|=1$, i.e., $L_1=3$.
Therefore Theorem \ref{thm:perturbative} gives
\begin{equation}
\lim_{\varrho\to0}\frac{v(U_{10}(c(\varrho)))}{\varrho}\leq\frac{10}{2\sqrt{(3+1)(3+1)}}=\frac{10}{8}=1.25.
\end{equation}
Thus, for small $\varrho$, the velocity is at most of order $1.25\,\varrho$, up to
higher-order corrections.

For comparison, if the weak coin at site 1 were not adjacent to any perfect transmitters, so that
$L_1=R_1=0$, the same theorem would only give
\begin{equation}
\lim_{\varrho\to0}\frac{v(U_{10}(c(\varrho)))}{\varrho}\leq\frac{10}{2}=5.
\end{equation}
This illustrates that the bottleneck effect becomes stronger when the weakest coin is
surrounded by longer strings of perfect transmitters.
\end{enumerate}

Theorem \ref{thm:perturbative} above is perturbative. It describes the bottleneck-effect  regime in which one transmission parameter is much smaller than all the others, and the velocity is studied through the small-$\varrho$ behavior of the dispersion relations. In particular, this result relies on singling out a unique weakest coin and on analyzing the corresponding decoupled limit. Next, we discuss a general a priori bound.

\subsection{General harmonic-type a priori bounds}

We now turn to a different type of bound, valid for $p$-periodic coins  with nonzero transmission parameters. This bound does not require a unique weakest coin. Instead, it bounds the velocity in terms of quantities averaged over the full period. The first such quantity is the harmonic mean of the transmission parameters, while the second is a refinement that also detects the spatial variation of neighboring coins. The resulting estimates are especially effective in the near-homogeneous regime, where the transmission parameters are close to one another.

For a $p$-periodic string $c$ in \eqref{def:c-p} with $a_j\neq0$ for all $j$,
we define its harmonic mean by
\begin{equation}\label{def:Hc}
H(c):=p\left(\sum_{j=1}^p\frac{1}{|a_j|}\right)^{-1}.
\end{equation}
We also define the refined harmonic mean by
\begin{equation}\label{def:Hc-tilde}
\widetilde H(c):=p\left(\sum_{j=1}^p\frac{1}{|a_j|}+\frac{1}{4 \mu_a}\sum_{j=1}^p\left|t_{j+1}-t_j\right|^2\right)^{-1}
\end{equation}
with
\begin{equation}
\mu_a:=\max_{1\leq k\leq p}\frac{|a_k|^2}{1+|a_k|}, \qquad t_j:=\frac{|b_j|}{1+|a_j|}=\sqrt{\frac{1-|a_j|}{1+|a_j|}}\in[0,1),
\end{equation}
where the indices are understood cyclically.

Both $H(c)$ and $\widetilde H(c)$ are independent of the period chosen to represent
the periodic sequence. Indeed, suppose that the displayed period $p$ is not
fundamental, and let $q$ be the fundamental period. Then $p=mq$ for some
$m\in\mathbb N$. Since the quantities $|a_j|^{-1}$ and
$|t_{j+1}-t_j|^2$ are $q$-periodic, we have
\begin{equation}
\sum_{j=1}^{p}\frac{1}{|a_j|}=m\sum_{j=1}^{q}\frac{1}{|a_j|},\qquad\sum_{j=1}^{p}|t_{j+1}-t_j|^2=m\sum_{j=1}^{q}|t_{j+1}-t_j|^2.
\end{equation}
Hence $H(c)$ and $\widetilde H(c)$ are unchanged whether $p$ or $q$ is used as
the period.
\begin{remark}\label{simpler12}
Since $0<|a_k|\leq1$, we have $\mu_a\leq 1/2$. Thus, in the definition of $\widetilde H(c)$, one may replace $1/(4 \mu_a)$ by $1/2$. This gives the simpler but slightly weaker bound
\begin{equation}
\widetilde H(c)\leq p\left(\sum_{j=1}^p\frac{1}{|a_j|}+\frac12\sum_{j=1}^p |t_{j+1}-t_j|^2\right)^{-1}.
\end{equation}
\end{remark}
\begin{theorem}\label{thm:harmonic}
For a given positive integer $p$, consider the $p$-periodic shift-coin walk
$U_p(c)=SC(c)$ determined by the string $c$ as in \eqref{def:c-p} where $a_j\neq 0$ for all $j=1,\ldots,p$. Then
\begin{equation}\label{eq:hatmonic:main}
v(U_p(c))\leq \widetilde H(c)\leq H(c),
\end{equation}
where $H(c)$ and $\widetilde H(c)$ are the harmonic mean and the refined
harmonic mean given in \eqref{def:Hc} and \eqref{def:Hc-tilde}, respectively.
\end{theorem}
The proof is included mainly in Section \ref{sec:harmonic}. After using the simplicity of
the Floquet spectrum from Theorem \ref{thm:simple-spec}, a Hellmann--Feynman
formula reduces the problem to estimating a site-independent quantity associated
with the Floquet eigenvector.

Here are some remarks.
\begin{enumerate}[(A)]\itemsep0.2cm
\item $H(c)$ depends only on the unordered collection of
transmission parameters $|a_1|,\ldots,|a_p|$.
In particular, the harmonic mean bound is invariant under all permutations of the period. By contrast,
the correction term $\sum |t_{j+1}-t_j|^2$ in $\widetilde H(c)$
depends on the relative placement of neighbouring coins. Thus $\widetilde H(c)$
detects the spatial arrangement of the coin sequence.
At the same time, the correction term is invariant under cyclic translations of the period:
\begin{equation}
(t_1,\ldots,t_p)\mapsto (t_{k+1},\ldots,t_p,t_1,\ldots,t_k).
\end{equation}
This is consistent with the regrouping invariance of the velocity discussed in
Remark \ref{rem:regrouping}. Hence $\widetilde H(c)$ gives a bound with the correct
translation symmetry: it is independent of the choice of the starting point of the
period, but it is sensitive to the order in which different coins appear.


\item Since the correction term in $\widetilde H(c)$ is nonnegative, we immediately have $\widetilde H(c)\leq H(c)$.
Equality holds if and only if $|a_1|=\cdots=|a_p|$.
Thus the refined estimate improves the harmonic bound precisely when the transmission
coefficients are not constant along the period.

\item The bounds in Theorem \ref{thm:harmonic} depend only on the moduli of the transmission parameters. Consequently, they do not capture velocity suppression caused purely by periodic phases. This is relevant for electric quantum walks, where the external field can often be absorbed into position-dependent phases of the coins, and where such phases may lead to suppression of the propagation velocity \cite{ewalks,ARS2023-CMP,ARCSW25}. Thus the harmonic-type bounds should be understood as controlling the contribution of the spatial transmission profile, rather than phase-induced transport effects. See Remark \ref{rem:Lower-bound}  for more details on the electric-field setting and on how the general lower bound, given in Theorem \ref{thm:v-lower}, still captures the exponential suppression of the velocity.

\item \label{passing-the-limit}
The velocity bounds in Theorem \ref{thm:harmonic} apply for periodic settings. One might
hope to approximate an inhomogeneous shift-coin walk $U$ by periodic walks
$U_p$, obtained by periodizing longer and longer finite blocks of the original
coin sequence, and then pass to the limit
\begin{equation}
v(U)\overset{?}{=}\lim_{p\to\infty}v(U_p).
\end{equation}
This conclusion is false in general, even if $U_p\to U$ strongly and even if the
periodic approximants use only coins appearing in the original walk.
Indeed, fix $0<\delta<1$, and consider the shift-coin walk $U=SC$ with
\begin{equation}
a_j=\begin{cases}
1, & j\leq 0,\\ 
\delta, & j\geq 1,
\end{cases}\qquad b_j=\sqrt{1-a_j^2}.
\end{equation}
Then all $a_j\neq 0$. Moreover, the left half-line is a perfect transmitter. For
example, 
\begin{equation}
U^t\delta_0^-=\delta_{-t}^-,\qquad \text{ for all }t\in\mathbb{N}.
\end{equation} 
This  shows that $v(U)=1$.

Now let $U_N:=SC_N$ be the periodic walk obtained by periodizing the finite block with transmitting parameters $a_j$'s as
\begin{equation}
\underbrace{1,\ldots,1}_{N\ \mathrm{times}},
\underbrace{\delta,\ldots,\delta}_{N\ \mathrm{times}}.
\end{equation}
The fundamental period of this block is $2N$. Also $U_N\to U$ strongly  \footnote{ Indeed, for
every finitely supported vector $\varphi$ (dense in $\mathcal{H}$), the coefficients of $C_N$ agree with those
of $C$ on the support of $\varphi$ for all sufficiently large $N$. Since the
operators $U_N$ and $U$ are unitary (hence bounded), this implies strong convergence on
$\mathcal H$.}.
However, the velocities $v(U_N)$ do not converge to $v(U)$. Applying the harmonic
bound to the fundamental period $2N$, we obtain
\begin{equation}
v(U_N)
\leq
2N\left(
N+\frac{N}{\delta}
\right)^{-1}
=
\frac{2\delta}{1+\delta} < v(U) \qquad\text{(because $\delta<1$)}.
\end{equation}
Therefore
\begin{equation}
\limsup_{N\to\infty}v(U_N)
<
v(U).
\end{equation}

\item  The example  in (\ref{passing-the-limit}) above does not settle the genuinely limit-periodic case, where the coin sequence is a uniform limit of periodic ones, see e.g., \cite{Ong2012LimitPeriodic,FillmanOng2017,fillmanLimitperiodic2017,DamanikFillman2020}. It only shows that strong convergence of periodic approximants is too weak for passing velocity bounds to the limit. Whether harmonic-type bounds such as those in Theorem \ref{thm:harmonic} extend to limit-periodic quantum walks remains open. 

However, we should note that related results suggest that the spreading rate of a limit periodic operator might not be similar to the spreading rates of its periodic approximants. For example, \cite{FillmanOng2017} shows that limit-periodic CMV operators generically have zero measure spectrum, even though their periodic approximants have absolutely continuous spectrum. If two operators are close in norm, this can imply that their eigenvalues are close, however we define velocity using the $\theta$-derivative of the eigenvalues rather than the location of the eigenvalues. So even if an operator is a norm limit of periodic operators, this will not imply that the velocity bound will carry over.

\end{enumerate}

\subsection{Periodic split-step walks}
We now turn to split-step walks with periodic coins. Given two sequences
\begin{equation}
c=\bigl((a_j,b_j)\bigr)_{j\in\mathbb Z},
\quad
\tilde c=\bigl((\tilde a_j,\tilde b_j)\bigr)_{j\in\mathbb Z},
\quad
(a_j, b_j), (\tilde a_j,\tilde b_j)\in \mathbb S^3. 
\end{equation}
We define the corresponding coin operators by
\begin{equation}
C(c)=\sum_{j\in\mathbb Z}|j\rangle\langle j|\otimes 
\begin{bmatrix}
a_j & b_j\\
-\overline{b_j} & \overline{a_j}
\end{bmatrix},
\qquad
C(\tilde c)=\sum_{j\in\mathbb Z}|j\rangle\langle j|\otimes \begin{bmatrix}
\tilde a_j & \tilde b_j\\
-\overline{\tilde b_j} & \overline{\tilde a_j}
\end{bmatrix}.
\end{equation}
The split-step walk is given by the unitary
\begin{equation}\label{def:split-step}
W(c,\tilde c):=S_+C(c)S_-C(\tilde c),
\end{equation}
where
\begin{equation}\label{def:half-shifts}
S_+=T\otimes P_+ + \idty_{\mathbb Z}\otimes P_-,
\qquad
S_-=\idty_{\mathbb Z}\otimes P_+ + T^{-1}\otimes P_-.
\end{equation}
\begin{remark}
If $C(c)=\idty_{\mathcal H}$ then $W=U(\tilde c)$. Moreover, if $C(\tilde c)=\idty_{\mathcal H}$ then  $W=S_+ C(c)S_-$, and  observe that in this case, see \cite[Lemma 3.2]{ARJS26}, 
\begin{equation}
v(W)=v\left(S_+ C(c)S_-\right)= v\left(S_- S_+ C(c)\right)=v(U(c)).
\end{equation}
Thus, in the following, we discuss the cases where $W(c,\tilde c)$ is  a genuine split-step walk, distinct from a shift-coin walk.
\end{remark}
Our main interest is the periodic case, where $c$ and $\tilde c$ are periodic,
possibly with different periods. We assume that $c$ has fundamental period
$p_1$ and $\tilde c$ has fundamental period $p_2$. Thus
\begin{equation}
(a_{j+p_1},b_{j+p_1})=(a_j,b_j),
\qquad
(\tilde a_{j+p_2},\tilde b_{j+p_2})
=
(\tilde a_j,\tilde b_j).
\end{equation}
In contrast with the shift-coin case, one has to keep track of the relative position of
the two coin sequences. A simultaneous translation of both $c$ and $\tilde c$
does not change the velocity, see Remark \ref{rem:regrouping}, but translating only one of them changes the relative
placement of the two half-steps and may lead to a different split-step walk.

We now explain how the preceding bounds for shift-coin walks imply bounds for
$W(c,\tilde c)$. Define the interlaced sequence $d:\Z\to\mathbb S^3$
by
\begin{equation}\label{def:interlaced-d}
(d)_{2n}:=(\tilde c)_n=(\tilde a_n,\tilde b_n),
\quad
(d)_{2n+1}:=(c)_n=(a_n,b_n),
\quad n\in\mathbb Z.
\end{equation}
Let $U(d):=SC(d)$ be the corresponding shift-coin walk.

The connection between $U(d)$ and $W(c,\tilde c)$ is given by the sieving theorem
\cite[Theorem 3.1]{ARCSW25}. The operator $U(d)$ changes the parity of the position at
each time step. Hence $(U(d))^2$ leaves the even and odd subspaces invariant:
\begin{equation}
\mathcal H=
\left[\ell^2(2\mathbb Z)\otimes\mathbb C^2\right]
\oplus
\left[\ell^2(2\mathbb Z+1)\otimes\mathbb C^2\right].
\end{equation}
Under a natural ordering of the basis \footnote{The precise ordering of the basis is part of the sieving construction in \cite[Theorem 3.1]{ARCSW25}. We omit it here because only the resulting unitary equivalence is used in our arguments.}, one obtains
\begin{equation}\label{eq:sieving-split-step}
(U(d))^2\equiv W(c,\tilde c)\oplus W(\tau\tilde c,c) \quad\text{ where }\quad(\tau\tilde c)_n:=(\tilde c)_{n+1}=(\tilde a_{n+1},\tilde b_{n+1}).
\end{equation}
We now compare the velocities. Since the identification of the even sublattice
$\ell^2(2\mathbb Z)$ with $\ell^2(\mathbb Z)$ rescales the position operator by a
factor of $2$, the velocity of $(U(d))^2$ satisfies, see \cite[Lemma 3.7]{ARCSW25},
\begin{equation}\label{eq:velocity-sieving}
v\left((U(d))^2\right)=2\max\{v(W(c,\tilde c)),v(W(\tau\tilde c,c))\}.
\end{equation}
On the other hand, see \cite[Lemma 3.8]{ARCSW25},
$v\left((U(d))^2\right)=2v(U(d))$.
Combining this with \eqref{eq:velocity-sieving}, we obtain
\begin{equation}\label{eq:split-step-max}
v(U(d))
=
\max\{v(W(c,\tilde c)),v(W(\tau\tilde c,c))\}.
\end{equation}
In particular,
\begin{equation}\label{eq:split-step-bound-by-shift}
v(W(c,\tilde c))\leq v(U(d)).
\end{equation}
Thus every upper bound for the interlaced shift-coin walk $U(d)$ gives an upper bound
for the split-step walk $W(c,\tilde c)$.

Let $L:=\mathrm{lcm}(p_1,p_2)$, then the interlaced sequence $d$ is $2L$-periodic \footnote{ However fundamental period may be a
proper divisor of $2L$, but this does not affect the velocity bound.}.
Assume that $a_j\neq0$, $\tilde a_j\neq0$, for all  $j\in\mathbb Z$.
Applying the harmonic bound for the shift-coin walk $U(d)$ and using
\eqref{eq:split-step-bound-by-shift}, we get
\begin{equation}\label{eq:split-harmonic-bound}
v(W(c,\tilde c))\leq2L \left( \sum_{j=1}^{L}\frac{1}{|a_j|} +\sum_{j=1}^{L}\frac{1}{|\tilde a_j|}\right)^{-1} = H(\{H(c), H(\tilde c)\}).
\end{equation}
i.e., the velocity of the split-step walk with two periodic coins is bounded by the harmonic mean of the two harmonic means
associated with the two coin sequences. In this sense, the velocity of the split-step walk exhibits a two-level harmonic averaging: first along each coin operator over one period, and then between the two alternating coins.

The refined harmonic bound for the shift-coin walk gives a sharper estimate. Set
\begin{equation}
t_j:=\frac{|b_j|}{1+|a_j|},
\qquad
\tilde t_j:=\frac{|\tilde b_j|}{1+|\tilde a_j|}, \text{ for }j=1,\ldots, L
\end{equation}

and the indices are defined cyclically. Using the interlaced order \eqref{def:interlaced-d}, the refined bound for $U(d)$
yields (we use the simpler bound in Remark \ref{simpler12}),
\begin{equation}\label{eq:split-refined-bound}
v(W(c,\tilde c))\leq2L\Bigg(\sum_{j=1}^{L}\frac{1}{|a_j|}+\sum_{j=1}^{L}\frac{1}{|\tilde a_j|}+\frac12\sum_{j=1}^{L}\left(|t_j-\tilde t_j|^2+|\tilde t_{j+1}-t_j|^2\right)\Bigg)^{-1}.
\end{equation}
The two terms in the last
sum correspond to the two adjacent differences in the interlaced sequence $d$.
This reflects the fact that the split-step dynamics depends not only on the two periodic
coin sequences separately, but also on their relative placement.

\subsection{The connection with CMV matrices}
CMV matrices form a natural unitary analogue of Jacobi matrices \cite{BHJ2003,  canteroFivediagonalMatricesZeros2003}. In the extended setting, they can be realized as split-step quantum walks, and this identification allows us to apply the preceding velocity bounds directly to periodic generalized extended CMV matrices.

For a \emph{Verblunsky pair}
$(\alpha,\rho)\in\mathbb S^3$, define
\begin{equation}
\Theta(\alpha,\rho):=\begin{bmatrix}\overline{\alpha} & \rho\\ \overline{\rho} & -\alpha\end{bmatrix}.
\end{equation}
Given a sequence of Verblunsky pairs $(\alpha_n,\rho_n)_{n\in\mathbb Z}$, we define the block diagonal operators $\mathcal L$ and $\mathcal M$ on
$\ell^2(\mathbb Z)$ by
\begin{align}
\mathcal L& =\mathcal L((\alpha_{2n},\rho_{2n})_{n\in\Z})=\bigoplus_{n\in\Z}\Theta(\alpha_{2n},\rho_{2n}),\\
\mathcal M &= \mathcal M ((\alpha_{2n+1},\rho_{2n+1})_{n\in\Z})=\bigoplus_{n\in\Z}\Theta(\alpha_{2n+1},\rho_{2n+1}),
\end{align}
where the block $\Theta(\alpha_j,\rho_j)$ acts on $\ell^2(\{j,j+1\})$. The associated \emph{generalized extended CMV matrix} is, see e.g., \cite{BHJ2003, CGMV2012QIP, canteroFivediagonalMatricesZeros2003,CFO1}
\begin{equation}\label{def:E}
\mathcal E= \mathcal E\left((\alpha_n,\rho_n)_{n\in\Z}\right)=\mathcal L \mathcal M.
\end{equation} 

In what follows, we are interested in the periodic case. Thus we assume that the sequence $(\alpha_n,\rho_n)_{n\in\mathbb Z}$ is $q$-periodic, that is,
\begin{equation}
(\alpha_{n+q},\rho_{n+q})=(\alpha_n,\rho_n), \qquad n\in\mathbb Z .
\end{equation}
The matrix $\mathcal E$ defined by \eqref{def:E} will be referred to as the corresponding $q$-periodic generalized extended CMV matrix.

Identify $\ell^2(\mathbb Z)$ with $\ell^2(\mathbb Z)\otimes\mathbb C^2$ via
\begin{equation}\label{eq:cmv-identification}
\delta_{2n-1}\mapsto \delta_n^+, \qquad \delta_{2n}\mapsto \delta_n^- .
\end{equation}
With this convention, \cite[Lemma A.1]{ARCSW25} gives the unitary equivalence
\begin{equation}\label{eq:cmv-split-step}
\mathcal E \equiv S_+ C(\gamma) S_- C(\tilde \gamma)=W(\gamma,\tilde \gamma), \text{ where }\gamma_n=(\overline{\rho_{2n}},-\alpha_{2n}),
 \ \tilde \gamma_n=(\overline{\rho_{2n+1}},-\alpha_{2n+1}).
\end{equation}
Thus a periodic generalized extended CMV matrix is, after the identification
\eqref{eq:cmv-identification}, precisely a periodic split-step walk whose two
coin sequences are obtained from the even and odd Verblunsky pairs.
Thus, 
\begin{equation}\label{eq:cmv-velocity-split}
v(\mathcal E)=v(W(\gamma,\tilde \gamma)).
\end{equation}
Consequently, the split-step estimates obtained above apply directly to
$\mathcal E$. 

More explicitly, if $q$ is odd, then both sequences $\gamma$ and $\tilde \gamma$
in \eqref{eq:cmv-split-step} are $q$-periodic. If $q$ is even, then both are
$(q/2)$-periodic. Of course, their fundamental periods may be proper
divisors of these numbers. We denote their fundamental periods by $p_1$ and
$p_2$, respectively, and set $L=\mathrm{lcm}(p_1,p_2)$.

In particular, assuming $\rho_j\neq0$ for all $j$, the
harmonic bound \eqref{eq:split-harmonic-bound} and the refined bound \eqref{eq:split-refined-bound} above hold with $a_j$ and $\tilde a_j$ are replaced by $\rho_{2j}$ and $\rho_{2j+1}$, respectively.

\section{Block velocity and Fourier transform}\label{sec:block-velocity-Fourier}
 In this section, we introduce the Fourier framework that underlies both main
results, Theorems \ref{thm:perturbative} and \ref{thm:harmonic}. By regrouping the original lattice into blocks of length $p$, the
$p$-periodic walk becomes translation invariant. The Fourier transform then
reduces the analysis to a family of finite-dimensional Floquet matrices, whose
dispersion relations determine the block velocity.

\subsection{The block velocity}
Regroup the lattice into blocks of length $p$ as
\begin{equation}\label{def:H-regrouped}
\mathcal H\equiv \ell^2(\Z)\otimes \C^{2p}.
\end{equation}
 The corresponding  position operator  is no longer the original position operator $Q$ in  \eqref{def:Q}, but rather the
operator which counts the block number. The purpose of this section is to
compare these two position operators and to show that they define the same velocity up to the factor $p$. This leads to the notion of
\emph{block velocity}, which is the velocity measured in the regrouped lattice and
will be the quantity computed through the dispersion relations.

We introduce the modified position operator associated with the regrouped lattice by
\begin{equation}
Q_p:=\sum_{j\in\Z} j\, |j\rangle\langle j|\otimes \idty_{2p}.
\end{equation}
Thus, $Q_p$ labels the blocks rather than the original lattice sites. Equivalently, on the
original labelling of $\ell^2(\Z)\otimes \C^2$, this operator is given by
\begin{equation}\label{def:Qp}
Q_p \equiv\sum_{k=0}^{p-1}\sum_{j\in\Z} j\, |pj+k\rangle\langle pj+k|\otimes \idty_2.
\end{equation}
In other words, $Q_p$ assigns the same position label $j$ to all sites in the $j$-th block
of $p$ consecutive original positions.

A direct computation using \eqref{def:Q} and \eqref{def:Qp} gives
\begin{equation}
Q-pQ_p=\sum_{k=0}^{p-1}\sum_{j\in\Z}k\, |pj+k\rangle\langle pj+k|\otimes \idty_2 .
\end{equation}
Thus, $Q-pQ_p$ is a bounded operator and $\|Q-pQ_p\|=p-1$. Hence, by \cite[Thm 3.1]{ARJS26}, the scaled modified position operator $p\tilde Q$ can be used instead of the $Q$ in the definition of the velocity, i.e., 
\begin{equation}
v(U_p(c))=
\sup_{\substack{\psi\in D(Q)\\ \|\psi\|=1}}
\limsup_{t\to\infty}\frac{1}{t}\left\|Q(U_p(c))^t\psi\right\|
=
\sup_{\substack{\psi\in D(Q_p)\\ \|\psi\|=1}}
\limsup_{t\to\infty}\frac{1}{t}\left\|pQ_p(U_p(c))^t\psi\right\|.
\end{equation}
We define the corresponding \emph{block velocity} by
\begin{equation}\label{def:v-p}
v^{(p)}_{\mathrm{blk}}(U_p(c)):=
\sup_{\substack{\psi\in D(Q_p)\\ \|\psi\|=1}}
\limsup_{t\to\infty}\frac{1}{t}
\left\|Q_p (U_p(c))^t \psi\right\|.
\end{equation}
Hence,
\begin{equation}\label{eq:v-vp}
v(U_p(c))=p\ v^{(p)}_{\mathrm{blk}}(U_p(c)).
\end{equation}

The factor $p$ has a simple interpretation. Consider, for instance, the trivial walk of
perfect transmitters, namely $|a_j|=1$ for all $j\in\Z$. With respect to the original
position operator $Q$, the velocity is equal to $1$. However, after regrouping $p$
consecutive sites into a single block, the operator $Q_p$ counts the block number rather
than the original position. Therefore, a particle must travel through $p$ original sites in
order to move from one regrouped site to the next. Consequently, its block velocity is
$1/p$, which is consistent with the relation
$v(U_p(c))=p\ v^{(p)}_{\mathrm{blk}}(U_p(c))$.

\begin{remark}\label{rem:regrouping}
The operator $Q_p$ in \eqref{def:Qp} corresponds to the regrouping of the original lattice sites into
blocks of length $p$ starting at site $0 \mod p$ as in Figure \ref{fig:regrouping}.
\begin{figure}[H]
\centering
\includegraphics[width=0.9\textwidth]{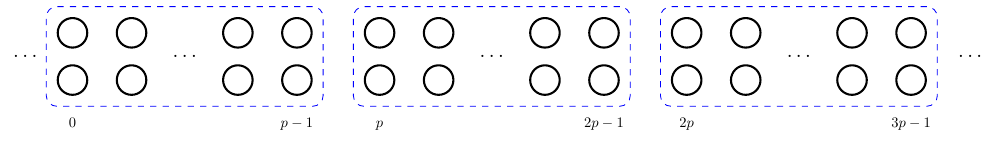}
\caption{Regrouping of the one-dimensional lattice sites into blocks of length $p$. Each site carries two internal states, represented by the upper circle for the spin basis $|+\rangle$ and the lower circle for $|-\rangle$. }
\label{fig:regrouping}
\end{figure}

Thus, $Q_p$ assigns the same label $j$ to all sites in the block
$\{pj,pj+1,\ldots,pj+p-1\}$.

We now observe that the block velocity is invariant under cyclic regrouping, i.e., the velocity is independent of the choice of the starting point of the groups. Indeed, let $T_{\mathcal H}:=T\otimes \idty_2$ denote the lattice translation and, for $k\in\Z$,
define the shifted regrouped position operator by
\begin{equation}
Q_{p,k}:=T_{\mathcal H}^k Q_p T_{\mathcal H}^{-k}.
\end{equation}
This operator corresponds to shifting the above decomposition into blocks by $k$ lattice
sites. We claim that $pQ_{p,k}$ differs from $Q$ only by a bounded operator. Indeed,
\begin{equation}
Q-pQ_{p,k}= T_{\mathcal H}^k[T_{\mathcal H}^{-k},Q]+T_{\mathcal H}^k\bigl(Q-pQ_p\bigr)T_{\mathcal H}^{-k}.
\end{equation}
Both terms are bounded, in particular,
$\|[T_{\mathcal H}^{-k},Q]\|=|k|$ and $\|Q-pQ_p\|=p-1$. Therefore, $Q-pQ_{p,k}$ is bounded for every
$k\in\Z$.

By  \cite[Thm 3.1]{ARJS26}, bounded perturbations of the position operator do not affect the
velocity. Hence the block velocity defined through the regrouped position operator is
invariant under finite lattice translations of the regrouping. In particular, it is
independent of the chosen starting site of the blocks.
\end{remark}

\subsection{Fourier transform and the dispersion relations}\label{sec:Fourier}
We now pass from the regrouped walk to its Floquet representation. Since the regrouped walk is translation invariant with respect to the blocks of length $p$ (and of size $2p$ counting the two degrees of freedom), the Fourier transform diagonalizes the block translations and reduces the walk to a $2p\times 2p$ Floquet matrix.

Here it is convenient to identify $\C^{2p}$ with $\C^p\otimes\C^2$. Thus we write
\begin{equation}\label{def:H-regrouped-2}
\mathcal H\equiv \ell^2(\Z)\otimes \C^{2p}
\equiv \ell^2(\Z)\otimes \C^p\otimes \C^2 .
\end{equation}

Let $\{e_j\}_{j=1}^p$ denote the standard basis of $\C^p$. In the regrouped representation \eqref{def:H-regrouped-2},
the coin operator is given by
\begin{equation}\label{C:ext}
C(c)= \idty_{\Z} \otimes \bigoplus_{j=1}^p C_j
\equiv \idty_\Z\otimes \sum_{j=1}^p |e_j\rangle\langle e_j|\otimes C_j .
\end{equation}
The shift operator takes the form
\begin{equation}\label{S:ext}
S\equiv \idty_{\Z}\otimes \sum_{j=1}^{p-1}
\left(|e_j\rangle\langle e_{j+1}|\otimes P_-
+
|e_{j+1}\rangle\langle e_j|\otimes P_+\right)
+ T^{-1}\otimes \left(|e_p\rangle\langle e_1|\otimes P_-\right)
+ T\otimes \left( |e_1\rangle\langle e_p|\otimes P_+\right).
\end{equation}
The last two terms encode the transitions between neighboring blocks.

Consider the Fourier transform $\mathscr{F}=\mathcal F\otimes \idty_{2p}$ on
$\ell^2(\Z)\otimes \C^{2p}$, where
$\mathcal F:\ell^2(\Z)\to L^2([0,2\pi))$ is given by
\begin{equation}
(\mathcal F\phi)(\theta)=\frac{1}{\sqrt{2\pi}}\sum_{n\in\Z}e^{-in\theta}\phi_n .
\end{equation}
It follows from \eqref{C:ext} and \eqref{S:ext} that, for $p\geq2$,
\begin{equation}
( \mathscr{F}U_p(c)\mathscr{F}^{-1}\hat\psi)(\theta)
=\widehat U_p(c,\theta)\hat{\psi}(\theta),
\end{equation}
where $\widehat U_p(c,\theta)$ is the $2p\times 2p$ unitary Floquet matrix acting on
$\C^p\otimes \C^2$ and defined by
\begin{equation}\label{def:Up-hat}
\widehat U_p(c,\theta)=\widehat S(\theta)\widehat C(c).
\end{equation}
Here
\begin{equation}\label{def:S-hat}
\widehat S(\theta)=\sum_{j=1}^{p}
\left(|e_{j-1}\rangle\langle e_{j}|\otimes P_-
+
|e_{j+1}\rangle\langle e_j|\otimes P_+\right),
\quad \text{i.e., }\widehat S(\theta)|e_j,\pm\rangle=|e_{j\pm1},\pm\rangle,
\end{equation}
where the boundary indices are understood according to the convention
\begin{equation}
 |e_0\rangle := e^{-i\theta}|e_p\rangle,
 \qquad
 |e_{p+1}\rangle := e^{i\theta}|e_1\rangle .
\end{equation}
Moreover,
\begin{equation}\label{def:C-hat}
\widehat C(c)=\sum_{j=1}^p |e_j \rangle\langle e_j|\otimes C_j.
\end{equation}

For example, when $p=3$, the Floquet matrix $\widehat U_3(c,\theta)$ is given by the
following block matrix:
\begin{equation}\label{def:U3}
  \begin{bmatrix}
 0 & P_-C_2 &  e^{i\theta}P_+C_3\\
 P_+ C_1 & 0 & P_- C_3 \\
  e^{-i\theta}P_- C_1 &  P_+ C_2 & 0
 \end{bmatrix}
 =
 \left[\begin{array}{cc|cc|cc}
 0 & 0 & 0 &0&  e^{i\theta}a_3 & e^{i\theta}b_3\\
 0 &0 & -\overline{b_2} & \overline{a_2} &  0&0\\
 \hline
 a_1 & b_1 & 0 & 0 & 0  & 0 \\
 0 & 0& 0& 0& -\overline{b_3} & \overline{a_3} \\
 \hline
 0 & 0&  a_2 & b_2 & 0 & 0\\
 -e^{-i\theta}\overline{b_1} & e^{-i\theta}\overline{a_1} &  0& 0& 0&0
 \end{array}\right].
\end{equation}

We note that, in the special case $p=2$, the matrix $\widehat U_2(c,\theta)$
has the structurally different form:
\begin{equation}\label{def:U2}
 \widehat U_2(c,\theta)=
 \left[
 \begin{array}{cc|cc}
 0 & 0 & e^{i\theta} a_2 & e^{i\theta}b_2\\
 0 & 0  & -\overline{b_2} &\overline{a_2}\\
 \hline
 a_1 & b_1 & 0 & 0\\
 -e^{-i\theta}\overline{b_1} & e^{-i\theta}\overline{a_1} & 0 & 0
 \end{array}
 \right].
\end{equation}

Let $\{\lambda_j(c,\theta)\}_{j=1}^{2p}$ denote the eigenvalues of
$\widehat U_p(c,\theta)$. Since $\widehat U_p(c,\theta)$ is a finite-dimensional
unitary matrix whose entries depend analytically on $\theta$, its eigenvalues
can be parametrized locally by analytic functions of $\theta$, counted with
multiplicity. The derivatives $\partial_\theta\lambda_j(c,\theta)$ below are
therefore understood along such local analytic branches. Eigenvalue crossings
may change the labelling of the branches, but they do not affect the value of
the maximum appearing below.

The block velocity is determined
by the maximal slope of the dispersion relations, see \cite[eq. (3.3)]{ARCSW25}, which follows from the fact that the rescaled dynamics of $Q_p$ converges strongly to the velocity operator on $\mathcal H$, see e.g.,  \cite[Theorem 9.4]{DamanikSpreading2016}, \cite[Proposition 2.2]{extail25}, \cite[Theorem 4]{AhlbrechtEtal2011}. 
\begin{equation}\label{def:block-v}
v^{(p)}_{\mathrm{blk}}(U_p(c))=
\max_{\substack{\theta\in[0,2\pi) \\ j=1,\ldots,2p}}
|\partial_\theta\lambda_j(c,\theta)|,
\qquad
\partial_{\theta}\lambda_j(c,\theta)
:=
\frac{\partial}{\partial \theta}\lambda_j(c,\theta),
\end{equation}
where $v^{(p)}_{\mathrm{blk}}(U_p(c))$ is defined in \eqref{def:v-p}. Hence, if the eigenvalues are taking the form $\lambda_j(c,\theta)=e^{i\omega_j(c,\theta)}$, then
\begin{equation}\label{def:v-dispersion}
v(U_p(c))=p\ v^{(p)}_{\mathrm{blk}}(U_p(c))
=p\max_{\substack{\theta\in[0,2\pi) \\ j=1,\ldots,2p}}
|\partial_\theta\omega_j(c,\theta)|.
\end{equation}

Thus, the problem of bounding $v(U_p(c))$ is equivalent to the problem of bounding the
derivatives of the eigenvalues of the $2p\times 2p$ Floquet matrix
$\widehat U_p(c,\theta)$.

\subsection{Numerical illustration of the bounds}\label{sec:numerics}

In this section, we provide numerical illustrations of the bounds on the dispersion relations $\omega_j(c,\theta)$ of the Floquet matrix $\widehat U_{p=5}(c,\theta)$, with $b_j=\sqrt{1-|a_j|^2}$,\footnote{Note that if $b_j\in[0,1)$ is real for all $j=1,\ldots,2p$, then the spectrum of the Floquet matrix $\widehat U_p(c,\theta)$ is invariant under complex conjugation. Consequently, its eigenvalues may be written in pairs of the form $e^{\pm i\omega_j(c,\theta)}$, see Remark \ref{spec:symmetry}. In figures \ref{fig:8710051} to \ref{fig:81914}, we plot only the nonnegative phases $\omega_j(c,\theta)\geq 0$.} and on the corresponding group velocities $|\partial_\theta\omega_j(c,\theta)|$ for $j=1,\ldots,10$.
We emphasize that the bounds displayed in the figures are for the block velocity $v_{\mathrm{blk}}^{(5)}(U_{p=5}(c))$.

After dividing by the period $p=5$ (see the relation between $v(U_5(c))$ and $v_{\mathrm{blk}}^{(5)}(U_5(c))$ in \eqref{def:v-dispersion}), the green dashed line represents the bottleneck-effect bound from Theorem \ref{thm:perturbative}, while the blue line and the red dotted line represent, respectively, the harmonic bound and the refined harmonic mean bound from Theorem \ref{thm:harmonic}. The values of the coefficients $a_j$ are specified in each graph.

We now illustrate the three velocity bounds in several representative cases with period $p=5$.

\begin{figure}[H]
\includegraphics[scale=0.55]{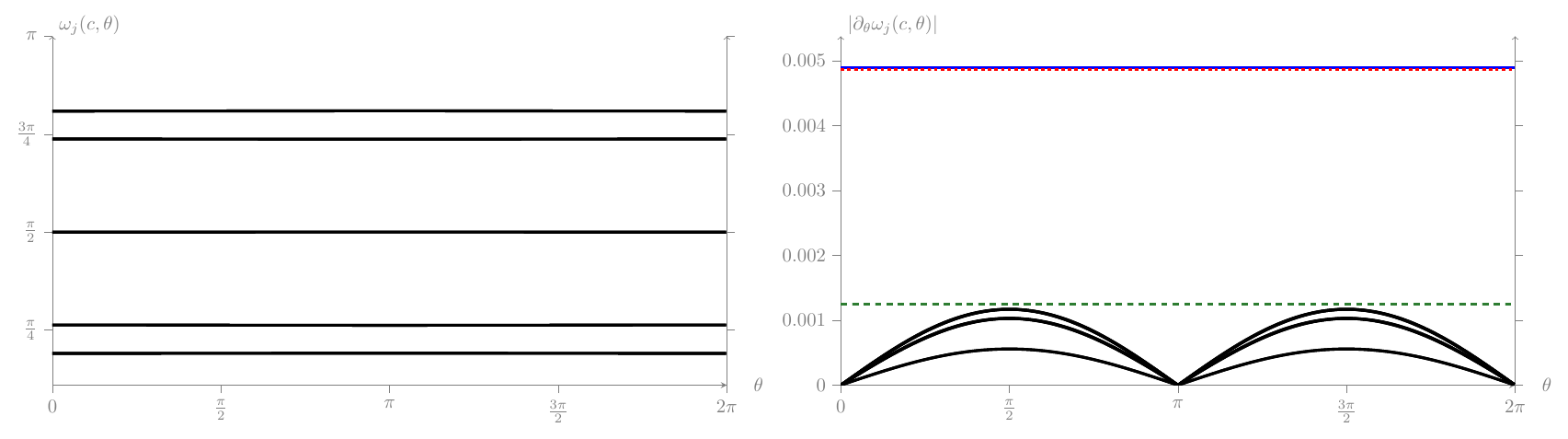}
\caption{transmission parameters: $a_1=0.8,\ a_2=0.7,\  a_3=1,\ a_4=0.005,\ a_5=1$.}
\label{fig:8710051}
\end{figure}
\begin{figure}[H]
\includegraphics[scale=0.55]{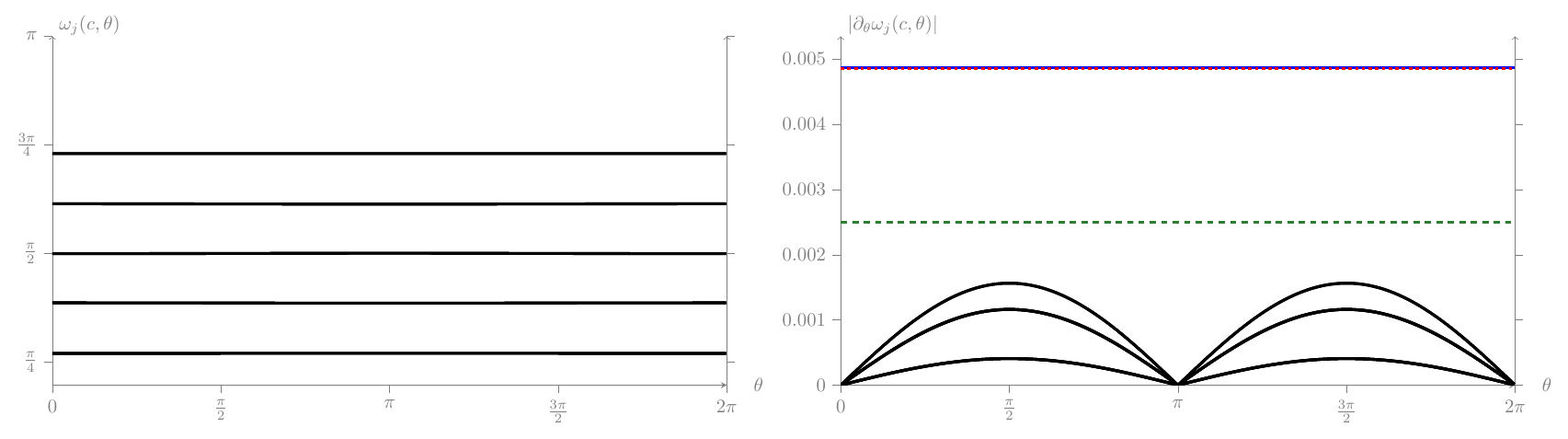}
\caption{transmission parameters: $a_1=0.8,\ a_2=0.7,\ a_3=0.6,\ a_4=0.005,\ a_5=0.8$.}
\label{fig:8760058}
\end{figure}
\begin{figure}[H]
\includegraphics[scale=0.55]{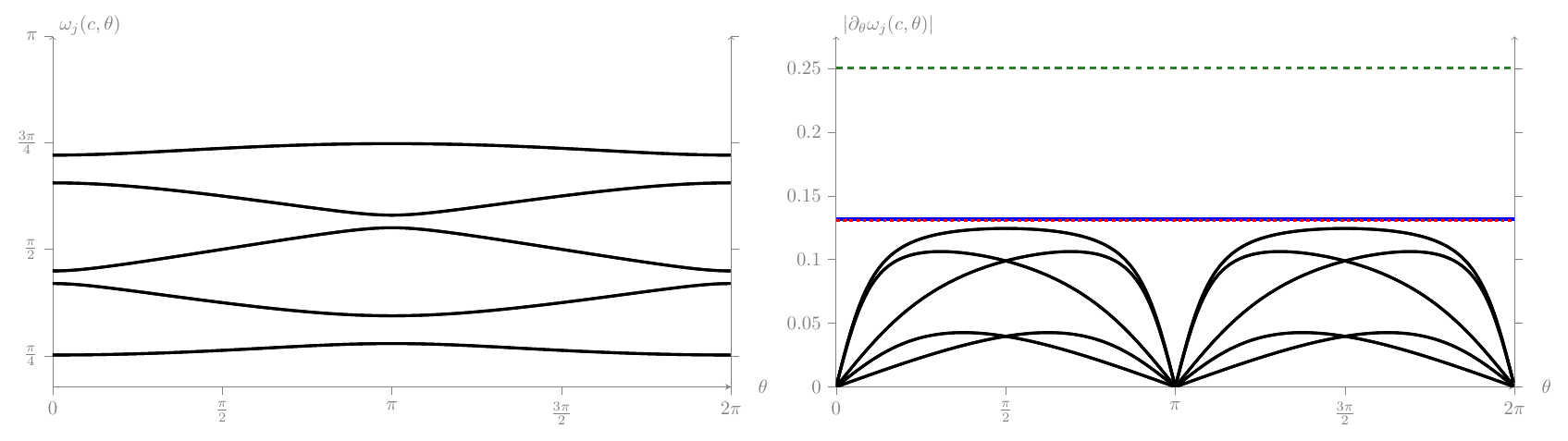}
\caption{transmission parameters: $a_1=0.8,\ a_2=0.7,\ a_3=0.6,\ a_4=0.5,\ a_5=0.8$.}
\label{fig:87658}
\end{figure}
\begin{figure}[H]
\includegraphics[scale=0.55]{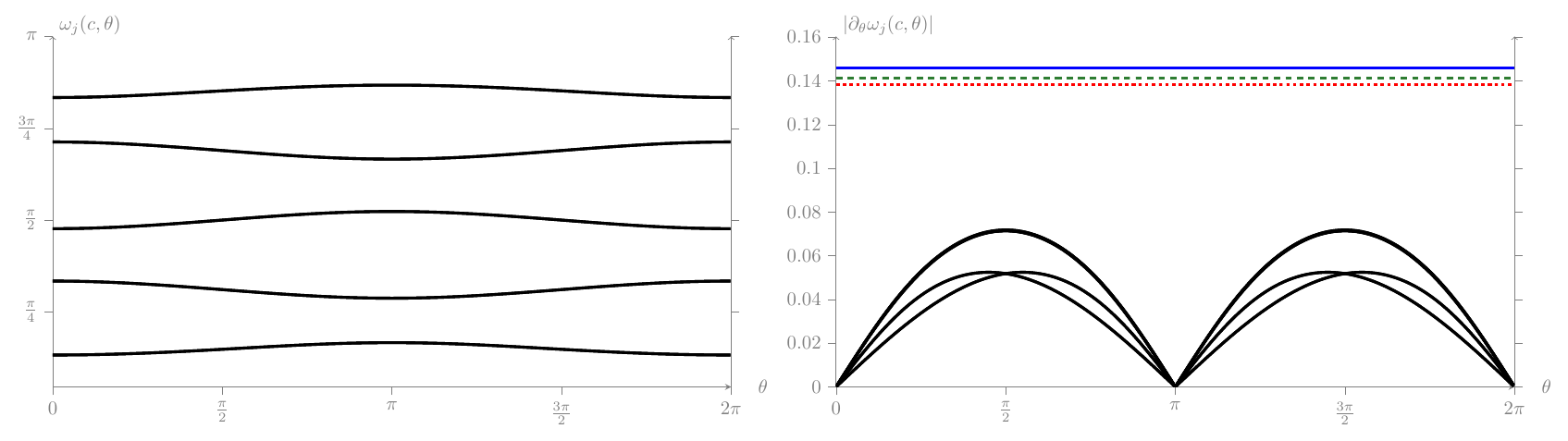}
\caption{transmission parameters: $a_1=0.8,\ a_2=1,\ a_3=0.9,\ a_4=1,\ a_5=0.4$.}
\label{fig:81914}
\end{figure}
The first two figures illustrate the bottleneck effect caused by the very small transmission parameter $a_4=0.005$. In Figure \ref{fig:8710051}, the bottleneck is adjacent to perfect transmitters, namely $a_3=a_5=1$. This is the regime in which the bottleneck-effect estimate is most pronounced, providing a sharper description of the velocity scale than the harmonic-type bounds; indeed letting $M$ be the true maximum and $m$ be the minimum bound we have that $\frac{m-M}{M}\approx 0.07$, meaning the minimum bound is within $7\%$ of its true value. Figure \ref{fig:8760058} shows that this bottleneck effect remains visible even when the neighboring coefficients are no longer perfect transmitters. In this case, the bottleneck-effect bound still captures the relevant velocity scale, although the separation from the harmonic-type bounds is less pronounced; indeed the minimum bound is within $59.9\%$  of the true maximum value. By contrast, Figure \ref{fig:87658} represents a more generic configuration, in which no transmission parameter is close to zero. In this regime, the bottleneck bound is no longer competitive, while the harmonic-type bounds, and in particular the refined harmonic mean bound, provide substantially better bounds on the group velocities. Here the refined harmonic mean is within $5.13\%$  of its true value. 
Finally, Figure \ref{fig:81914} illustrates the effect of the correction term in the refined harmonic bound \eqref{def:Hc-tilde}. Although in the examples illustrated in Figures \ref{fig:8710051}, \ref{fig:8760058}, and \ref{fig:87658}, the refined harmonic bound is very close to the harmonic bound, this configuration has larger variations between consecutive transmission parameters. In this case, the refined harmonic mean bound is visibly sharper than the harmonic bound. 

\begin{remark}\label{rem:not-min}
Figures~\ref{fig:8710051}, \ref{fig:8760058}, and \ref{fig:87658} show that the quantity
$\frac{1}{5}\min_k |a_k|$ does not provide an upper bound for the block velocity (That is, $\min_k |a_k|$ is not an upper bound for $v(U_{p}(c))$). For instance, in Figure~\ref{fig:8710051} one has
$\frac{1}{5}\min_k |a_k|=0.001$,
whereas the graphs of the derivatives of dispersion relations show clearly that $v_{\mathrm{blk}}^{(5)}(U_5(c))$ is larger than $0.001$.
\end{remark}

\begin{remark}\label{rem:not-min-2}
Figure~\ref{fig:81914} also shows that, even in the phase-free real-coin case, the quantity
$\frac{1}{p}\min_j |a_j|$ need not be a lower bound for the velocity. In this example $p=5$ and
$\frac{1}{5}\min_j |a_j|=0.08$,
whereas the graphs of the derivatives of the dispersion relations show that
$v_{\mathrm{blk}}^{(5)}(U_5(c))<0.08$.
Thus this simple minimum-type candidate does not provide a lower bound for the velocity. A general lower bound is still available, but it has a different exponential scale, see Appendix \ref{sec:Lower-bound}.
\end{remark}

\begin{remark}\label{rem:largeperiod}
The numerical examples above illustrate the behavior of the bounds for period $p=5$. Figure \ref{fig:largeperiod1} tests the same bounds in a larger-period regime, namely $p=50$.
\begin{figure}[H]
\includegraphics[scale=0.5]{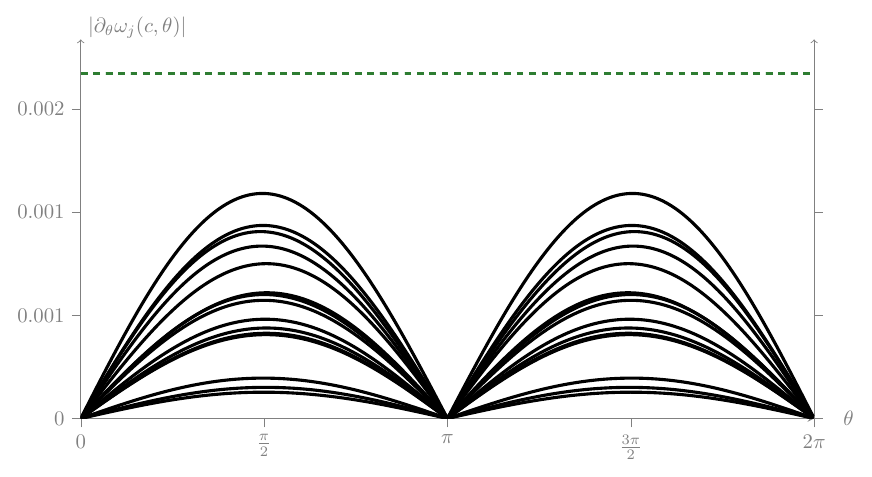}
\includegraphics[scale=0.5]{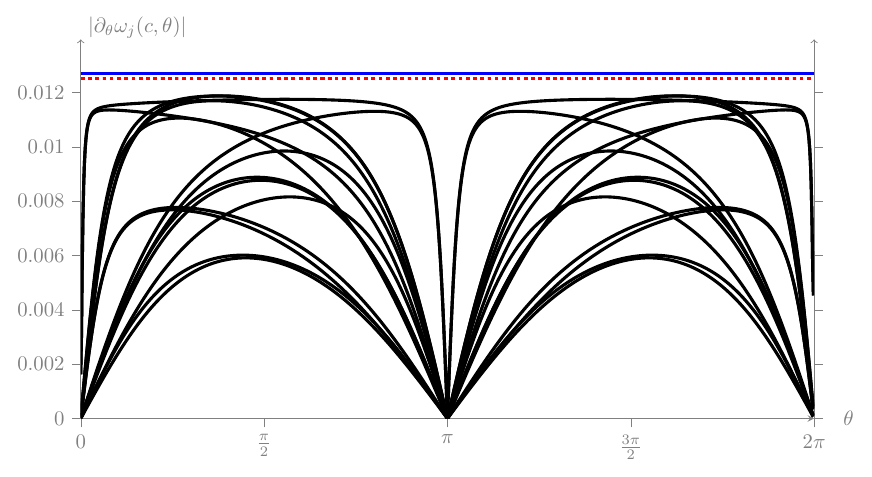}
\caption{ (Left)  $a_{16}=0.005$ cushioned from the left and right by 15 and 14 perfect transmitters, i.e., $a_i=1$, respectively, followed by $20$ parameters sampled from $[0.6,0.7]$.\\
(Right) $50$ transmission parameters uniformly sampled from $[0.6,0.8]$.}
\label{fig:largeperiod1}
\end{figure}
\footnote{In both figures here only the $30$ largest eigenvalue derivatives are plotted for reasons of legibility and rendering.}
The left panel of Figure \ref{fig:largeperiod1} considers a bottleneck configuration, where the very small transmission parameter $0.005$
  is cushioned on either side by perfect transmitters then by parameters sampled from $[0.6,0.7]$. In this regime, the minimum bound remains within $40.13\%$ of the exact velocity. The right panel considers a near-homogeneous profile, with $50$ real transmission parameters uniformly sampled from $[0.6,0.8]$. In this case, the refined harmonic bound is within $15.4\%$ of the exact velocity.
\end{remark}

\section{Perturbation around the decoupled limit}\label{sec:perturbation}
In this section we include the proof of Theorem \ref{thm:perturbative}. We use a perturbative argument that combines the block velocity and Fourier transform framework developed earlier in Section \ref{sec:block-velocity-Fourier} with the simplicity of the Floquet spectrum in the decoupled limit discussed in  Section \ref{sec:simple-spec} below. Related perturbative methods were used in \cite{ARFFW} to determine the velocity scaling for Schr\"odinger operators with periodic non-degenerate potentials.

The case $p=1$ of Theorem \ref{thm:perturbative} follows from the homogeneous case, where the velocity is known to be equal to $|a_1|$, see, e.g., \cite{ARCSW25}. The case $p=2$ follows from Lemma \ref{lem:warm-up}. Hence, for the remainder of this section, we assume that $p>2$.

\subsection{The decoupled limit and analytic perturbation} 
For the $p$-periodic string $c$ in \eqref{def:c-p},
suppose that
\begin{equation}
0<\varrho=\min_{1\leq k\leq p}|a_k| \ll 1,
\end{equation}
and assume that this minimum is attained at a unique site, i.e., 
\begin{equation}
|\{k:\ |a_k|=\varrho\}|=1.
\end{equation}
By the cyclic regrouping invariance of the velocity, see Remark \ref{rem:regrouping}, we may assume without loss of generality that
$|a_1|=\varrho$.
Write
\begin{equation}\label{def:a1-b1}
a_1=\alpha\varrho,\
\text{ and }\ 
b_1=\beta\sqrt{1-\varrho^2},
\quad \text{ with }\quad |\alpha|=|\beta|=1.
\end{equation}

We understand $U_p(c):=U_p(c(\varrho))$ as a perturbation of $U_p(c(0))$, where
\begin{equation}\label{def:c0}
c(0):=((0,\beta),(a_2,b_2),\ldots,(a_p,b_p)),
\qquad a_k\neq0 \text{ for all } k=2,\ldots, p.
\end{equation}
Equivalently, we consider the one-parameter family
\begin{equation}\label{def:c-eps-family}
c(\varrho):=((\alpha\varrho,\beta\sqrt{1-\varrho^2}),(a_2,b_2),\ldots,(a_p,b_p)),
\end{equation}
where the parameters $(a_2,b_2),\ldots,(a_p,b_p)$ are kept fixed. The nontrivial walk corresponds to the given positive value of $\varrho$, while the decoupled walk corresponds to $\varrho=0$. In this decoupled setting, the coins $C_{kp}$, where $k\in\mathbb Z$, are then regarded as $p$-periodic perfect reflectors, see \cite[Sec. 2.1]{ARJS26}. Due to the periodicity of these reflectors, initial states are trapped between each pair of consecutive reflectors, leading to zero velocity.

The Floquet matrix $\widehat U_p(c(0),\theta)$ of $U_p(c(0))$ is given by, see \eqref{def:S-hat} and \eqref{def:C-hat},
\begin{equation}\label{def:Up-all}
\widehat U_p(c(0),\theta)
=
\sum_{j=1}^{p}
\left(
|e_{j-1}\rangle\langle e_j|\otimes P_-
+
|e_{j+1}\rangle\langle e_j|\otimes P_+
\right)
(\idty_p\otimes C_j),
\qquad
C_1=
\begin{bmatrix}
0 & \beta\\
-\overline{\beta} & 0
\end{bmatrix},
\end{equation}
and the boundary convention is
\begin{equation}
|e_0\rangle := e^{-i\theta}|e_p\rangle,
\qquad
|e_{p+1}\rangle := e^{i\theta}|e_1\rangle.
\end{equation}
For example, for $p=3$, we obtain
\begin{equation}\label{ex:hatU-3}
\widehat U_3(c(0),\theta)=
\left[\begin{array}{cc|cc|cc}
0 & 0 & 0 & 0 & e^{i\theta}a_3 & e^{i\theta}b_3\\
0 & 0  & -\overline{b_2} & \overline{a_2} & 0 & 0\\
\hline
0 & \beta & 0 & 0 & 0 & 0\\
0 & 0 & 0 & 0 & -\overline{b_3} & \overline{a_3}\\
\hline
0 & 0 & a_2 & b_2 & 0 & 0\\
-e^{-i\theta}\overline{\beta} & 0 & 0 & 0 & 0 & 0
\end{array}
\right].
\end{equation}

It is direct to check that
\begin{equation}\label{U-gauge}
\widehat U_p(c(0),\theta)=D_{\theta}\widehat U_p(c(0),0)D_{\theta}^{-1},
\end{equation}
where $D_\theta$ is defined by
\begin{equation}\label{def:D-theta}
D_{\theta}=\idty_{2p}+(e^{i\theta}-1)\left(|e_1\rangle\langle e_1|\otimes P_+\right)=\diag\{e^{i\theta},1,1,\ldots,1\}_{2p\times 2p}.
\end{equation}
Thus $\widehat U_p(c(0),\theta)$ is unitarily equivalent to $\widehat U_p(c(0),0)$. In particular, $\mathrm{spec}(\widehat U_p(c(0),\theta))$ is independent of $\theta$. Hence, by \eqref{def:v-dispersion},
\begin{equation}
v(U_p(c(0)))=v^{(p)}_{\mathrm{blk}}(U_p(c(0)))=0.
\end{equation}
Moreover, Theorem \ref{thm:simple-spec}\eqref{thm:spectrum:2} shows that $\mathrm{spec}(\widehat U_p(c(0),\theta))$ is simple.

We now perturb away from the decoupled case. Since
\begin{equation}
b_1(\varrho)=\beta\sqrt{1-\varrho^2}
\end{equation}
is real analytic near $\varrho=0$, the matrix $\widehat U_p(c(\varrho),\theta)$ admits an analytic extension in $\varrho$ to a neighborhood of $0$. Although by definition we have $\varrho >0$, this analytic extension allows us to apply analytic perturbation theory at $\varrho=0$. Moreover, the dependence on $\theta$ is analytic through the factors $e^{\pm i\theta}$.

Since the spectrum of $\widehat U_p(c(0),\theta)$ is independent of $\theta$ by \eqref{U-gauge}, and simple by Theorem \ref{thm:simple-spec}\eqref{thm:spectrum:2}, the gaps between distinct eigenvalues of $\widehat U_p(c(0),\theta)$ are uniform in $\theta$. Therefore, for each eigenvalue $\lambda_j(c(0),0)$ of $\widehat U_p(c(0),0)$, there exists a contour $\Gamma_j=\Gamma_j(c(0))$ enclosing only the simple eigenvalue $\lambda_j(c(0),0)=\lambda_j(c(0),\theta)$ (after relabeling the eigenvalues) of $\widehat U_p(c(0),\theta)$, uniformly for all $\theta\in[0,2\pi)$.

Define the corresponding Riesz projection by \footnote{The integral in \eqref{def:Riesz-proj} is understood for fixed $\theta\in[0,2\pi)$. The contour $\Gamma_j$ is chosen independently of $\theta$ because the spectral gap at $\varrho=0$ is uniform in $\theta$.}
\begin{equation}\label{def:Riesz-proj}
P_j(\varrho,\theta)
=
\frac{1}{2\pi i}
\int_{\Gamma_j}
\left(\xi-\widehat U_p(c(\varrho),\theta)\right)^{-1}d\xi.
\end{equation}
There exists $\rho_0=\rho_0(c(0),j)>0$ such that for $|\varrho|<2\rho_0$, $P_j(\varrho,\theta)$ has rank one. The associated eigenvalue branch is given by
\begin{equation}\label{def:eig-branch-trace}
\lambda_j(c(\varrho),\theta)
=
\mathrm{Tr}\left(\widehat U_p(c(\varrho),\theta)P_j(\varrho,\theta)\right).
\end{equation}
By analytic perturbation theory for simple eigenvalues, see e.g., \cite[Chap.~II, Sec.~1]{Kato}, the eigenvalue $\lambda_j(c(\varrho),\theta)$ of $\widehat U_p(c(\varrho),\theta)$ is real analytic on the set
\begin{equation}
\mathcal A(2\rho_0):=\{(\varrho,\theta):\ |\varrho|<2\rho_0,\ \theta\in[0,2\pi)\}.
\end{equation}

Taylor's theorem in the variable $\varrho$ gives
\begin{equation}\label{eig-taylor}
\lambda_j(c(\varrho),\theta)
=
\lambda_j(c(0),\theta)
+
\left(\left.\partial_{\varrho}\lambda_j(c(\varrho),\theta)\right|_{\varrho=0}\right)\varrho
+
R_j(\varrho,\theta)\varrho^2,
\end{equation}
where $\partial_{\varrho}$ denotes differentiation with respect to $\varrho$, keeping the remaining parameters $(a_2,b_2),\ldots,(a_p,b_p)$ fixed. The remainder $R_j(\varrho,\theta)$ is given by \footnote{ Here Taylor's theorem is applied to the one-variable function $g_j(\varrho)=\lambda_j(c(\varrho),\theta)$,
with $\theta$ and the remaining parameters fixed. The remainder is written in integral form:
\[
g_j(\varrho)=g_j(0)+g_j'(0)\varrho+\varrho^2\int_0^1(1-s)g_j''(s\varrho)\,ds.
\]}
\begin{equation}\label{def:Rj-integral}
R_j(\varrho,\theta)
=
\int_0^1
(1-s)
\left.
\partial_\eta^2\lambda_j(c(\eta),\theta)
\right|_{\eta=s\varrho}
\,ds.
\end{equation}
Since $\lambda_j(c(\varrho),\theta)$ is analytic for $(\varrho,\theta)\in\mathcal A(2\rho_0)$, we have
\begin{equation}\label{eq:Rj-bound}
\sup_{(\varrho,\theta)\in \overline{\mathcal A(\rho_0)}}
\left(
|R_j(\varrho,\theta)|
+
|\partial_\theta R_j(\varrho,\theta)|
\right)
<\infty.
\end{equation}
Indeed, differentiating \eqref{def:Rj-integral} with respect to $\theta$ gives
\begin{equation}
\partial_\theta R_j(\varrho,\theta)
=
\int_0^1
(1-s)
\partial_\theta
\left(
\left.
\partial_\eta^2\lambda_j(c(\eta),\theta)
\right|_{\eta=s\varrho}
\right)
\,ds,
\end{equation}
and the integrand is continuous on a compact set.

Our goal is to bound the derivative of $\lambda_j(c(\varrho),\theta)$, expanded in \eqref{eig-taylor}, with respect to $\theta$. The first term in \eqref{eig-taylor}, $\lambda_j(c(0),\theta)$, is independent of $\theta$ by \eqref{U-gauge}. Therefore, differentiating \eqref{eig-taylor} with respect to $\theta$, and using \eqref{eq:Rj-bound}, we obtain
\begin{equation}\label{partial-theta-lambda-b}
\partial_{\theta}\lambda_j(c(\varrho),\theta)=\left(\left.\partial_\theta\partial_{\varrho}\lambda_j(c(\varrho),\theta)\right|_{\varrho=0}\right)\varrho+\mathcal{O}(\varrho^2),
\end{equation}
where the error term is uniform in $\theta\in[0,2\pi)$.

\subsection{Eigenvector estimates and the leading-order velocity bound} 
We now extract the leading-order behavior of the velocity from the perturbation expansion \eqref{partial-theta-lambda-b} above. Proofs of Lemmas \ref{lem:2-term} and \ref{lem:eigen-bonds} are presented in Section \ref{subsec:pf:1} below.

Let $x_j(c(0),0)$ be the normalized eigenvector of $\widehat U_p(c(0),0)$ corresponding to the eigenvalue $\lambda_j(c(0),0)$.

\begin{lemma}\label{lem:2-term}
For every $\theta\in[0,2\pi)$, we have
\begin{equation}\label{lem:eq-2-term}
\left.
\partial_\theta\partial_{\varrho}\lambda_j(c(\varrho),\theta)
\right|_{\varrho=0}=2i\Re\left(e^{i\theta}\alpha\overline{\beta}\,x_{j,1}^+\overline{x_{j,1}^-}\right)\lambda_j(c(0),0),
\end{equation}
where
$x_{j,1}^\pm:=\langle \delta_1^{\pm}|x_j(c(0),0)\rangle$, \footnote{Recall that $\delta_1^\pm=|e_1,\pm\rangle=e_1\otimes|\pm\rangle$.}
and $\alpha,\beta$ are the phases defined in \eqref{def:a1-b1}. 
\end{lemma}

We obtain from \eqref{partial-theta-lambda-b}, Lemma \ref{lem:2-term}, and \eqref{def:block-v}
\begin{align}\label{pf:ratio}
\frac{v^{(p)}_{\mathrm{blk}}(U_p(c(\varrho)))}{\varrho}
&=
2\max_{\substack{\theta\in[0,2\pi)\\ j=1,\ldots,2p}}\left|\Re\left(e^{i\theta}\alpha\overline{\beta}\,x_{j,1}^+\overline{x_{j,1}^-}\right)\right|+\mathcal{O}(\varrho)\quad \text{ as }\varrho\downarrow 0,
\notag\\
&=
2\max_{j=1,\ldots,2p}|x_{j,1}^+x_{j,1}^-|+\mathcal{O}(\varrho)
\end{align}
The last step follows from observing that, for $z\in\mathbb C$, if $z=|z|e^{i\phi}$, then
\begin{equation}
\Re(e^{i\theta}z)=|z|\Re(e^{i(\theta+\phi)})=|z|\cos(\theta+\phi),
\end{equation}
which is maximized when $\theta=2\pi-\phi$.

\begin{lemma}\label{lem:eigen-bonds}
For any fixed $j=1,\ldots,2p$, the components of
$x_j(c(0),0)$ satisfy
\begin{equation}\label{bones:diag}
|x_{j,k}^-|=|x_{j,k+1}^+|, \quad \text{(the diagonal blue bonds in Figure \ref{Fig:Eigen-bones})}
\end{equation}
where $x_{j,p+1}^+:=x_{j,1}^+$. Moreover, if $\ell\neq1$ and
$|a_\ell|=1$, then
\begin{equation}\label{bones:vert}
|x_{j,\ell}^+|=|x_{j,\ell}^-|  \quad \text{(the vertical red dotted bonds in Figure \ref{Fig:Eigen-bones})}.
\end{equation}

\begin{figure}[H]
\centering
\includegraphics[width=0.9\textwidth]{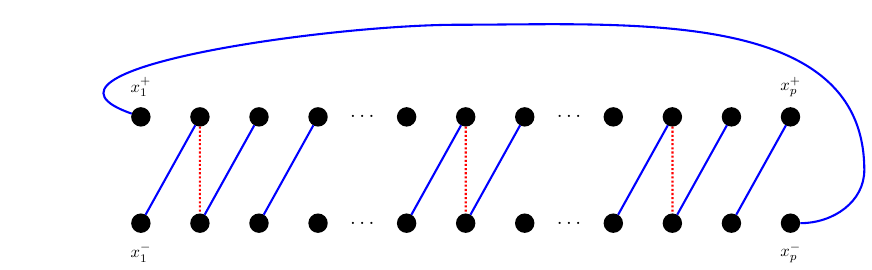}
\caption{
Schematic representation of the $2p$ components of the eigenvector
$x_j(c(0),0)$, grouped into pairs $(x_{j,k}^+,x_{j,k}^-)$.
The upper row denotes the $(+)$ components and the lower row denotes the
$(-)$ components. Blue diagonal bonds represent the equal-modulus
relations $|x_{j,k}^-|=|x_{j,k+1}^+|$, with cyclic convention
$x_{j,p+1}^+=x_{j,1}^+$. Dotted red vertical bonds indicate perfect
transmitter sites $|a_\ell|=1$, where $|x_{j,\ell}^+|=|x_{j,\ell}^-|$.
}
\label{Fig:Eigen-bones}
\end{figure}
\end{lemma}

Observe first that in the case
\begin{equation}\label{eq:perfect-all}
|a_1|=\varrho, \qquad |a_2|=\cdots=|a_p|=1,
\end{equation}
where $\varrho\in(0,1)$, the corresponding decoupled vector has
\begin{equation}
a_1=0,\qquad |a_2|=\cdots=|a_p|=1.
\end{equation}
Hence Lemma \ref{lem:eigen-bonds}, applied to the eigenvectors
$x_j(c(0),0)$, shows that
\begin{equation}
|x_{j,n}^{s_1}|=|x_{j,m}^{s_2}|
\end{equation}
for any $n,m\in\{1,\ldots,p\}$, any $s_1,s_2\in\{+,-\}$, and any
$j=1,\ldots,2p$. Since $\|x_j(c(0),0)\|=1$, it follows that
\begin{equation}
|x_{j,n}^\pm| = \frac{1}{\sqrt{2p}}, \qquad n=1,\ldots,p, \quad j=1,\ldots,2p.
\end{equation}
Substituting this into \eqref{pf:ratio}, we obtain
\begin{equation}
\frac{v^{(p)}_{\mathrm{blk}}(U_p(c))}{\varrho} = \frac{1}{p} + \mathcal{O}(\varrho).
\end{equation}
Then \eqref{eq:v-vp} gives the desired conclusion for the case
$|a_1|=\varrho, |a_2|=\cdots=|a_p|=1$,
in \eqref{eq:main:bound}.

We now consider the case where \eqref{eq:perfect-all} is not satisfied, but we assume that some perfect transmitters are placed to
the right and to the left of $a_1$. Let $R_1$ and $L_1$ be the
number of consecutive perfect transmitters to the right and to the left,
respectively, of $a_1$ in the periodic string, as in \eqref{def:RL}.
By Lemma \ref{lem:eigen-bonds}, these perfect transmitters produce two
chains of equal-modulus entries in the decoupled eigenvector
$x_j(c(0),0)$: one chain of length $2(R_1+1)$ involving
$|x_{j,1}^-|$, and one chain of length $2(L_1+1)$ involving
$|x_{j,1}^+|$. Therefore, $\|x_j(c(0),0)\|^2=1$ implies
\begin{equation}
2(L_1+1)|x_{j,1}^+|^2+2(R_1+1)|x_{j,1}^-|^2+\mathrm{Rest}=1,
\end{equation}
where $\mathrm{Rest}\geq0$. In particular,
\begin{equation}
2(L_1+1)|x_{j,1}^+|^2+2(R_1+1)|x_{j,1}^-|^2\leq1.
\end{equation}
Set
\begin{equation}
x:=|x_{j,1}^+|,
\qquad
y:=|x_{j,1}^-|.
\end{equation}
Then the quantity appearing in \eqref{pf:ratio} satisfies
\begin{equation}
2|x_{j,1}^+x_{j,1}^-|=2xy,
\end{equation}
and it is bounded above by the solution of the optimization problem
\begin{equation}
\max 2xy
\quad
\text{subject to}
\quad
2(L_1+1)x^2+2(R_1+1)y^2\leq1,
\qquad x,y\geq0.
\end{equation}
The maximum is attained when the constraint is saturated. Using
Lagrange multipliers, or equivalently setting
\begin{equation}
X=\sqrt{2(L_1+1)}x,
\qquad
Y=\sqrt{2(R_1+1)}y,
\end{equation}
we get
\begin{equation}\label{eq:X-Y}
X^2+Y^2=1,
\qquad
2xy=
\frac{XY}{\sqrt{(L_1+1)(R_1+1)}}.
\end{equation}
Since
\begin{equation}
2XY\leq X^2+Y^2=1,
\qquad\text{hence}\qquad
XY\leq\frac12,
\end{equation}
we obtain from \eqref{eq:X-Y}
\begin{equation}
2|x_{j,1}^+x_{j,1}^-|
\leq
\frac{1}{2\sqrt{(R_1+1)(L_1+1)}}.
\end{equation}

\begin{remark}\label{rem:multiple-reflectors}
The only point in the  argument above where the assumption of a single perfect reflector is used is the simplicity of the spectrum of the decoupled
Floquet matrix. In the case of one perfect reflector, this simplicity is provided by Theorem \ref{thm:simple-spec}(\ref{thm:spectrum:2}). Hence the
same perturbative argument applies to configurations with several perfect reflectors, provided that the corresponding decoupled Floquet matrix has
simple spectrum.

More precisely, let $c(\varrho)$ be a $p$-periodic family such that
\begin{equation}
0<\varrho=\min_{1\leq k\leq p}|a_k|\ll 1,
\end{equation}
and assume that, for $\varrho$ sufficiently small, the set of minimizers is
fixed and given by
\begin{equation}
T=\{j_1,\ldots,j_n\}\subset \{1,\ldots,p\}.
\end{equation}
Then $U_p(c(\varrho))$ may be viewed as a perturbation of the decoupled walk
$U_p(c_{\perp T})$, where $c_{\perp T}$ denotes the $p$-periodic coin
sequence obtained by replacing the coins at the sites in $T$ by perfect
reflectors, in the real case, this means $(c_{\perp T})_m=(0,1)$, for $m\in T$.
Assume that the spectrum of the corresponding Floquet matrix
$\widehat U_p(c_{\perp T},0)$ is simple, as in the Theorem \ref{thm:spectrum:3} below. Applying the one-reflector estimate locally at
each perfect reflector $m\in T$ gives
\begin{equation}\label{rem:MR-bound}
\limsup_{\varrho\to 0}
\frac{v_{\mathrm{blk}}^{(p)}(U_p(c(\varrho)))}{\varrho}
\leq
\min_{m\in T}
\frac{1}{2\sqrt{(L_m+1)(R_m+1)}}=\frac{1}{2B},
\quad
B:=\max_{m\in T}\sqrt{(L_m+1)(R_m+1)}.
\end{equation}
\end{remark}

\subsection{Proofs of Lemmas \ref{lem:2-term} and \ref{lem:eigen-bonds}}\label{subsec:pf:1}
We finish the proof of Theorem \ref{thm:perturbative} by proving the two
 lemmas used above, Lemmas \ref{lem:2-term} and
\ref{lem:eigen-bonds}.

\begin{proof}[Proof of Lemma \ref{lem:2-term}]
Let $x_j(\theta):=x_j(c(0),\theta)$ be the normalized eigenvector of
$\widehat U_p(c(0),\theta)$ corresponding to the eigenvalue
$\lambda_j(c(0),\theta)$. By the gauge equivalence \eqref{U-gauge}, we may choose
\begin{equation}\label{eq:x-theta-0}
x_j(\theta)=D_{\theta}x_j(0).
\end{equation}

By the Hellmann--Feynman formula for a simple eigenvalue of a unitary
matrix, see Remark \ref{rem:HF} below, we have
\begin{equation}\label{eq:H-F-1}
\left.
\partial_{\varrho}\lambda_j(c(\varrho),\theta)
\right|_{\varrho=0}
=
\left\langle
x_j(\theta),
\left.
\partial_{\varrho}\widehat U_p(c(\varrho),\theta)
\right|_{\varrho=0}
x_j(\theta)
\right\rangle.
\end{equation}
Using \eqref{eq:x-theta-0} and then differentiating with respect to $\theta$, we obtain
\begin{equation}\label{partial-theta-lambda}
\left.
\partial_\theta\partial_{\varrho}\lambda_j(c(\varrho),\theta)
\right|_{\varrho=0}
=
\left\langle
x_j(0),
\partial_\theta
\left(
D_{\theta}^{-1}
\left.
\partial_{\varrho}\widehat U_p(c(\varrho),\theta)
\right|_{\varrho=0}
D_{\theta}
\right)
x_j(0)
\right\rangle.
\end{equation}

\begin{remark}\label{rem:HF}
Here is a careful derivation of \eqref{eq:H-F-1}.

For $(\varrho,\theta)\in\mathcal A(\rho_0)$, take the normalized analytic eigenvector
$x_j(c(\varrho),\theta)\in\mathrm{Ran}\ P_j(\varrho,\theta)$ defined in
\eqref{def:Riesz-proj}.

 That is, we have
\begin{equation}
\lambda_j(c(\varrho),\theta)
=
\left\langle
x_j(c(\varrho),\theta),
\widehat U_p(c(\varrho),\theta)x_j(c(\varrho),\theta)
\right\rangle.
\end{equation}
Taking the derivative with respect to $\varrho$, and omitting all arguments to simplify notations, we obtain
\begin{eqnarray}
\partial_{\varrho}\lambda_j
&=&
\langle x'_j,\widehat U x_j\rangle
+
\langle x_j,\widehat U' x_j\rangle
+
\langle x_j,\widehat U x'_j\rangle
\notag\\
&=&
\lambda_j
\left(
\langle x'_j,x_j\rangle
+
\langle x_j,x'_j\rangle
\right)
+
\langle x_j,\widehat U'x_j\rangle,
\end{eqnarray}
where all derivatives are taken with respect to $\varrho$. In the second equality, we used
$\widehat U x_j=\lambda_jx_j$ and, since $\widehat U$ is unitary, hence normal,
$\widehat U^*x_j=\overline{\lambda_j}x_j$. Then observe that
\begin{equation}
\langle x'_j,x_j\rangle+\langle x_j,x'_j\rangle
=
\partial_{\varrho}\|x_j(c(\varrho),\theta)\|^2
=
0.
\end{equation}
Substituting $\varrho=0$, we obtain \eqref{eq:H-F-1}.
\end{remark}

Now
\begin{equation}
\partial_{\varrho}C_1
=
\begin{bmatrix}
\alpha&-\dfrac{\beta\varrho}{\sqrt{1-\varrho^2}}\\
\dfrac{\overline{\beta}\varrho}{\sqrt{1-\varrho^2}}&\overline{\alpha}
\end{bmatrix} \text{ and hence }\left.
\partial_{\varrho}C_1\right|_{\varrho=0}
=
\begin{bmatrix}
\alpha & 0\\
0 & \overline{\alpha}
\end{bmatrix}.
\end{equation}
Therefore, by \eqref{def:Up-hat}, \eqref{def:S-hat}, and \eqref{def:C-hat},
\begin{equation}
\left.
\partial_{\varrho}\widehat U_p(c(\varrho),\theta)
\right|_{\varrho=0}
=
e^{-i\theta}|e_p\rangle\langle e_1|\otimes \overline{\alpha}P_-
+
|e_2\rangle\langle e_1|\otimes \alpha P_+.
\end{equation}
Thus,
\begin{equation}\label{eq:partial-U}
D_{\theta}^{-1}
\left.
\partial_{\varrho}\widehat U_p(c(\varrho),\theta)
\right|_{\varrho=0}
D_{\theta}
=
e^{-i\theta}|e_p\rangle\langle e_1|\otimes \overline{\alpha}P_-
+
e^{i\theta}|e_2\rangle\langle e_1|\otimes \alpha P_+.
\end{equation}
Differentiating \eqref{eq:partial-U} with respect to $\theta$, we get
\begin{equation}\label{d-theta-3}
\partial_\theta
\left(
D_\theta^{-1}
\left.
\partial_{\varrho}\widehat U_p(c(\varrho),\theta)
\right|_{\varrho=0}
D_\theta
\right)
=
-i e^{-i\theta}|e_p\rangle\langle e_1|\otimes \overline{\alpha}P_-
+
i e^{i\theta}|e_2\rangle\langle e_1|\otimes \alpha P_+.
\end{equation}

Write
\begin{equation}
x_j:=x_j(c(0),0)
=
\sum_{m=1}^p e_m\otimes x_{j,m},
\qquad
x_{j,m}:=
\begin{bmatrix}
x_{j,m}^+\\
x_{j,m}^-
\end{bmatrix}.
\end{equation}
Substituting \eqref{d-theta-3} into \eqref{partial-theta-lambda}, we obtain
\begin{align}\label{partial-2-1}
\left.
\partial_{\theta}\partial_{\varrho}\lambda_j(c(\varrho),\theta)
\right|_{\varrho=0}
&=
-i e^{-i\theta}\overline{\alpha}\,
\langle x_{j,p},P_-x_{j,1}\rangle
+
i e^{i\theta}\alpha\,
\langle x_{j,2},P_+x_{j,1}\rangle
\notag\\
&=
i\left(
e^{i\theta}\alpha\,\overline{x_{j,2}^+}x_{j,1}^+
-
e^{-i\theta}\overline{\alpha}\,\overline{x_{j,p}^-}x_{j,1}^-
\right).
\end{align}

It remains to express $x_{j,2}^+$ and $x_{j,p}^-$ in terms of the two
components at the reflector. We claim that
\begin{equation}\label{eq:x0-relations}
x_{j,1}^-=\lambda_j(c(0),0)\overline{\beta}\,x_{j,2}^+,
\qquad
x_{j,1}^+=-\lambda_j(c(0),0)\beta\,x_{j,p}^-.
\end{equation}
Assuming \eqref{eq:x0-relations} for the moment.  Use
$|\lambda_j(c(0),0)|=1$ and $|\alpha|=|\beta|=1$, then \eqref{partial-2-1} gives
\begin{align}
\left.
\partial_{\theta}\partial_{\varrho}\lambda_j(c(\varrho),\theta)
\right|_{\varrho=0}
&=
i\lambda_j(c(0),0)
\left(
e^{i\theta}\alpha\overline{\beta}\,
\overline{x_{j,1}^-}x_{j,1}^+
+
e^{-i\theta}\overline{\alpha}\beta\,
\overline{x_{j,1}^+}x_{j,1}^-
\right)
\notag\\
&=
2i\Re\left(
e^{i\theta}\alpha\overline{\beta}\,
x_{j,1}^+\overline{x_{j,1}^-}
\right)
\lambda_j(c(0),0).
\end{align}
This proves \eqref{lem:eq-2-term}.

It remains to prove \eqref{eq:x0-relations}.
From \eqref{def:Up-hat}, $\widehat U_p(c(0),0)=\widehat S(0)\widehat C(c(0))$. The eigenvalue equation
\begin{equation}
\widehat U_p(c(0),0)x_j=\lambda_jx_j
\end{equation}
implies
\begin{equation}
\widehat C(c(0))x_j=\lambda_j(\widehat S(0))^*x_j.
\end{equation}
Hence we have
\begin{equation}\label{collecting}
C_nx_{j,n}
=
\lambda_j
\begin{bmatrix}
x_{j,n+1}^+\\
x_{j,n-1}^-
\end{bmatrix},
\end{equation}
with the definitions $x_{j,p+1}:=x_{j,1}$ and $x_{j,0}:=x_{j,p}$, i.e., the indices are defined cyclically.

In particular, for $n=1$,
\begin{equation}\label{x_k-relations}
\begin{bmatrix}
0 & \beta\\
-\overline{\beta} & 0
\end{bmatrix}
x_{j,1}
=
\lambda_j
\begin{bmatrix}
x_{j,2}^+\\
x_{j,p}^-
\end{bmatrix}.
\end{equation}
Thus
\begin{equation}
\beta x_{j,1}^-=\lambda_jx_{j,2}^+,
\qquad
-\overline{\beta}x_{j,1}^+=\lambda_jx_{j,p}^-.
\end{equation}
Equivalently,
\begin{equation}
x_{j,1}^-=\lambda_j\overline{\beta}\,x_{j,2}^+,
\qquad
x_{j,1}^+=-\lambda_j\beta\,x_{j,p}^-,
\end{equation}
which is \eqref{eq:x0-relations}. This finishes the proof.
\end{proof}
\begin{proof}[Proof of Lemma \ref{lem:eigen-bonds}]

Taking the Euclidean norm of both sides of \eqref{collecting}, and using that
$C_n$ is unitary and $|\lambda_j|=1$, we obtain
\begin{equation}\label{x-general}
|x_{j,n}^+|^2+|x_{j,n}^-|^2
=
|x_{j,n+1}^+|^2+|x_{j,n-1}^-|^2.
\end{equation}
Equivalently,
\begin{equation}\label{eq:dn-constant}
|x_{j,n+1}^+|^2-|x_{j,n}^-|^2=|x_{j,n}^+|^2-|x_{j,n-1}^-|^2.
\end{equation}
Thus the quantity
\begin{equation}\label{def:dn}
d_n:=|x_{j,n+1}^+|^2-|x_{j,n}^-|^2
\end{equation}
is independent of $n$.

Moreover, \eqref{x_k-relations} shows that for $n=1$ we have
\begin{equation}
|x_{j,1}^-|=|x_{j,2}^+|,
\qquad
|x_{j,1}^+|=|x_{j,p}^-|.
\end{equation}
In particular,
\begin{equation}
d_1=|x_{j,2}^+|^2-|x_{j,1}^-|^2=0.
\end{equation}
Since $d_n$ is independent of $n$, it follows that $d_n=0$ for all
$n=1,\ldots,p$. Therefore
\begin{equation}\label{pf:bones:diag}
|x_{j,n}^-|=|x_{j,n+1}^+|,
\qquad n=1,\ldots,p.
\end{equation}
This proves \eqref{bones:diag}, i.e., the blue bonds in Figure
\ref{Fig:Eigen-bones}.

Furthermore, if $|a_\ell|=1$ for some $\ell\neq1$, then $b_\ell=0$ and therefore
\begin{equation}
C_\ell=
\begin{bmatrix}
a_\ell & 0\\
0 & \overline{a_\ell}
\end{bmatrix}.
\end{equation}
Hence \eqref{collecting} gives
\begin{equation}\label{pf:bones:horiz}
|x_{j,\ell}^+|=|x_{j,\ell+1}^+|,
\qquad
|x_{j,\ell}^-|=|x_{j,\ell-1}^-|.
\end{equation}
Combining \eqref{pf:bones:horiz} with \eqref{pf:bones:diag}, we obtain
\begin{equation}
|x_{j,\ell}^+|
\overset{\eqref{pf:bones:horiz}}{=}
|x_{j,\ell+1}^+|
\overset{\eqref{pf:bones:diag}}{=}
|x_{j,\ell}^-|.
\end{equation}
This proves \eqref{bones:vert}, i.e., the dotted red bonds in Figure \ref{Fig:Eigen-bones}.
\end{proof}

\section{The harmonic-type velocity bounds}\label{sec:harmonic}
In this section we prove Theorem \ref{thm:harmonic}. The argument gives a
non-perturbative upper bound for arbitrary $p$-periodic strings with nonzero
transmission parameters. 

As in the cases $p=1$ and $p=2$ of Theorem \ref{thm:perturbative}, the case $p=1$ of Theorem \ref{thm:harmonic} follows from the homogeneous case, where the velocity is equal to $|a_1|$, while the case $p=2$ follows from Lemma \ref{lem:warm-up}. Hence, for the remainder of this section, we assume that $p>2$.

Let $c$ be the $p$-periodic string as in \eqref{def:c-p}.
Throughout this section, we assume that
\begin{equation}
a_j\neq0,\qquad j=1,\ldots,p.
\end{equation}
The corresponding local coins are
\begin{equation}
C_j=
\begin{bmatrix}
a_j & b_j\\
-\overline{b_j} & \overline{a_j}
\end{bmatrix},
\qquad (a_j,b_j)\in\mathbb S^3.
\end{equation}
As explained in Section \ref{sec:Fourier} above, the study of the velocity of the corresponding shift-coin operator
$U_p(c)=SC(c)$ reduces to the study of the unitary $2p\times 2p$ Floquet matrix
\begin{equation}
\widehat U_p(c,\theta)=\widehat S(\theta)\widehat C(c),
\end{equation}
where $\widehat S(\theta)$ and $\widehat C(c)$ are given in
\eqref{def:S-hat} and \eqref{def:C-hat}, respectively. We present here
$\widehat S(\theta)$ for the reader's convenience:
\begin{equation}\label{def:S-hat-harmonic}
\widehat S(\theta)=\sum_{j=1}^{p}
\left(|e_{j-1}\rangle\langle e_j|\otimes P_-+|e_{j+1}\rangle\langle e_j|\otimes P_+\right), \text{ where }|e_0\rangle := e^{-i\theta}|e_p\rangle,
\qquad
|e_{p+1}\rangle := e^{i\theta}|e_1\rangle.
\end{equation}
For a fixed $p$-periodic string $c$, Theorem \ref{thm:simple-spec}(\ref{thm:spectrum:1}) shows that the spectrum of
$\widehat U_p(c,\theta)$ is simple for every $\theta \in E(c)\subset[0,2\pi)$ where
\begin{equation}
E(c):=\left\{\theta\in[0,2\pi):\ e^{-2i\theta}\neq\prod_{k=1}^p\frac{a_k}{\overline{a_k}} \right\}.
\end{equation}
Thus, for every $\theta\in E(c)$, the Floquet matrix $\widehat U_p(c,\theta)$ has simple eigenvalues
\begin{equation}\label{eq:simple-10}
\lambda_j(c,\theta)=e^{i\omega_j(c,\theta)},
\qquad j=1,\ldots,2p,
\end{equation}
where, on each connected component of $E(c)$, the eigenvalues may be labeled so that the functions $\omega_j(c,\theta)$ are real analytic in $\theta$. Consequently, by \eqref{def:v-dispersion}, the velocity is given by
\begin{equation}\label{velocity:sup-E(c)}
v(U_p(c))=p\sup_{\substack{\theta\in E(c)\\ j=1,\ldots,2p}}
|\partial_\theta\omega_j(c,\theta)|,
\end{equation}
Although the labeling in \eqref{eq:simple-10} may break down at the finite set $[0,2\pi)\setminus E(c)$, the analytic unitary family $\theta\mapsto \widehat U_p(c,\theta))$ admits local analytic eigenvalue parametrizations through these points. Consequently, after relabeling, the functions $\partial_\theta\omega_j(c,\theta)$ have continuous extensions to the points in $[0,2\pi)\setminus E(c)$, and omitting those points does not change the supremum in \eqref{def:v-dispersion}.

\subsection{A Hellmann--Feynman formula for the block velocity}
Our first step in bounding the right-hand side of \eqref{velocity:sup-E(c)} is to establish a Hellmann--Feynman type identity for the derivative of a dispersion branch.
Lemma \ref{lem:HF} expresses the dispersion relations associated with a simple Floquet eigenvalue in terms of the boundary components of the corresponding eigenvector. This formula allows us to estimate the block velocity without requiring an explicit expression for the dispersion relations.

\begin{lemma}\label{lem:HF}
Suppose that
\begin{equation}
\lambda(c,\theta):=e^{i\omega(c,\theta)}
\end{equation}
is a simple eigenvalue of $\widehat U_p(c,\theta)$ with corresponding
 eigenvector
\begin{equation}
\psi=\psi(c,\theta)=\sum_{j=1}^p e_j\otimes\psi_j,
\qquad
\psi_j=
\begin{bmatrix}
\psi_j^+\\
\psi_j^-
\end{bmatrix}.
\end{equation}
Then
\begin{equation}\label{lem:eq:main}
\partial_\theta\omega(c,\theta)=\frac{|\psi_1^+|^2-|\psi_p^-|^2}{\|\psi\|^2}.
\end{equation}
\end{lemma}

Thus the derivative of the eigenphase is the normalized net boundary current through the block, see Remark \ref{rem:current} below.
\begin{remark}
With Lemma \ref{lem:eigen-bonds},  formula \eqref{lem:eq:main} confirms, in the decoupled case, that $\partial_\theta\omega(c(0),\theta)=0$.
\end{remark}
\begin{proof}[Proof of Lemma \ref{lem:HF}]
We first recall the Hellmann--Feynman formula for simple eigenvalues of unitary
matrices. Since
\begin{equation}
\lambda(c,\theta)
=
\frac{\langle \psi,\widehat U_p(c,\theta)\psi\rangle}{\|\psi\|^2},
\end{equation}
differentiating with respect to $\theta$, and omitting all arguments for
notational simplicity, gives
\begin{align}
\partial_\theta\lambda\,\|\psi\|^2+\lambda\,\partial_\theta\|\psi\|^2&=\langle \psi',\widehat U_p\psi\rangle+\langle \psi,\widehat U_p'\psi\rangle +\langle \psi,\widehat U_p\psi'\rangle \notag\\
&=\lambda\left(\langle \psi',\psi\rangle+\langle \psi,\psi'\rangle\right)+\langle \psi,\widehat U_p'\psi\rangle.
\end{align}
Since
$\partial_\theta\|\psi\|^2=\langle \psi',\psi\rangle+\langle \psi,\psi'\rangle$,
the terms involving $\partial_\theta\|\psi\|^2$ cancel. Hence
\begin{equation}
\partial_\theta\lambda\,\|\psi\|^2=\langle \psi,\widehat U_p'\psi\rangle.
\end{equation}
Since $\lambda=e^{i\omega}$, we have $\partial_\theta\lambda=i\lambda\,\partial_\theta\omega$.
Therefore,
\begin{equation}
i\lambda\|\psi\|^2\partial_\theta\omega=\langle \psi,\widehat U_p'\psi\rangle,
\end{equation}
and hence
\begin{equation}
\partial_\theta\omega=-i\frac{\langle \psi,\widehat U_p'\widehat U_p^*\psi\rangle}{\|\psi\|^2}.
\end{equation}
Now
\begin{equation}
(\partial_\theta \widehat U_p(c,\theta))(\widehat U_p(c,\theta))^*=\widehat S'(\theta)\widehat C(c)(\widehat C(c))^*(\widehat S(\theta))^*=\widehat S'(\theta)(\widehat S(\theta))^*.
\end{equation}
Thus
\begin{equation}\label{pf:omega-prime-2}
\partial_\theta\omega(c,\theta)=-i\frac{\left\langle\psi,\widehat S'(\theta)(\widehat S(\theta))^*\psi\right\rangle}{\|\psi\|^2}.
\end{equation}

A direct calculation from \eqref{def:S-hat-harmonic} gives
\begin{align}
 \widehat S'(\theta)&=i\left(e^{i\theta}|e_1\rangle\langle e_p|\otimes P_+-e^{-i\theta}|e_p\rangle\langle e_1|\otimes P_-\right), \notag\\
(\widehat S(\theta))^*&=\sum_{j=1}^p\left(|e_j\rangle\langle e_{j-1}|\otimes P_-+|e_j\rangle\langle e_{j+1}|\otimes P_+\right).
\end{align}
Therefore,
\begin{equation}
\widehat S'(\theta)(\widehat S(\theta))^*=i\left(|e_1\rangle\langle e_1|\otimes P_+-
|e_p\rangle\langle e_p|\otimes P_-\right).
\end{equation}
Substituting this into \eqref{pf:omega-prime-2}, we obtain
\begin{equation}
\partial_\theta\omega(c,\theta)=\frac{|\psi_1^+|^2-|\psi_p^-|^2}{\|\psi\|^2}.
\end{equation}
This proves \eqref{lem:eq:main}.
\end{proof}

\subsection{A conserved quantity and the harmonic bound}

Lemma \ref{lem:HF} above reduces the estimation of the block velocity to
the control of a difference of boundary masses of the Floquet eigenvector. We
now show that this difference is in fact the value of a conserved quantity along
the period. This conservation law allows us to compare the same quantity at each site and leads directly to a harmonic mean bound on the velocity.

Observe that
\begin{equation}
\widehat S(\theta)\widehat C(c)\psi=\lambda\psi
\qquad\Longrightarrow\qquad
\widehat C(c)\psi=\lambda(\widehat S(\theta))^*\psi.
\end{equation}
Hence
\begin{equation}\label{eq:C-psi}
C_j\begin{bmatrix}\psi_j^+\\\psi_j^-\end{bmatrix}
=
\lambda\begin{bmatrix}\psi_{j+1}^+\\\psi_{j-1}^-
\end{bmatrix},
\qquad j=1,\ldots,p,
\end{equation}
with the cyclic boundary conventions
\begin{equation}
\psi_{p+1}^+=e^{-i\theta}\psi_1^+, \qquad \psi_0^-=e^{i\theta}\psi_p^-.
\end{equation}
Since $C_j$ is unitary and $|\lambda|=1$, \eqref{eq:C-psi} implies
\begin{equation}\label{eq:=LR}
|\psi_j^+|^2+|\psi_j^-|^2 = |\psi_{j+1}^+|^2+|\psi_{j-1}^-|^2
\end{equation}
Equivalently,
\begin{equation}
|\psi_j^+|^2-|\psi_{j-1}^-|^2 = |\psi_{j+1}^+|^2-|\psi_j^-|^2,  \text{ for all }j=1,\ldots,p.
\end{equation}
Thus the quantity
\begin{equation}\label{eq:conserved}
K^\psi(c,\theta) := |\psi_j^+|^2-|\psi_{j-1}^-|^2 = |\psi_{j+1}^+|^2-|\psi_j^-|^2
\end{equation}
is independent of $j$ \footnote{Compare with \eqref{def:dn} when $(\varrho,\theta)=(0,0)$.}.
In particular, by Lemma \ref{lem:HF},
\begin{equation}\label{11}
\partial_\theta\omega(c,\theta)=\frac{K^\psi(c,\theta)}{\|\psi\|^2}.
\end{equation}
\begin{remark}\label{rem:current}
The conserved quantity $K^\psi(c,\theta)$ has a natural interpretation in terms of (probability) fluxes. For a Floquet eigenvector $\psi=\psi(c,\theta)$, equation \eqref{eq:C-psi} (and its consequence \eqref{eq:=LR}) says that, after applying the coin at site $j$ and then shifting, the outgoing $(+)$ component crossing the cut from $j$ to $j+1$ is represented, up to the eigenvalue phase, by $\psi_{j+1}^+$. Thus the corresponding right-going flux is $|\psi_{j+1}^+|^2$. 
Similarly, the left-going flux through the same cut is $|\psi_j^-|^2$, as follows from \eqref{eq:C-psi} applied to site $j+1$; see Figure \ref{fig:flux-cut}.

\begin{figure}[H]
\centering
\includegraphics[width=0.25\textwidth]{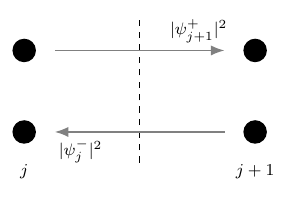}
\caption{Right- and left-going fluxes through the cut between sites $j$ and $j+1$.}
\label{fig:flux-cut}
\end{figure}
\noindent Thus, the signed difference is the net flux through the cut. This is interpreted as the current $K^{\psi}$, which is independent of the position of the cut, as explained in \eqref{eq:conserved}. With this notion, the block velocity is the normalized maximum current carried by the Floquet eigenvectors.

In the decoupled case $c=c(0)$, the equal-modulus relations of Lemma~\ref{lem:eigen-bonds} (see Figure \ref{Fig:Eigen-bones}) correspond to a vanishing current. Hence $\partial_\theta\omega(c(0),0)=0$, in agreement with the fact that the spectrum of $\widehat U_p(c(0),\theta)$ is independent of $\theta$.
\end{remark}

We have,
\begin{align}\label{K:Inner}
K^\psi(c,\theta)&=|\psi_j^+|^2-|\psi_{j-1}^-|^2\overset{(\ref{eq:C-psi})}{=}|\psi_j^+|^2-\left|-\overline{b_j}\psi_j^++\overline{a_j}\psi_j^-\right|^2\notag\\
&= \left\langle\begin{bmatrix}1 & 0\\ -\overline{b_j} & \overline{a_j}\end{bmatrix}\psi_j, \begin{bmatrix}1 & 0\\ \overline{b_j} & -\overline{a_j}\end{bmatrix}\psi_j\right\rangle \notag\\
&=\left\langle\psi_j,
\begin{bmatrix}
|a_j|^2 & \overline{a_j}b_j\\
a_j\overline{b_j} & -|a_j|^2
\end{bmatrix}\psi_j\right\rangle \notag\\
&= |a_j| \langle\psi_j, A_j \psi_j\rangle,
\end{align}
where (since  $|a_j|>0$)
\begin{equation}\label{def:Aj}
A_j := \frac{1}{|a_j|}
\begin{bmatrix}
|a_j|^2 & \overline{a_j}b_j\\
a_j\overline{b_j} & -|a_j|^2
\end{bmatrix}.
\end{equation}
Observe that $A_j=A_j^*$ and $A_j^2=\idty_2$. Hence $A_j$ is unitary with spectrum contained in $\{-1,1\}$.

To prove the harmonic bound for the velocity 
we may use \eqref{K:Inner} as follows to obtain a simple bound on $K^\psi(c,\theta)$. To obtain the refined harmonic bound in Theorem \ref{thm:harmonic}, a more involved argument is needed, see Section \ref{subsection:refined} below.

Since $\|A_j\|=1$, we have
\begin{equation}\label{Hf:11}
|K^\psi(c,\theta)| \leq |a_j|\left(|\psi_j^+|^2+|\psi_j^-|^2\right).
\end{equation}
Since $|a_j|>0$, this implies
\begin{equation}
\frac{|K^\psi(c,\theta)|}{|a_j|} \leq |\psi_j^+|^2+|\psi_j^-|^2.
\end{equation}
Summing over one period gives
\begin{equation}
|K^\psi(c,\theta)|\sum_{j=1}^p\frac{1}{|a_j|} \leq \|\psi\|^2.
\end{equation}
Substituting this estimate into Lemma \ref{lem:HF}, we obtain
\begin{equation}\label{eq:v-harmonic}
|\partial_\theta\omega(c,\theta)| \leq\left(\sum_{j=1}^p\frac{1}{|a_j|}\right)^{-1}.
\end{equation}
Consequently, \eqref{eq:v-vp} gives,
\begin{equation}
v(U_p(c)) = p\,v_{\mathrm{blk}}^{(p)}(U_p(c)) \leq p\left(\sum_{j=1}^p\frac{1}{|a_j|}\right)^{-1}.
\end{equation}
Thus the velocity is bounded by the harmonic mean of
$|a_1|,\ldots,|a_p|$. This bound is sharp in the homogeneous case, when
$|a_1|=\cdots=|a_p|$, see e.g. \cite{ARCSW25}. However, it is insensitive
to the ordering of the periodic sequence. This suggests that \eqref{eq:v-harmonic} is too coarse to capture the effect of the spatial arrangement of the coins, and motivates the search for a refined bound with the natural translation symmetry of the periodic walk.

\subsection{A refined harmonic bound}\label{subsection:refined}

We now refine the preceding estimate. Define
\begin{equation}\label{def:tau}
\tau_j:=\frac{b_j}{1+|a_j|},\qquad \alpha_j:=\frac{a_j}{|a_j|}.
\end{equation}
One checks directly from \eqref{def:Aj} that
\begin{equation}
A_j=
\begin{bmatrix}
|a_j| & \overline{\alpha_j}b_j\\
\alpha_j\overline{b_j} & -|a_j|
\end{bmatrix}=Q_j\begin{bmatrix}1&0\\0&-1\end{bmatrix}Q_j^*,
\end{equation}
where
\begin{equation}
Q_j=\frac{1}{\sqrt{2(1+|a_j|)}}\begin{bmatrix}
1+|a_j| & \overline{\alpha_j}b_j\\
\alpha_j\overline{b_j} & -(1+|a_j|)
\end{bmatrix}.
\end{equation}
Indeed, the two columns of $Q_j$ are normalized eigenvectors of $A_j$ with eigenvalues
$1$ and $-1$, respectively.

Therefore,
\begin{equation}
Q_j^*\psi_j
=
\sqrt{\frac{1+|a_j|}{2}}
\begin{bmatrix}
\psi_j^+ + \overline{\alpha_j}\tau_j\psi_j^-\\
\alpha_j\overline{\tau_j}\psi_j^+-\psi_j^-
\end{bmatrix}.
\end{equation}
Hence
\begin{align}
\langle \psi_j,A_j\psi_j\rangle&=\|\psi_j\|_{\mathbb C^2}^2-(1+|a_j|)\left|\alpha_j\overline{\tau_j}\psi_j^+-\psi_j^-\right|^2,\label{eq:Aj-positive}\\
-\langle \psi_j,A_j\psi_j\rangle&=\|\psi_j\|_{\mathbb C^2}^2-(1+|a_j|)\left|\psi_j^++\overline{\alpha_j}\tau_j\psi_j^-\right|^2.\label{eq:Aj-negative}
\end{align}

Thus, setting $K^\psi:=K^\psi(c,\theta)$, we consider two cases:
\begin{align}
K^\psi\geq0&\quad\Longrightarrow\quad\frac{|K^\psi|}{|a_j|}=\|\psi_j\|_{\mathbb C^2}^2-(1+|a_j|)\left|R_j^{(+)}\right|^2,\\
K^\psi\leq0&\quad\Longrightarrow\quad\frac{|K^\psi|}{|a_j|}=\|\psi_j\|_{\mathbb C^2}^2-(1+|a_j|)\left|R_j^{(-)}\right|^2,
\end{align}
where
\begin{equation}\label{def:R+R-}
R_j^{(+)}:=\alpha_j\overline{\tau_j}\psi_j^+-\psi_j^-,\qquad
R_j^{(-)}:=\psi_j^++\overline{\alpha_j}\tau_j\psi_j^-.
\end{equation}
If $K^\psi\neq0$, take $s:=\mathrm{sign}(K^\psi)$. Summing over one period gives
\begin{equation}
\|\psi\|^2=|K^\psi|\sum_{j=1}^p\frac{1}{|a_j|}+\sum_{j=1}^p(1+|a_j|)\left|R_j^{(s)}\right|^2.
\end{equation}
Therefore, by \eqref{11}
\begin{equation}
|\partial_\theta\omega(c,\theta)|=\left(\sum_{j=1}^p\frac{1}{|a_j|}+\frac{1}{|K^\psi|}\sum_{j=1}^p(1+|a_j|)\left|R_j^{(s)}\right|^2\right)^{-1}.
\end{equation}
If $K^\psi=0$, then $\partial_\theta\omega(c,\theta)=0$, and the upper
bounds below are trivial.

\begin{lemma}\label{lem:noVectors}
Assume $K^\psi\neq0$, and let $s:=\mathrm{sign}(K^\psi)$. Then
\begin{equation}\label{eq:noVectors}
\frac{1}{|K^\psi|}\sum_{j=1}^p(1+|a_j|)\left|R_j^{(s)}\right|^2\geq\frac{1}{4 \mu_a}\sum_{j=1}^p\left||\tau_{j+1}|-|\tau_j|\right|^2, \text{ where }\mu_a:=\max_{1\leq k\leq p}\frac{|a_k|^2}{1+|a_k|}
\end{equation}
where the indices are understood cyclically.
\end{lemma}

It follows that
\begin{equation}\label{eq:refined-bound-complex}
|\partial_\theta\omega(c,\theta)|\leq\left(\sum_{j=1}^p\frac{1}{|a_j|}+\frac{1}{4 \mu_a}\sum_{j=1}^p\left|\frac{|b_{j+1}|}{1+|a_{j+1}|}-\frac{|b_j|}{1+|a_j|}\right|^2
\right)^{-1}.
\end{equation}
Consequently,
\begin{equation}
v(U_p(c))\leq p\left(\sum_{j=1}^p\frac{1}{|a_j|}+\frac{1}{4 \mu_a}\sum_{j=1}^p\left|\frac{|b_{j+1}|}{1+|a_{j+1}|}-\frac{|b_j|}{1+|a_j|}\right|^2\right)^{-1}.
\end{equation}

\begin{proof}[Proof of Lemma \ref{lem:noVectors}]
We will later prove the following two identities:
\begin{align}
|a_{j+1}|R_{j+1}^{(+)}-\alpha_{j+1}e^{i\omega}|a_j|R_j^{(+)}&=-\alpha_{j+1}\left(\overline{\tau_{j+1}}+e^{2i\omega}\overline{\tau_j}\right)\psi_{j+1}^+,
\label{R+}\\
e^{i\omega}|a_{j+1}|R_{j+1}^{(-)}-\alpha_j |a_j|R_j^{(-)}&=\left(\tau_j+e^{2i\omega}\tau_{j+1}\right)\psi_j^-.\label{R-}
\end{align}
Taking moduli gives
\begin{align}
|a_{j+1}||R_{j+1}^{(+)}|+|a_j||R_j^{(+)}|&\geq\left|\overline{\tau_{j+1}}+e^{2i\omega}\overline{\tau_j}\right||\psi_{j+1}^+|,\notag\\
|a_{j+1}||R_{j+1}^{(-)}|+|a_j||R_j^{(-)}|&\geq\left|\tau_j+e^{2i\omega}\tau_{j+1}\right||\psi_j^-|.\label{|R|}
\end{align}
Since
\begin{equation}K^\psi=|\psi_{j+1}^+|^2-|\psi_j^-|^2,
\end{equation}
we have
\begin{align}K^\psi>0&\quad\Longrightarrow\quad|\psi_{j+1}^+|\geq \sqrt{|K^\psi|},\notag\\
K^\psi<0&\quad\Longrightarrow\quad|\psi_j^-|\geq \sqrt{|K^\psi|}.\label{Sign-K}
\end{align}
Moreover, by the reverse triangle inequality,
\begin{equation}\label{ineq:t}
\left|\overline{\tau_{j+1}}+e^{2i\omega}\overline{\tau_j}\right|\geq\bigr| |\tau_{j+1}|-|\tau_j| \bigr|,
\qquad
\left|\tau_j+e^{2i\omega} \tau_{j+1}\right|\geq \bigr| |\tau_{j+1}|-|\tau_j| \bigr|.
\end{equation}
If $K^\psi>0$, we use the first inequality in \eqref{|R|}. If
$K^\psi<0$, we use the second inequality in \eqref{|R|}. Hence, with
$s=\mathrm{sign}(K^\psi)$,
\begin{equation}
|a_{j+1}||R_{j+1}^{(s)}|+|a_j||R_j^{(s)}|\geq\bigl||\tau_{j+1}|-|\tau_j|\bigr|\sqrt{|K^\psi|}.
\end{equation}
Squaring both sides and using
\begin{equation}\label{bound}
2(A^2+B^2)\geq(A+B)^2,
\qquad
|a_j|^2\leq (1+|a_j|) \max_{1\leq k\leq p}\frac{|a_k|^2}{1+|a_k|}=: (1+|a_j|)\mu_a
\end{equation}
we get
\begin{equation}
2 \mu_a(1+|a_{j+1}|)|R_{j+1}^{(s)}|^2+2 \mu_a (1+|a_j|)|R_j^{(s)}|^2\geq \bigl||\tau_{j+1}|-|\tau_j|\bigr|^2 |K^\psi|.
\end{equation}
Summing over one period gives
\begin{equation}
\sum_{j=1}^p(1+|a_j|)|R_j^{(s)}|^2\geq\frac{|K^\psi|}{4 \mu_a }\sum_{j=1}^p\bigl||\tau_{j+1}|-|\tau_j|\bigr|^2.
\end{equation}
This proves \eqref{eq:noVectors}.

It remains to prove \eqref{R+} and \eqref{R-}.
Recall from \eqref{eq:C-psi} that
\begin{align}
e^{i\omega}\psi_{j+1}^+&=a_j\psi_j^+ + b_j\psi_j^-,\label{1}\\
e^{i\omega} \psi_j^-&=-\overline{b_{j+1}}\psi_{j+1}^++\overline{a_{j+1}}\psi_{j+1}^-,\label{2}
\end{align}
and from \eqref{def:R+R-}
\begin{equation}\label{3}
R_j^{(+)}=\alpha_j\overline{\tau_j}\psi_j^+-\psi_j^-,\qquad R_j^{(-)}=\psi_j^++\overline{\alpha_j}\tau_j\psi_j^-,
\end{equation}
where
\begin{equation}
\alpha_j=\frac{a_j}{|a_j|},
\qquad
\tau_j=\frac{b_j}{1+|a_j|}.
\end{equation}
We will use the identities
\begin{equation}\label{4}
\overline{\tau_j}b_j=\tau_j\overline{b_j}=1-|a_j|,
\qquad
b_j-|a_j|\tau_j=\tau_j,
\qquad
\overline{b_j}-|a_j|\overline{\tau_j}=\overline{\tau_j}.\\[0.5cm]
\end{equation}

\noindent\underline{Identity \eqref{R+}:}
We have (use $a_j=\alpha_j |a_j|$)
\begin{align}
|a_j| R_j^{(+)}&\overset{\eqref{3}}{=} a_j \overline{\tau_j}\psi_j^+ -|a_j|\psi_j^- \notag\\
&\overset{\eqref{1}}{=} \overline{\tau_j} (e^{i\omega}\psi_{j+1}^+- b_j\psi_j^-)-|a_j|\psi_j^- \notag\\
&\overset{\eqref{4}}{=} \overline{\tau_j}e^{i\omega}\psi_{j+1}^+-\psi_j^-. \label{R+:1}
\end{align}
Then observe that
\begin{align}
e^{i\omega} \psi_j^- & \overset{\eqref{2}}{=} -\overline{b_{j+1}}\psi_{j+1}^++\overline{a_{j+1}}\psi_{j+1}^- \notag\\
& = \left(-\overline{b_{j+1}}+|a_{j+1}|\overline{\tau_{j+1}}\right)\psi_{j+1}^+-\overline{\alpha_{j+1}}|a_{j+1}|\left(\alpha_{j+1}\overline{\tau_{j+1}}\psi_{j+1}^+-\psi_{j+1}^-\right)\notag\\
&\overset{\eqref{4}}{=} -\overline{\tau_{j+1}}\psi_{j+1}^+ - \overline{\alpha_{j+1}}|a_{j+1}|R_{j+1}^{(+)}.
\end{align}
Substitute in \eqref{R+:1}
\begin{align}
|a_j| R_j^{(+)}& =  \overline{\tau_j}e^{i\omega}\psi_{j+1}^+ +e^{-i\omega}\overline{\tau_{j+1}} \psi_{j+1}^+ + e^{-i\omega}\overline{\alpha_{j+1}}|a_{j+1}| R_{j+1}^{(+)} \notag \\
&= e^{-i\omega}\overline{\alpha_{j+1}}|a_{j+1}| R_{j+1}^{(+)} +e^{-i\omega}\left(\overline{\tau_{j+1}} + e^{2i\omega}\overline{\tau_j}\right)\psi_{j+1}^+
\end{align}
This proves \eqref{R+}.\\

\noindent\underline{Identity \eqref{R-}:} We have (use $\overline{a_{j+1}}=\overline{\alpha_{j+1}}|a_{j+1}|$)

\begin{align}
|a_{j+1}| R_{j+1}^{(-)}& \overset{\eqref{3}}{=} |a_{j+1}| \psi_{j+1}^++\overline{a_{j+1}}\tau_{j+1}\psi_{j+1}^- \notag\\
&\overset{\eqref{2}}{=} |a_{j+1}| \psi_{j+1}^+ +\tau_{j+1}\left(e^{i\omega}\psi_j^-+\overline{b_{j+1}}\psi_{j+1}^+\right) \notag\\
&\overset{\eqref{4}}{=} \psi_{j+1}^+ + e^{i\omega}\tau_{j+1} \psi_j^- \label{R-:1}
\end{align}
Then observe that
\begin{align}
e^{i\omega} \psi_{j+1}^+ &\overset{\eqref{1}}{=} a_j \psi_j^+ +b_j \psi_j^- \notag\\
& \overset{}{=} a_j \left(\psi_j^++\overline{\alpha_j}\tau_j \psi_j^-\right)+\left(b_j-|a_j|\tau_j\right)\psi_j^- \notag\\
&\overset{\eqref{4}}{=} a_j R_j^{(-)}+\tau_j\psi_j^-
\end{align}
Substitute in \eqref{R-:1} to obtain
\begin{equation}
|a_{j+1}|R_{j+1}^{(-)}=e^{-i\omega} a_j R_j^{(-)} +e^{-i\omega}\tau_j \psi_j^{-}+e^{i\omega}\tau_{j+1}\psi_j^-
\end{equation}
Arrange the terms to obtain \eqref{R-}. 

This finishes the proof.
\end{proof}


\section{Simplicity of the Floquet spectrum}\label{sec:simple-spec} 

The results of this section provide the crucial spectral simplicity input needed in
Sections \ref{sec:perturbation} and \ref{sec:harmonic} to prove our main
results, Theorems \ref{thm:perturbative} and \ref{thm:harmonic}.

Recall that the Floquet matrix is defined by
\begin{equation}\label{def:floquet}
\widehat U_p(c,\theta)=\widehat{S}(\theta)\widehat{C}(c),
\end{equation}
where $\widehat{S}(\theta)$ is given in \eqref{def:S-hat} for $\theta\in[0,2\pi)$, and $\widehat{C}(c)$ is given in \eqref{def:C-hat} with the periodic string $c$ as in \eqref{def:c-p},
and local coins
\begin{equation}\label{eq:coin-def}
C_j=
\begin{bmatrix}
a_j & b_j\\
-\overline{b_j} & \overline{a_j}
\end{bmatrix},
\qquad (a_j, b_j)\in\mathbb S^3 \text{ for }j=1,\ldots,p.
\end{equation}
We emphasize here that $c$ does not include a shorter repeating block. For example, the constant coin case is not considered here.
\subsection{The cases with at most one reflector}\label{sec:spectrum-simple-1R}

We first treat the regimes with zero or one reflector, where a reflector means a site $j$ for which $a_j=0$. In these cases the eigenvalue equation has a particularly rigid structure. In the no-reflector case, simplicity fails only at finitely many values of $\theta$. This result is needed in the proof of Theorem \ref{thm:harmonic}. When there is exactly one reflector, the spectrum is  independent of $\theta$ and simple. This is needed in the perturbative argument in the proof of Theorem \ref{thm:perturbative}. 

The following theorem summarizes the corresponding simplicity statements.

\begin{theorem}\label{thm:simple-spec}
Consider the Floquet matrix \eqref{def:floquet}.
\begin{enumerate}[(i)]
\item \label{thm:spectrum:1} (No reflectors) Assume that $a_j\neq0$ for every $j=1,\ldots,p$. 
Then, for every 
\begin{equation}\label{eq:exceptional-theta}
\theta\in E(c):=\left\{\theta\in[0,2\pi):\ e^{-2i\theta}\neq \prod_{j=1}^p\frac{a_j}{\overline{a_j}}\right\}.
\end{equation}
the Floquet matrix $\widehat U_p(c,\theta)$ has simple spectrum.
\item \label{thm:spectrum:2} (One-reflector) If $a_1=0$ and $a_j\neq0$ for all $j\neq1$, then
\begin{equation}
\mathrm{spec}(\widehat U_p(c(0),\theta))
=
\mathrm{spec}(\widehat U_p(c(0),0))
\end{equation}
is simple, where
\begin{equation}
c(0):=((0,b_1),(a_2,b_2),\ldots,(a_p,b_p)), \text{ and }|b_1|=1.
\end{equation}
\end{enumerate}
\end{theorem}
\begin{proof}
Let $\lambda\in\mathrm{spec}(\widehat U_p(c,\theta))$. Since $\widehat U_p(c,\theta)$
is unitary, $|\lambda|=1$ and, in particular, $\lambda\neq0$.

Let
\begin{equation}
\psi=\sum_{j=1}^p e_j\otimes \psi_j, \qquad
\psi_j=
\begin{bmatrix}
\psi_j^+\\
\psi_j^-
\end{bmatrix}.
\end{equation}
Let $0\neq \psi\in\ker(\lambda\idty-\widehat U_p(c,\theta))$. The equation
$\widehat U_p(c,\theta)\psi=\lambda\psi$ is equivalent to
\begin{equation}\label{eq:direct-eigen-system-vector}
C_j\psi_j=\lambda
\begin{bmatrix}
\psi_{j+1}^+\\
\psi_{j-1}^-
\end{bmatrix}
\qquad j=1,\ldots,p,
\end{equation}
where the boundary conditions are
\begin{equation}\label{eq:direct-component-boundary}
\psi_{p+1}^+=e^{-i\theta}\psi_1^+,
\qquad
\psi_0^-=e^{i\theta}\psi_p^-.
\end{equation}
Since
\begin{equation}
C_j=
\begin{bmatrix}
a_j & b_j\\
-\overline{b_j} & \overline{a_j}
\end{bmatrix},
\end{equation}
Equations \eqref{eq:direct-eigen-system-vector} read as the system
\begin{equation}\label{eq:direct-eigen-system}
\begin{aligned}
\lambda\psi_{j+1}^+
&=
a_j\psi_j^+ + b_j\psi_j^-,
\\
\lambda\psi_{j-1}^-
&=
-\overline{b_j}\psi_j^+ + \overline{a_j}\psi_j^-,
\end{aligned}
\qquad j=1,\ldots,p.
\end{equation}
For every $j$ such that $a_j\neq0$, solve the second equation in
\eqref{eq:direct-eigen-system} for $\psi_j^-$. This gives
\begin{equation}
\psi_j^-
=
\frac{1}{\overline{a_j}}
\left(\lambda\psi_{j-1}^-+\overline{b_j}\psi_j^+\right).
\end{equation}
Substituting this into the first equation in \eqref{eq:direct-eigen-system} and using
$|a_j|^2+|b_j|^2=1$, we obtain
\begin{equation}
\psi_{j+1}^+
=
\frac{1}{\overline{a_j}}
\left(\lambda^{-1}\psi_j^+ + b_j\psi_{j-1}^-\right).
\end{equation}
Thus \eqref{eq:direct-eigen-system} may be written in transfer matrix form as
\begin{equation}\label{eq:direct-transfer-relation}
\begin{bmatrix}
\psi_{j+1}^+\\
\psi_j^-
\end{bmatrix}
=
T_j(\lambda)
\begin{bmatrix}
\psi_j^+\\
\psi_{j-1}^-
\end{bmatrix},
\qquad j=1,\ldots,p,
\end{equation}
where
\begin{equation}\label{eq:direct-transfer-matrix}
T_j(\lambda)
=
\frac{1}{\overline{a_j}}
\begin{bmatrix}
\lambda^{-1} & b_j\\
\overline{b_j} & \lambda
\end{bmatrix},
\qquad
\det T_j(\lambda)=\frac{a_j}{\overline{a_j}}.
\end{equation}
Note that if $a_j=0$, then the corresponding transfer matrix $T_j(\lambda)$ is not defined.\\

\noindent \underline{Case \eqref{thm:spectrum:1}:}
Here, $a_j\neq0$ for all $j$, hence all transfer matrices $T_j(\lambda)$ for
$j=1,\ldots,p$ are defined as in \eqref{eq:direct-transfer-matrix}. Define the
\emph{monodromy} matrix
\begin{equation}\label{def:mono}
M(\lambda)=T_p(\lambda)T_{p-1}(\lambda)\cdots T_1(\lambda).
\end{equation}
Iterating \eqref{eq:direct-transfer-relation} gives
\begin{equation}\label{eq:direct-monodromy-action}
M(\lambda)
\begin{bmatrix}
\psi_1^+\\
\psi_0^-
\end{bmatrix}
=
\begin{bmatrix}
\psi_{p+1}^+\\
\psi_p^-
\end{bmatrix}
\overset{(\ref{eq:direct-component-boundary})}{=}
e^{-i\theta}
\begin{bmatrix}
\psi_1^+\\
\psi_0^-
\end{bmatrix}.
\end{equation}
Now observe that the recursion \eqref{eq:direct-transfer-relation} and the fact that
$\psi\neq0$ imply that
\begin{equation}
\begin{bmatrix}
\psi_1^+\\
\psi_0^-
\end{bmatrix}
\neq0.
\end{equation}
Thus $e^{-i\theta}$ is an eigenvalue of $M(\lambda)$.

Moreover, \eqref{eq:direct-transfer-relation} implies that the map
\begin{equation}\label{eq:direct-isomorphism-map}
\psi\longmapsto
\begin{bmatrix}
\psi_1^+\\
\psi_0^-
\end{bmatrix}
\end{equation}
is an isomorphism from $\ker(\lambda\idty-\widehat U_p(c,\theta))$ onto
$\ker(M(\lambda)-e^{-i\theta}\idty_2)$. Thus,
\begin{equation}\label{eigen-U-M}
\lambda\in\mathrm{spec}(\widehat U_p(c,\theta)) \ \ \iff e^{-i\theta}\in\mathrm{spec}(M(\lambda)).
\end{equation}
Consequently,
\begin{equation}\label{eq:direct-nullity-monodromy}
\dim\ker(\lambda\idty-\widehat U_p(c,\theta))
=
\dim\ker(M(\lambda)-e^{-i\theta}\idty_2).
\end{equation}

By \eqref{def:mono} and \eqref{eq:direct-transfer-matrix}, we have
\begin{equation}
\det M(\lambda)=\prod_{j=1}^p\frac{a_j}{\overline{a_j}}=:\Delta(c).
\end{equation}
Since $\lambda_1=e^{-i\theta}$ is one eigenvalue of $M(\lambda)$, the other eigenvalue is
$\lambda_2=\Delta(c)e^{i\theta}$. We have
\begin{equation}
\lambda_1\neq \lambda_2 \iff e^{-i\theta}\neq \Delta(c)e^{i\theta} \iff e^{-2i\theta}\neq\Delta(c).
\end{equation}
That is, under the condition $e^{-2i\theta}\neq\Delta(c)$, $\lambda_1=e^{-i\theta}$ is a simple eigenvalue of the $2\times2$ matrix
$M(\lambda)$. Hence
\begin{equation}
\dim\ker(M(\lambda)-e^{-i\theta}\idty_2)=1,
\end{equation}
and we conclude that
\begin{equation}\label{end}
\dim\ker(\lambda\idty-\widehat U_p(c,\theta))=1.
\end{equation}
Since $\widehat U_p(c,\theta)$ is unitary, \eqref{end} shows that every eigenvalue of
$\widehat U_p(c,\theta)$ is algebraically simple whenever
$\theta\in[0,2\pi)$ satisfies \eqref{eq:exceptional-theta}.

\begin{remark}\label{rem:spectrum-real}
In the real case considered previously, one has $\Delta(c)=1$, and the condition
\begin{equation}
e^{-2i\theta}=\Delta(c)
\end{equation}
reduces to $\theta=0$ or $\theta=\pi$.
\end{remark}

\begin{remark}\label{spec:symmetry}
We note here that if $b_j\in\mathbb R$ for all $j$, then $\mathrm{spec}(\widehat U_p(c,\theta))$ is invariant under complex conjugation. This follows from the following argument.

Consider the reflection $R:=\begin{bmatrix}0 & 1\\ 1 & 0\end{bmatrix}$. Since $b_j\in\mathbb R$, then the transfer matrices \eqref{eq:direct-transfer-matrix} satisfy
\begin{equation}
T_j(\lambda^{-1})= R T_j(\lambda) R\qquad \Longrightarrow \qquad M(\lambda^{-1})= T_p(\lambda^{-1})\ldots T_1(\lambda^{-1})= R M(\lambda)R.
\end{equation}
Since $R^2=\idty_2$, then the matrices $M(\lambda)$ and $M(\lambda^{-1})$ are similar. Hence $\mathrm{spec}(M(\lambda))=\mathrm{spec}(M(\lambda^{-1}))$. Then
\begin{equation}
\lambda\in \mathrm{spec}(\widehat U_p(c,\theta)) \overset{\eqref{eigen-U-M}}{\iff} e^{-i\theta}\in\mathrm{spec}(M(\lambda))\iff e^{-i\theta}\in\mathrm{spec}(M(\lambda^{-1})) \overset{\eqref{eigen-U-M}}{\iff}  \lambda^{-1}\in \mathrm{spec}(\widehat U_p(c,\theta)).
\end{equation}
Since $\widehat U_p(c,\theta)$ is unitary, then $e^{i\omega(c,\theta)}\in \mathrm{spec}(\widehat U_p(c,\theta))$ if and only if $e^{-i\omega(c,\theta)}\in \mathrm{spec}(\widehat U_p(c,\theta))$.\\
\end{remark}

\noindent \underline{Case \eqref{thm:spectrum:2}:}
Assume that $a_1=0$ and $a_j\neq0$ for $j=2,\ldots,p$. Since
$|a_1|^2+|b_1|^2=1$, we have $|b_1|=1$. 
That
\begin{equation}
\mathrm{spec}(\widehat U_p(c(0),\theta))=\mathrm{spec}(\widehat U_p(c(0),0)).
\end{equation}
Follows from the gauge invariance \eqref{U-gauge}. It remains to prove that $\mathrm{spec}(\widehat U_p(c(0),0))$ is simple.

The relations \eqref{eq:direct-eigen-system} read for $j=1$ as
\begin{equation}\label{eq:direct-reflector-equations}
\lambda\psi_2^+=b_1\psi_1^-,\qquad \lambda\psi_0^-=-\overline{b_1}\psi_1^+.
\end{equation}
Using the boundary conditions \eqref{eq:direct-component-boundary}, the second equation in
\eqref{eq:direct-reflector-equations} gives
\begin{equation}
\psi_p^-= -\lambda^{-1}\overline{b_1}e^{-i\theta}\psi_1^+,
\end{equation}
and let's recall the first boundary condition $\psi_{p+1}^+=e^{-i\theta}\psi_1^+$. Hence
\begin{equation}\label{eq:direct-right-boundary-reflector}
\begin{bmatrix}\psi_{p+1}^+\\ \psi_p^- \end{bmatrix}
=
e^{-i\theta}\psi_1^+
\begin{bmatrix}1\\ -\lambda^{-1}\overline{b_1}\end{bmatrix}.
\end{equation}
On the other hand, the first equation in \eqref{eq:direct-reflector-equations} gives
\begin{equation}\label{eq:direct-left-boundary-reflector}
\begin{bmatrix}
\psi_2^+\\
\psi_1^-
\end{bmatrix}
=
\psi_1^-
\begin{bmatrix}
\lambda^{-1}b_1\\
1
\end{bmatrix}.
\end{equation}
For $j=2,\ldots,p$, all transfer matrices are defined, and therefore
\begin{equation}\label{eq:direct-reflector-monodromy}
T_p(\lambda)T_{p-1}(\lambda)\cdots T_2(\lambda)
\begin{bmatrix}
\psi_2^+\\
\psi_1^-
\end{bmatrix}
=
\begin{bmatrix}
\psi_{p+1}^+\\
\psi_p^-
\end{bmatrix}.
\end{equation}
 Let $0\neq\psi\in\ker(\lambda\idty-\widehat U_p(c(0),\theta))$. We first claim that $\psi_1^-\neq0$. Indeed, if $\psi_1^-=0$, then
\eqref{eq:direct-left-boundary-reflector} gives
\begin{equation}
\begin{bmatrix} \psi_2^+\\ \psi_1^- \end{bmatrix} =0.
\end{equation}
Since the matrices $T_j(\lambda)$ are invertible for $j=2,\ldots,p$,
\eqref{eq:direct-reflector-monodromy} gives
\begin{equation}
\begin{bmatrix} \psi_{p+1}^+\\ \psi_p^- \end{bmatrix} =0.
\end{equation}
Then \eqref{eq:direct-right-boundary-reflector} gives $\psi_1^+=0$, and the transfer recursion yields
$\psi=0$, a contradiction. Thus $\psi_1^-\neq0$.

Equations \eqref{eq:direct-left-boundary-reflector} and \eqref{eq:direct-reflector-monodromy} show that every vector in
$\ker(\lambda\idty-\widehat U_p(c(0),\theta))$ is determined uniquely by the single scalar
$\psi_1^-$. Therefore
\begin{equation}
\dim\ker(\lambda\idty-\widehat U_p(c(0),\theta))\leq1.
\end{equation}
Since $\lambda$ is an eigenvalue, the kernel is nontrivial, and hence
\begin{equation}
\dim\ker(\lambda\idty-\widehat U_p(c(0),\theta))
=
1.
\end{equation}
This proves that the spectrum of $\widehat U_p(c(0),0)$, and hence of
$\widehat U_p(c(0),\theta)$ for every $\theta$, is simple.
\end{proof}

\subsection{The multiple-reflector case}\label{sec:simple-spec-case3}

To simplify the presentation, we restrict our analysis in this section to real coins. Thus
\begin{equation}
c_j=(a_j,b_j),\ a_j\in[-1,1], \text{ and }b_j=\sqrt{1-a_j^2}.
\end{equation}
At $\theta=0$, the Floquet matrix acts on the standard basis by (see \eqref{def:Up-hat})
\begin{equation}\label{eq:real-action}
\widehat U_p(c,0)\delta_j^+=a_j\delta_{j+1}^+-b_j\delta_{j-1}^-,\qquad\widehat U_p(c,0)\delta_j^-=b_j\delta_{j+1}^+ + a_j\delta_{j-1}^-,
\end{equation}
where the indices are understood periodically.

Let $T\subsetneq \{1,\ldots,p\}$ be the set of perfect reflectors. In the
multiple-reflector case, $|T|\ge 2$. After a cyclic relabeling of the
period, we may assume that $1\in T$, and we write
\begin{equation}
T=\{j_1=1,j_2,\ldots,j_n\}, \qquad 1=j_1<j_2<\cdots<j_n\le p.
\end{equation}
For simplicity, we use the real reflector convention $(c_{\perp T})_\ell=(0,1)$, for $\ell\in T$.
Thus
\begin{equation}
(c_{\perp T})_\ell=
\begin{cases}
(0,1), & \ell\in T,\\
(a_\ell,b_\ell),\ a_\ell\neq 0,& \ell\notin T .
\end{cases}
\end{equation}
We write the elements of $T$ in cyclic order and set
\begin{equation}
j_{n+1}:=j_1+p=p+1 .
\end{equation}
For $k=1,\ldots,n$, define
\begin{equation}
I_k:=\{j_k,j_k+1,\ldots,j_{k+1}-1\},\qquad L_k:=|I_k|=j_{k+1}-j_k .
\end{equation}
Then
\begin{equation}\label{dec:T}
\{1,\ldots,p\}=\bigcup_{k=1}^n I_k .
\end{equation}
The string on the $k$-th interval is
$c^{(k)}:=\left.(c_{\perp T})\right|_{I_k}$.
Equivalently,
\begin{equation}\label{eq:ck-real}
(c^{(k)})_1=(0,1),\qquad (c^{(k)})_m=(a_{j_k+m-1},b_{j_k+m-1}),\qquad m=2,\ldots,L_k . \footnote{When $L_k=1$, this simply means $c^{(k)}=((0,1))$.}
\end{equation}
Hence each
$c^{(k)}$ is a one-reflector string.

We also introduce the subspace
\begin{equation}\label{def:Hk}
\mathcal H_k
:=
\operatorname{span}
\left\{
\{\delta_{j_k}^-,\delta_{j_{k+1}}^+\}
\cup
\{\delta_j^+,\delta_j^-: j_k<j<j_{k+1}\}
\right\},
\end{equation}
with the periodic convention $\delta_{j_{n+1}}^+=\delta_{p+1}^+=\delta_1^+ $.
Then $\dim \mathcal H_k=2L_k$, and \eqref{dec:T} induces the orthogonal
decomposition
\begin{equation}\label{dec:Hk}
\mathbb C^{2p}
=
\bigoplus_{k=1}^n \mathcal H_k .
\end{equation}
Notice that the two spin components at a reflecting site are assigned to
adjacent blocks: $\delta_{j_k}^-$ belongs to $\mathcal H_k$, while
$\delta_{j_k}^+$ belongs to $\mathcal H_{k-1}$, with cyclic indexing.

\begin{theorem}[Multiple reflectors]\label{thm:spectrum:3}
In the multiple reflectors setting above, we have
\begin{equation}\label{eq:block-dec-multiple}
\widehat U_p(c_{\perp T},0)
\equiv
\bigoplus_{k=1}^n \widehat U_{L_k}(c^{(k)},0).
\end{equation}
Here each $\widehat U_{L_k}(c^{(k)},0)$ is a one-reflector Floquet matrix
as in Theorem \ref{thm:simple-spec}(\ref{thm:spectrum:2}). In particular, each block
$\widehat U_{L_k}(c^{(k)},0)$ has simple spectrum.

Consequently, the full matrix $\widehat U_p(c_{\perp T},0)$ has simple
spectrum if and only if
\begin{equation}\label{eq:overlap}
\operatorname{spec}\bigl(\widehat U_{L_k}(c^{(k)},0)\bigr)
\cap
\operatorname{spec}\bigl(\widehat U_{L_\ell}(c^{(\ell)},0)\bigr)
=
\varnothing
\qquad
\text{for all } k\neq \ell .
\end{equation}
\end{theorem}

\begin{remark}
Although each one-reflector block has simple spectrum, spectra of
different blocks need not be disjoint. Hence the full decoupled Floquet
matrix has simple spectrum only under the additional non-overlap condition
\eqref{eq:overlap}. This condition is generic, but it is not automatic.
\end{remark}

\begin{proof}
Let $\{\mathfrak e_m^\pm\}_{m=1}^{L_k}$ denote the standard basis of
$\mathbb C^{2L_k}$, with the cyclic convention
\begin{equation}
\mathfrak e_{L_k+1}^+=\mathfrak e_1^+,
\qquad
\mathfrak e_0^-=\mathfrak e_{L_k}^- .
\end{equation}
For the proof, set
\begin{equation}
V_k:=\widehat U_{L_k}(c^{(k)},0).
\end{equation}
Define a unitary map $W_k:\mathbb C^{2L_k}\to\mathcal H_k$ by
\begin{equation}\label{eq:Wk}
W_k\mathfrak e_1^+ := \delta_{j_{k+1}}^+,
\qquad
W_k\mathfrak e_1^- := \delta_{j_k}^-,
\end{equation}
and, for $m=2,\ldots,L_k$,
\begin{equation}\label{eq:W-action}
W_k\mathfrak e_m^\pm := \delta_{j_k+m-1}^\pm .
\end{equation}
Thus the basis of $\mathbb C^{2L_k}$ is mapped directly onto the natural
basis of the block $\mathcal H_k$.

We now verify the intertwining relation
\begin{equation}\label{eq:intertwining}
\widehat U_p(c_{\perp T},0)W_k=W_k V_k .
\end{equation}
Since $(c^{(k)})_1=(0,1)$, the one-reflector matrix $V_k$ satisfies, by \eqref{eq:real-action}
\begin{equation}\label{V-action-1}
V_k\mathfrak e_1^+=-\mathfrak e_{L_k}^-,\qquad V_k\mathfrak e_1^-=\mathfrak e_2^+,
\end{equation}
where $\mathfrak e_2^+=\mathfrak e_1^+$ if $L_k=1$. On the other hand,
using \eqref{eq:real-action} and the fact that both $j_k$ and
$j_{k+1}$ are reflectors, we get
\begin{equation}
\widehat U_p(c_{\perp T},0)W_k\mathfrak e_1^+ \overset{\eqref{eq:Wk}}{=}\widehat U_p(c_{\perp T},0)\delta_{j_{k+1}}^+\overset{\eqref{eq:real-action}}{=}-\delta_{j_{k+1}-1}^-\overset{\eqref{eq:W-action}}{=}W_k(-\mathfrak e_{L_k}^-)\overset{\eqref{V-action-1}}{=} W_k V_k \mathfrak e_1^+,\\
\end{equation}
Similarly, 
\begin{equation}
\widehat U_p(c_{\perp T},0)W_k\mathfrak e_1^- =\widehat U_p(c_{\perp T},0)\delta_{j_k}^-=\delta_{j_k+1}^+=W_k\mathfrak e_2^+ =W_k V_k \mathfrak e_1^-.
\end{equation}
Thus \eqref{eq:intertwining} holds on the reflecting site of the
one-reflector block.

It remains to check the non-reflecting sites when $L_k>1$. Let $2\le m\le L_k$ and put
$j=j_k+m-1$.
Then $j\notin T$, and by \eqref{eq:ck-real}, $(c^{(k)})_m=(a_j,b_j)$.
Therefore
\begin{equation}
\widehat U_p(c_{\perp T},0)W_k\mathfrak e_m^+
\overset{\eqref{eq:W-action}}{=} \widehat U_p(c_{\perp T},0)\delta_j^+ 
 \overset{\eqref{eq:real-action}}{=}
a_j\delta_{j+1}^+-b_j\delta_{j-1}^-  
 \overset{\eqref{eq:W-action}}{=}
W_k\bigl(a_j\mathfrak e_{m+1}^+
        -b_j\mathfrak e_{m-1}^-\bigr)  
\overset{\eqref{eq:real-action}}{=}
W_kV_k\mathfrak e_m^+ ,
\end{equation}
and similarly
\begin{equation}
\widehat U_p(c_{\perp T},0)W_k\mathfrak e_m^-=\widehat U_p(c_{\perp T},0)\delta_j^- =b_j\delta_{j+1}^+ + a_j\delta_{j-1}^-  =W_k\bigl(b_j\mathfrak e_{m+1}^++a_j\mathfrak e_{m-1}^-\bigr)  =W_kV_k\mathfrak e_m^- .
\end{equation}
Hence \eqref{eq:intertwining} holds on the whole basis.

Now define
\begin{equation}
W:=\bigoplus_{k=1}^n W_k:
\bigoplus_{k=1}^n \mathbb C^{2L_k}
\longrightarrow
\mathbb C^{2p}.
\end{equation}
Since the ranges of the maps $W_k$ are precisely the mutually orthogonal
subspaces $\mathcal H_k$, the map $W$ is unitary. From
\eqref{eq:intertwining}, we obtain
\begin{equation}
W^*\widehat U_p(c_{\perp T},0)W
=
\bigoplus_{k=1}^n V_k
=
\bigoplus_{k=1}^n \widehat U_{L_k}(c^{(k)},0).
\end{equation}
This proves the block decomposition \eqref{eq:block-dec-multiple}.

By Theorem \ref{thm:simple-spec}, each one-reflector block
$\widehat U_{L_k}(c^{(k)},0)$ has simple spectrum. Therefore
$\widehat U_p(c_{\perp T},0)$ has simple spectrum if and only if no
eigenvalue belongs to two distinct one-reflector blocks. Equivalently,
\begin{equation}
\operatorname{spec}\bigl(\widehat U_{L_k}(c^{(k)},0)\bigr)
\cap
\operatorname{spec}\bigl(\widehat U_{L_\ell}(c^{(\ell)},0)\bigr)
=
\varnothing
\qquad
\text{for all } k\neq \ell .
\end{equation}
This is \eqref{eq:overlap}, and the proof is complete.
\end{proof}

\appendix

\section{The two-periodic velocity formula}\label{sec:warm-up}
We now prove Lemma \ref{lem:warm-up}. Let
\begin{equation}
c=((a_1,b_1),(a_2,b_2)),\quad (a_j,b_j)\in\mathbb S^3, \text{ and }
C_j=
\begin{bmatrix}
a_j & b_j\\
-\overline{b_j} & \overline{a_j}
\end{bmatrix} \text{ for }j=1,2.
\end{equation}
From \eqref{def:U2}, we have
\begin{equation}
\left(\widehat U_2(c,\theta)\right)^2=
\begin{bmatrix}
M(c,\theta) & 0_{2\times 2}\\
0_{2\times 2} & \widetilde M(c,\theta)
\end{bmatrix},
\end{equation}
where
\begin{equation}\label{eq:M-p2-complex}
M(c,\theta)=
\begin{bmatrix}
e^{i\theta}a_1a_2-b_2\overline{b_1}& e^{i\theta}b_1a_2+\overline{a_1}b_2 \\
-a_1\overline{b_2}-e^{-i\theta}\overline{b_1}\,\overline{a_2} & -b_1\overline{b_2}+e^{-i\theta}\overline{a_1}\,\overline{a_2}
\end{bmatrix}.
\end{equation}
Moreover, $\widetilde M(c,\theta)$ is similar to $M(c,\theta)$. Hence
$(\widehat U_2(c,\theta))^2$ has the same two eigenvalues as $M(c,\theta)$,
each with multiplicity two.

Since $M(c,\theta)$ is unitary and $\det M(c,\theta)=1$,
its two eigenvalues are of the form
\begin{equation}
\lambda_\pm(c,\theta)=e^{\pm i\omega_2(c,\theta)}.
\end{equation}
By \cite[Lemma 3.8]{ARCSW25},
\begin{equation}\label{1234}
v(U_2(c))=\frac{1}{2}v\left((U_2(c))^2\right)\overset{\eqref{def:v-dispersion}}{=}v_{\mathrm{blk}}^{(2)}\left((U_2(c))^2\right)=\sup_{\theta\in[0,2\pi)}|\partial_\theta\omega_2(c,\theta)|.
\end{equation}

Using
\begin{equation}
\mathrm{Tr}\,M(c,\theta)=\lambda_+(c,\theta)+\lambda_-(c,\theta)=2\cos(\omega_2(c,\theta)),
\end{equation}
we obtain
\begin{equation}\label{eq:cos-omega-complex}
\cos(\omega_2(c,\theta))=\Re(e^{i\theta}a_1a_2)-\Re(b_1\overline{b_2}).
\end{equation}
We note that the case $|a_1a_2|=0$ yields zero velocity, as this corresponds to the existence of a (bi-infinite) sequence of perfect reflectors, see \cite[Lemma 2.2]{ARJS26}.  This can be seen also from $M(c,\theta)$ in \eqref{eq:M-p2-complex} above. If $a_2=0$, then observe that $M(c,\theta)=M(c,0)$, and if $a_1=0$ then it is direct to see that
\begin{equation}
M(c,\theta)=\begin{bmatrix} e^{i\theta} & 0\\ 0& 1\end{bmatrix} M(c,0)\begin{bmatrix} e^{-i\theta} & 0\\ 0& 1\end{bmatrix}.
\end{equation}
That is $\mathrm{spec}(M(c,\theta))=\mathrm{spec}(M(c,0))$. Hence, the eigenvalues of $M(c,\theta)$ are independent of $\theta$ and the velocity is zero.

Moreover, the extreme case of $|a_1|=|a_2|=1$ corresponds to the trivial perfect transmitting walk with velocity 1. This can be seen also from $M(c,\theta)$ in \eqref{eq:M-p2-complex} above, where in this case we have $b_1=b_2=0$ and hence
\begin{equation}
M(c,\theta)=\begin{bmatrix} e^{i\theta} a_1 a_2 &0 \\ 0& \overline{e^{i\theta}  a_1 a_2}\end{bmatrix}
\end{equation}
which gives by \eqref{1234}, $v(U_2(c))=1$. Thus, we assume in the following that
\begin{equation}
0<|a_1 a_2|<1.
\end{equation} 

Writing $a_1a_2=|a_1a_2|e^{i\phi}$,  \eqref{eq:cos-omega-complex} becomes
\begin{equation}
\cos(\omega_2(c,\theta))=|a_1a_2|\cos(\theta+\phi)-\Re(b_1\overline{b_2}).
\end{equation}
Therefore, away from the points where the denominator vanishes,
\begin{equation}\label{eq:omega-derivative-complex}
|\partial_\theta\omega_2(c,\theta)|^2=\frac{|a_1a_2|^2\sin^2(\theta+\phi)}{1-\left(|a_1a_2|\cos(\theta+\phi)-\Re(b_1\overline{b_2})\right)^2}.
\end{equation}

Put $x:=\cos(\theta+\phi)$, and define
\begin{equation}\label{def:F}
F(x):=\frac{|a_1a_2|^2(1-x^2)}{1-\left(|a_1a_2|x-\Re(b_1\overline{b_2})\right)^2},\qquad x\in D_F,
\end{equation}
where
\begin{equation}\label{def:Df}
D_F:=\{x\in [-1,1];\ |a_1a_2|x-\Re(b_1\overline{b_2})\neq \pm1\}.
\end{equation}
Thus,
\begin{equation}\label{max-F}
\left(v(U_2(c))\right)^2=\sup_{x \in D_F}F(x)=:\mu.
\end{equation}

\noindent \underline{Case 1: $b_1\overline{b_2}\notin\mathbb R$:} In this case, we will show that 
\begin{equation}\label{p2-case1}
\left(v(U_2(c))\right)^2=\mu< |a_1|^2
\end{equation}
and since $F$ is symmetric in $|a_1|$ and $|a_2|$ then \eqref{p2-case1} shows that $v(U_2(c))<\min\{|a_1|,|a_2|\}$, as desired.

First observe that if $b_1\overline{b_2}\notin\mathbb R$ then $F$ is defined for all $x\in[-1,1]$. This follows because $|\Re(b_1\overline{b_2})|<|b_1 b_2|$, then by Cauchy-Schwarz (and recall that $|a_j|^2+|b_j|^2=1$)
\begin{equation}
\left||a_1a_2|x-\Re(b_1\overline{b_2})\right| < |a_1a_2|+|b_1b_2|\leq \sqrt{|a_1|^2+|b_1|^2}\sqrt{|a_2|^2+|b_2|^2}=1,
\end{equation}
Proving \eqref{p2-case1} is equivalent to proving that
\begin{equation}
|a_1a_2|^2(1-x^2)<\left(1-\left(|a_1a_2|x-\Re(b_1\overline{b_2})\right)^2\right)|a_1|^2
\end{equation}
Divide by $|a_1|^2>0$ to see that we need to prove that
\begin{equation}\label{p2-case1-q>0}
q(x):=1-\left(|a_1a_2|x-\Re(b_1\overline{b_2})\right)^2 - |a_2|^2(1-x^2)\ >0 \text{ for all }x\in[-1,1].
\end{equation}
Write $q(x)$ as a quadratic polynomial in $x$,
\begin{equation}
q(x)= |a_2|^2(1-|a_1|^2)\ x^2 + 2 |a_1 a_2|\Re(b_1\overline{b_2})\ x+ 1-\Re(b_1\overline{b_2})^2-|a_2|^2
\end{equation}
Observe first that 
\begin{equation}
q(0)=1-\Re(b_1\overline{b_2})^2-|a_2|^2 > 1-(1-|a_1|^2)(1-|a_2|^2)-|a_2|^2=|a_1|^2 (1-|a_2|^2)>0.
\end{equation}
Thus to prove that $q(x)>0$ for all $x$, it is enough to prove that its discriminant is negative. The latter is given as (after dividing by $4|a_2|^2$)
\begin{equation}
|a_1|^2+|a_2|^2-|a_1 a_2|^2+\Re(b_1\overline{b_2})^2-1
\end{equation}
which is negative because $\Re(b_1\overline{b_2})^2 < (1-|a_1|^2)(1-|a_2|^2)$, and note that this is a strict inequality because  $b_1\overline{b_2}\notin\mathbb R$.

This proves \eqref{p2-case1} and hence, we proved that $v(U_2(c))<\min\{|a_1|,|a_2|\}$.\\

\noindent \underline{Case 2: $b_1\overline{b_2}\in\mathbb R$:} 
Note that if $|a_1|\neq |a_2|$ then the denominator of $F$ in \eqref{def:F} is not zero, as 
\begin{equation}
\left||a_1a_2|x-\Re(b_1\overline{b_2})\right|\leq |a_1a_2|+|b_1b_2| \text{ for all }x\in[-1,1]
\end{equation}
Then, by Cauchy–Schwarz
\begin{equation} 
|a_1a_2|+|b_1b_2| \leq 1,
\end{equation}
and the last inequality is strict if and  if $|a_1|\neq |a_2|$. Hence, if $|a_1|\neq |a_2|$ then 
\begin{equation}\label{D:notzero}
\left||a_1a_2|x-\Re(b_1\overline{b_2})\right|<1 \text{ for all }x\in[-1,1]
\end{equation}

If $|a_1|=|a_2|$, then $|b_1|=|b_2|$. Since $b_1\overline{b_2}\in\mathbb R$, the relative phase between $b_1$ and $b_2$ is either $0$ or $\pi$. Hence, in this case, either $b_1=b_2$ or $b_1=-b_2$. Therefore, in Case 2, we consider the three possibilities
\begin{equation}
(1)\ |a_1|=|a_2|,\ b_1=b_2\qquad
(2)\ |a_1|=|a_2|,\ b_1=-b_2;\quad \text{ and }\quad
(3)\ |a_1|\neq |a_2|.
\end{equation}
\begin{enumerate}[(1)]

\item \underline{$|a_1|=|a_2|$ and $b_1=b_2$}\\
In this case $F$ simplifies to
\begin{equation}
F(x)=\frac{|a_1|^2(1-x)}{2-|a_1|^2(1+x)},\quad x\in(-1,1]
\end{equation}
Note that $F$ is decreasing 
\begin{equation}
\lim_{x\to -1^+}F(x)=|a_1|^2 \quad \text{ and } F(1)=0.
\end{equation}
Thus 
\begin{equation}
v(U_2(c))=(\sup F)^{1/2}=|a_1|.
\end{equation}

\item \underline{$|a_1|=|a_2|$ and $b_1=-b_2$}\\
In this case $F$ simplifies to the increasing function
\begin{equation}
F(x)=\frac{|a_1|^2(1+x)}{2-|a_1|^2(1-x)},\quad x\in[-1,1)
\end{equation}
Then
\begin{equation}
F(-1)=0 \text{ and }\lim_{x\to 1^-}F(x)=|a_1|^2 \quad \Rightarrow v(U_2(c))=(\sup F)^{1/2}=|a_1|.
\end{equation}

\item  \underline{$|a_1|\neq|a_2|$}\\
Here we discuss a direct case separately. Namely, if the maximum of $|a_1|$ and $|a_2|$ is equal to 1, then $F$ in \eqref{def:F} reads as (assuming that $1=|a_1|>|a_2|\neq 0$)
\begin{equation}
F(x)=\frac{|a_2|^2(1-x^2)}{1-|a_2|^2x^2}
\end{equation}
which attains its maximum at $x=0$, and $\sup F=|a_2|^2$. Hence $v(U_2(c))=|a_2|=\min\{|a_1|,|a_2|\}$. The case $1=|a_2|>|a_1|\neq 0$ is identical by symmetry.

In the following we assume that $0< |a_1|,|a_2| <1$.\\
By \eqref{def:F} and \eqref{D:notzero}, we have $F(\pm1)=0$. Then the maximum in \eqref{max-F} is attained at an interior point of $[-1,1]$. At an interior maximizer $x_0\in (-1,1)$, we have
\begin{equation}
F(x_0)=\mu,\qquad F'(x_0)=0.
\end{equation}
Equivalently, the level equation
\begin{equation}\label{eq:level}
F(x)=\mu
\end{equation}
has a repeated root (root of multiplicity 2). Indeed, \eqref{eq:level} reads as the quadratic equation (in $x$)
\begin{equation}
\mu\left(1-\left(|a_1a_2|x-\Re(b_1\overline{b_2})\right)^2\right)=|a_1a_2|^2(1-x^2).
\end{equation}
Expand to obtain the quadratic equation in $x$,
\begin{equation}\label{Warmup:x2}
|a_1a_2|^2(1-\mu)x^2+2\mu |a_1a_2|\Re(b_1\overline{b_2})x+\mu\left(1-\Re(b_1\overline{b_2})^2\right)-|a_1a_2|^2=0.
\end{equation}
The repeated root condition is that the discriminant vanishes. After dividing by
$4|a_1a_2|^2$, this gives
\begin{equation}
\mu^2\Re(b_1\overline{b_2})^2-(1-\mu)\left(\mu\left(1-\Re(b_1\overline{b_2})^2\right)-|a_1a_2|^2\right)=0.
\end{equation}

Rearrange the terms,
\begin{equation}
\mu^2-B\mu+|a_1a_2|^2=0 \quad \text{ where }\quad B:=1-\Re(b_1\overline{b_2})^2+|a_1a_2|^2.
\end{equation}
Since $b_1\overline{b_2}\in\mathbb R$ then
\begin{equation}
\Re(b_1\overline{b_2})^2=(1-|a_1|^2)(1-|a_2|^2) \quad\text{ and hence }\quad B=|a_1|^2+|a_2|^2.
\end{equation}
This gives the solutions
\begin{align}
\mu_\pm &= \frac{B\pm\sqrt{B^2-4|a_1a_2|^2}}{2} \notag\\
&= \frac{1}{2}\left(|a_1|^2+|a_2|^2\pm \left||a_1|^2-|a_2|^2\right|\right).
\end{align}
Thus,
\begin{equation}
\mu_+=\max\{|a_1|^2, |a_2|^2\}\neq 1\quad \text{ and }\quad \mu_-=\min\{|a_1|^2, |a_2|^2\}.
\end{equation}
We show that in this case $\mu_+$ is rejected. For any root $\mu$, the corresponding repeated solution is, from \eqref{Warmup:x2},
\begin{equation}
x_\mu=-\frac{\mu b_1 \overline{b_2}}{|a_1a_2|(1-\mu)}.
\end{equation}
Assume first that $\mu_+=|a_1|^2$, i.e., $1>|a_1|> |a_2|>0$. Then
\begin{equation}
|x_{\mu_+}|^2=\frac{|a_1|^2(1-|a_2|^2)}{|a_2|^2(1-|a_1|^2)} > 1.
\end{equation}
Hence $x_{\mu_+}\notin[-1,1]$, so $\mu_+$ is not admissible. The case
$\mu_+=|a_2|^2$, i.e., $|a_2|>|a_1|$, is identical and gives again
$|x_{\mu_+}|^2>1$. Therefore the admissible root is
$\mu=\mu_-=\min\{|a_1|^2,|a_2|^2\}$.
Consequently,
\begin{equation}
v(U_2(c))=\min\{|a_1|,|a_2|\}.
\end{equation}
\end{enumerate}

This finishes the proof of Lemma \ref{lem:warm-up}.

\section{Periodic walks with one nontrivial coin}\label{sec:one-nontrivial-coin}
Consider the $p$-periodic string $c$ as in \eqref{def:c-p}, given as
\begin{equation}\label{def:c-1-sefect}
c:=\left((a,b), (1,0), (1,0),\ldots, (1,0)\right),\ (a,b)\in\mathbb S^3
\end{equation}
i.e., a periodic shift-coin walk that includes only one nontrivial coin. 
\begin{theorem}\label{thm:1-defect}
For $p> 2$, let $U_p(c)$ be the $p$-periodic shift-coin walk generated by the string $c$  in \eqref{def:c-1-sefect}, then $v(U_p)=|a|$.
\end{theorem}
\begin{proof}
First, we note that if $a=0$ then it is direct to see that $v(U_p)=0$, see e.g., \cite[Lemma 2.2]{ARJS26}. So in the following we assume that $a\neq 0$.

We prove the equality by proving the two inequalities separately.\\

\noindent \underline{$v(U_p)\leq|a|$:} Here we use the methods in Section \ref{sec:harmonic}.
By Theorem \ref{thm:simple-spec}(\ref{thm:spectrum:2}) and Remark \ref{rem:spectrum-real}, the spectrum of $\widehat U_p(c,\theta)$ is simple for all 
\begin{equation}
E(c)=\left\{\theta\in[0,2\pi):\ e^{-2i\theta}\neq \frac{a}{\overline{a}}\right\},
\end{equation}
 then the velocity is given by \eqref{velocity:sup-E(c)}. 

Let $\psi$ be an eigenvector of $\widehat U(c,\theta)$ corresponding to a simple eigenvalue $\lambda:=\lambda(c,\theta)=e^{i\omega(c,\theta)}$. Recall from \eqref{eq:C-psi} that
\begin{equation}
C_j\begin{bmatrix}\psi_j^+\\ \psi_j^-\end{bmatrix}=\lambda \begin{bmatrix}\psi_{j+1}^+\\ \psi_{j-1}^-\end{bmatrix}, \text{ for }j=1,\ldots, p.
\end{equation}
with the cyclic conventions
\begin{equation}
\psi_{p+1}^+=e^{-i\theta}\psi_1^+,
\qquad
\psi_0^-=e^{i\theta}\psi_p^-.
\end{equation}
Then, for $j=2,\ldots,p$, we have
\begin{equation}\label{j23...p}
\psi_j^+=\lambda \psi_{j+1}^+, \qquad \psi_j^-=\lambda \psi_{j-1}^-,
\end{equation}
Since $|\lambda|=1$, it follows directly from \eqref{j23...p} that
\begin{equation}
|\psi_1^+|=|\psi_2^+|=\cdots=|\psi_p^+| \quad\text{ and }\quad |\psi_1^-|=|\psi_2^-|=\cdots=|\psi_p^-|.
\end{equation}
Consequently,
\begin{equation}\label{eq:1p-psi2}
|\psi_j^+|^2+|\psi_j^-|^2=\frac{1}{p}\|\psi\|^2, \qquad j=1,\ldots,p.
\end{equation}

We now apply Lemma \ref{lem:HF} 
\begin{equation}
\partial_\theta\omega(c,\theta)= \frac{K^\psi}{\|\psi\|^2},\quad \text{ with }\quad K^\psi:=|\psi_1^+|^2-|\psi_p^-|^2
\end{equation}
and, by the steps from  \eqref{K:Inner} to \eqref{Hf:11} with $j=1$, we obtain
\begin{equation}
|K^\psi| \leq |a|\left(|\psi_1^+|^2+|\psi_1^-|^2\right)\overset{\eqref{eq:1p-psi2}}{=}\frac{|a|}{p}\|\psi\|^2.
\end{equation}
Thus, we have
\begin{equation}
|\partial_\theta\omega(c,\theta)|\leq \frac{|a|}{p} \quad \Longrightarrow\quad v_{\mathrm{blk}}(U_p(c))\leq \frac{|a|}{p}.
\end{equation}
Equivalently, by \eqref{eq:v-vp}
\begin{equation}\label{defect:<=}
v(U_p(c))\leq |a|.\\[0.5cm]
\end{equation}

\noindent \underline{$v(U_p)\geq |a|$:}
Assume that the nontrivial coin sits on the sites $p\mathbb Z$, and that
\begin{equation}
C_{kp}=
\begin{bmatrix}
a & b\\
-\overline b & \overline a
\end{bmatrix},
\qquad |a|^2+|b|^2=1, \quad \text{ and }C_n=\idty_2, \qquad n\notin p\mathbb Z.
\end{equation}
Let $\mathcal H_p$ be the closed subspace of $\mathcal H$ consisting of
states supported on positions $p\mathbb Z$. Equivalently,
\begin{equation}
\mathcal H_p:=\ell^2(p\mathbb Z)\otimes \C^2\subset
=\mathcal H =\ell^2(\mathbb Z)\otimes \C^2.
\end{equation}

We first observe that $\mathcal H_p$ is invariant under $U^p$. Indeed, a
direct calculation shows that, for every $k\in\mathbb Z$,
\begin{align}\label{eq:Up-basis}
U^p \delta_{kp}^+
&=
a \delta_{(k+1)p}^+
-\overline b\, \delta_{(k-1)p}^-,
\notag\\
U^p \delta_{kp}^-
&=
b \delta_{(k+1)p}^+
+\overline a\, \delta_{(k-1)p}^-.
\end{align}
Thus, (since $U^p$ is bounded) $U^p\mathcal H_p\subseteq \mathcal H_p$.

Define the unitary
\begin{equation}\label{def:J}
J:\mathcal H\longrightarrow \mathcal H_p \ \ \text{ by }\ \ J =\sum_{k\in\mathbb Z} |kp\rangle\langle k|\otimes \idty_2.
\end{equation}
Let $\widetilde U$ be the homogeneous shift-coin walk on $\mathcal H$ with
constant coin $\widetilde C=\begin{bmatrix}a & b\\ -\overline b & \overline a \end{bmatrix}$.
A direct calculation shows (using \eqref{eq:Up-basis}, \eqref{def:J}, and the definition \eqref{def:Q} of the position operator $Q$)
\begin{equation}\label{J*-J}
J^*U^pJ=\widetilde U \quad \text{ and }\quad J^*QJ=pQ.
\end{equation}
Moreover, for every $\phi\in\mathcal H$,
\begin{equation}
\|QJ\phi\|^2=\sum_{k\in\mathbb Z}|pk|^2\|\phi_k\|_{\C^2}^2=p^2\|Q\phi\|^2.
\end{equation}
Therefore
$\phi\in D(Q)$ if and only if  $J\phi\in D(Q)\cap\mathcal H_p$,
and hence
\begin{equation}\label{J-psi-phi}
J D(Q)=D(Q)\cap\mathcal H_p.
\end{equation}
Restrict the supremum to states in $D(Q)\cap\mathcal H_p$, and
take the subsequence $t=pm$, we obtain
\begin{equation}
v(U_p(c)) \geq \sup_{\substack{\psi\in D(Q)\cap\mathcal H_p\\ \|\psi\|=1}} \limsup_{m\to\infty} \frac{1}{pm}\|Q (U^{p})^m\psi\|.
\end{equation}
Note that because $\psi\in\mathcal H_p$ then $Q (U^p)\psi \in \mathcal H_p$.

Now take any $\psi\in D(Q)\cap\mathcal H_p$ with $\|\psi\|=1$. Then we have,
\begin{equation}
\frac{1}{pm}\|Q U^{pm}\psi\|=\frac{1}{pm}\|(J^*QJ) ( J^* U^p J)^m J^*\psi\| \overset{\eqref{J*-J}}{=} \frac{1}{m}\|Q \widetilde U^mJ^*\psi\| 
\end{equation}
Consequently, set $\phi:=J^* \psi$ and use  \eqref{J-psi-phi} to see that
\begin{equation}\label{defect:>=}
v(U_p(c))\geq \sup_{\substack{\phi\in D(Q)\\ \|\phi\|=1}} \limsup_{m\to\infty} \frac{1}{m} \|Q\widetilde U^m\phi\|=v(\widetilde U)=|a|.
\end{equation}
In the last step we used the fact that the walk $\widetilde U$ is homogeneous with transmission coefficient
$|a|$. Therefore, by the standard results
$v(\widetilde U)=|a|$, see, e.g., \cite{ARCSW25}.

The inequalities \eqref{defect:<=} and \eqref{defect:>=} show 
\begin{equation}
v(U_p(c))=|a|.
\end{equation}
\end{proof}
\section{A general lower bound for the velocity}\label{sec:Lower-bound}
Consider the general periodic shift-coin walk $U_p$ determined by the $p$-periodic string 
\begin{equation}\label{lower-bound-c}
c=\left((a_1,b_1),\ldots (a_p,b_p)\right) \text{ where }(a_j,b_j)\in\mathbb S^3 \text{ and }a_j\neq 0 \text{ for all }j.
\end{equation}

The following theorem gives a general lower bound for the velocity of $U_p(c)$. We recall that if $a_j=0$ for some site $j$, then $v(U_p)=0$. 
\begin{theorem}\label{thm:v-lower}
For the $p$-periodic shift-coin $U_p(c)$ as above, the velocity has the following general lower bound.
\begin{equation}\label{eq:v-lower}
v(U_p)\geq pL(c), \quad \text{ where }L(c):=\left(\prod_{j=1}^p \frac{|a_j|}{1+|b_j|} \right)\left(\sum_{j=1}^p\frac{1}{1+|b_j|}\right)^{-1}.
\end{equation}
\end{theorem}


\begin{remark}\label{rem:Lower-bound}
Theorem \ref{thm:v-lower} also explains the scale of the velocity of quantum walks in electric fields in
\cite{ARCSW25,ARS2023-CMP}. For example, a homogeneous walk $U((a,b))$ with constant coin
$(a,b)$ and a specific field of odd period $p$ has velocity exactly $|a|^p$. More precisely, define the unitary electric field 
\begin{equation}
F_p=\sum_{j=1}^p \alpha_j |j\rangle\langle j|\otimes \idty_2, \qquad \alpha_j:=e^{\frac{2\pi i j}{p}} .
\end{equation}
The electric shift-coin walk is then given by $U_p:=F_p U((a,b))$. By \cite[Lemma 3.2]{ARJS26},
\begin{equation}
v(U_p)=v\left(F_p S C((a,b))\right)=v\left(S C((a,b))F_p\right)=v(U_p(c)),
\end{equation}
where $c$ is the $p$-periodic string $c=((a_1,b_1), \ldots, (a_p,b_p))$ with $(a_j,b_j)=\alpha_j(a,b)$ for $j=1,\ldots,p$. That is, the field phases are absorbed into the coins, and the periodic electric walk becomes a $p$-periodic walk generated by $c$, with
$|a_1|=\cdots=|a_p|=|a|$ and $|b_1|=\cdots=|b_p|=|b|$. In this case
Theorem \ref{thm:v-lower} gives
\begin{equation}
v(U_p) \overset{(\text{odd }p)}{=}|a|^p \geq  pL(c)=\frac{|a|^p}{(1+|b|)^{p-1}}\geq 2 \left(\frac{|a|}{2}\right)^p.
\end{equation}
Thus, in the odd-period case, the lower bound captures the exponential suppression in the period. Here we recall, see \cite{ARCSW25}
\begin{equation}
v(U_p)=\begin{cases}
|a|^p &  p \text{  is odd}\\
|a|^{p/2} & p \text{ is even}.
\end{cases}
\end{equation}
We note that a similar argument applies to the split-step walk (and hence for the corresponding generalized extended CMV setting), where the corresponding field can be absorbed into the second coin.
\end{remark}
\begin{proof}
The proof depends on the transfer matrices argument presented in Section \ref{sec:spectrum-simple-1R}. Recall that the transfer matrices associated with the Floquet matrix $\widehat U_p(c,\theta)$ (given in \eqref{def:Up-hat}) are given as
\begin{equation}
T_j(\lambda)=\frac{1}{\overline{a_j}}\begin{bmatrix} \lambda^{-1} & b_j\\ \overline{b_j} & \lambda\end{bmatrix} \text{ for }j=1,\ldots,p.
\end{equation}
and the monodromy matrix
\begin{equation}
M(\lambda)= T_p(\lambda) T_{p-1}(\lambda)\ldots T_1(\lambda), \text{ and observe that }\det M(\lambda)=\Delta(c):=\prod_{j=1}^p\frac{a_j}{\overline{a_j}}.
\end{equation}
Recall also \eqref{eigen-U-M}, that
\begin{equation}
\lambda\in\mathrm{spec}(\widehat U_p(c,\theta))\quad \iff \quad e^{-i\theta}\in\mathrm{spec}(M(\lambda))
\end{equation}
Thus, for a simple eigenvalue $\lambda(c,\theta)=e^{i\omega(c,\theta)}$ of $\widehat U_p(c,\theta)$, the matrix $M(e^{i\omega(c,\theta)})$ has one eigenvalue $e^{-i\theta}$. Since $\det M(e^{i\omega(c,\theta)})=\Delta(c)$, its second eigenvalue is $\Delta(c)e^{i\theta}$. Hence
\begin{equation}
\mathrm{Tr}\ M(e^{i\omega(c,\theta)})=e^{-i\theta}+\Delta(c)e^{i\theta}.
\end{equation}
Taking the derivative with respect to $\theta$, we obtain
\begin{equation}
\partial_\theta\omega(c,\theta)=\frac{i(\Delta(c)e^{i\theta}-e^{-i\theta})}{\frac{d}{d\omega}\mathrm{Tr}\ M(e^{i\omega(c,\theta)})}
\end{equation}

By Theorem \ref{thm:simple-spec}(\ref{thm:spectrum:1}), all the eigenvalues
$\{\lambda_j=e^{i\omega_j(c,\theta)}\}_{j=1}^{2p}$ of $\widehat U_p(c,\theta)$
are simple for $\theta\in E(c)$, where
\begin{equation}
E(c)=\{\theta\in[0,2\pi):\ e^{-2i\theta}\neq \Delta(c)\}.
\end{equation}
Let $\theta_0$ be the angle such that 
\begin{equation}\label{lower:2}
|\Delta(c)e^{i\theta_0}-e^{-i\theta_0}|=2,
\end{equation}
 and observe that $\theta_0\in E(c)$.

Therefore, the block velocity is given as, see \eqref{velocity:sup-E(c)}
\begin{equation}
v_{\mathrm{blk}}^{(p)}(U_p)=\sup_{\substack{\theta\in E(c)\\ j=1,\ldots,2p}}|\partial_\theta\omega_j(c,\theta)| \geq \sup_{j}|\partial_\theta\omega_j(c,\theta_0)|=\frac{2}{\displaystyle\left| \left.\frac{d}{d\omega}\right|_{\omega=\omega_j(c,\theta_0)}\mathrm{Tr}\ M(e^{i\omega(c,\theta_0)})\right|}
\end{equation}
We claim now that
\begin{equation}\label{eq:transfer:bound}
\left| \left.\frac{d}{d\omega}\right|_{\omega=\omega_j(c,\theta_0)} \mathrm{Tr}\ M(e^{i\omega})\right| \leq 2\left(\prod_{\ell=1}^p\frac{1+|b_\ell|}{|a_\ell|}\right) \left(\sum_{k=1}^p \frac{1}{1+|b_k|}\right).
\end{equation}
which gives directly the desired bound \eqref{eq:v-lower}.

It remains to prove \eqref{eq:transfer:bound}.

Observe that, (we use $\omega_j:=\omega_j(c,\theta_0)$)
\begin{equation}
\left. \frac{d}{d\omega}\right|_{\omega=\omega_j}\mathrm{Tr}\ M(e^{i\omega})=\mathrm{Tr}\ \left(\sum_{k=1}^p T_p(e^{i\omega_j})\ldots T_{k+1}(e^{i\omega_j})\left(\left.\frac{d}{d\omega_j}\right|_{\omega=\omega_j} T_k(e^{i\omega})\right) T_{k-1}(e^{i\omega_j})\ldots T_1(e^{i\omega_j})\right)
\end{equation}
Then use $|\mathrm{Tr}\ A|\leq 2\|A\|$ for any $2\times 2$ matrix $A$, to find
\begin{align}
\left|  \left.\frac{d}{d\omega_j}\right|_{\omega=\omega_j}\mathrm{Tr}\ M(e^{i\omega_j}) \right| \leq 2 \sum_{k=1}^p \prod_{\substack{\ell=1\\ \ell\neq k}}^p\left\|T_\ell(e^{i\omega_j})\right\|\ \left\|\left.\frac{d}{d\omega_j}\right|_{\omega=\omega_j} T_k(e^{i\omega})\right\|.
\end{align}

Use
\begin{equation}
\left\|T_\ell(e^{i\omega})\right\|=\frac{1+|b_\ell|}{|a_\ell|}
\qquad \text{and} \qquad
\left\|\frac{d}{d\omega} T_k(e^{i\omega})\right\|=\frac{1}{|a_k|}.
\end{equation}
to obtain the bound \eqref{eq:transfer:bound}.

\end{proof}

\vspace{0.5cm}
\noindent\textbf{\large Data Availability}\\
No data are associated with this article.\\

\noindent\textbf{\large Conflict of Interest}\\
The authors have no conflict of interest to declare.


\bibliographystyle{abbrvArXiv}
\bibliography{References.bib}
\end{document}